\documentclass{article}

\usepackage{ifthen}
\newcommand{\AlgorithmStyle}{1}
\newcommand{\IfUseAlgorithmIIeStyle}[2]{\ifthenelse{\equal{\AlgorithmStyle}{1}}{#1}{#2}}

\usepackage[ruled]{algorithm2e}
\renewcommand{\algorithmcfname}{Algorithm}

\usepackage[pdftex,paper=a4paper,left=3cm,right=3cm,top=4cm,bottom=4cm]{geometry}
\usepackage{amsmath, amssymb, amsthm}
\usepackage{mathrsfs}
\usepackage{graphicx}
\usepackage{algorithmic}
\usepackage{xspace}
\usepackage{subfig}
\usepackage{tabularx}

\allowdisplaybreaks[4]

\newcommand{\NN}{\mathbb{N}}

\newcommand{\RR}{\mathbb{R}}
\renewcommand{\S}{\mathcal{S}}
\newcommand{\K}{\mathcal{K}}
\newcommand{\R}{\mathcal{R}}
\newcommand{\C}{\mathcal{C}}
\newcommand{\B}{\mathbb{B}}
\newcommand{\CS}{\mathscr{C}}

\renewcommand{\d}{\mathrm{d}}
\newcommand{\D}{\mathrm{D}}
\newcommand{\T}{\mathrm{T}}
\newcommand{\OK}{\text{OK}}
\newcommand{\OKZ}{\text{OK}_0}
\newcommand{\PO}{\text{PO}}
\newcommand{\EQ}{\mathrm{eq}}
\newcommand{\NEQ}{\mathrm{neq}}
\newcommand{\VOL}{\mathrm{vol}}
\newcommand{\tWitness}{$\mathtt{Witness}$\xspace}
\newcommand{\Witness}[1]{$\mathtt{Witness(}#1\mathtt{)}$\xspace}
\newcommand{\tWitnessZ}{$\mathtt{Witness_0}$\xspace}
\newcommand{\WitnessZ}[1]{$\mathtt{Witness_0(}#1\mathtt{)}$\xspace}
\newcommand{\bWitnessZ}[1]{$\mathtt{Witness_0\big(}#1\mathtt{\big)}$\xspace}
\newcommand{\SET}[1]{\left\{#1\right\}}
\newcommand{\bSET}[1]{\big\{#1\big\}}
\newcommand{\tSET}[1]{\{#1\}}
\newcommand{\Ex}[2][]{\text{\rm\bf E}_{#1}\hspace{-0.03cm}\left[#2\right]}
\renewcommand{\Pr}[2][]{\text{\rm\bf Pr}_{#1}\hspace{-0.03cm}\left[#2\right]}
\newcommand{\DOT}{\,.}
\newcommand{\COMMA}{\,,}
\newcommand{\WHERE}{\,\colon\,}
\newcommand{\DEF}{\mathop{:=}}
\newcommand{\FED}{\mathop{=:}}
\newcommand{\OTHERWISE}{\text{otherwise}}
\newcommand{\tIF}{\text{if}\ }

\newcommand{\poly}{\text{poly}}
\newcommand{\ceil}[1]{\left\lceil #1 \right\rceil}
\newcommand{\floor}[1]{\left\lfloor #1 \right\rfloor}
\newcommand{\e}{\varepsilon}
\newcommand{\ID}[1]{\mathbb{I}_{#1}}
\newcommand{\ZERO}[2]{\mathbb{O}_{#1\times#2}}
\newcommand{\NULL}[1][]{\mathbb{O}_{#1}}
\renewcommand{\th}{^\text{th}}

\newcommand{\lmapsto}{\mathrel{\reflectbox{$\mapsto$}}}

\providecommand{\qedhere}{\tag*{\qed}}

\DeclareMathOperator*{\argmin}{arg\,min}

\newtheorem{theorem}{Theorem}
\newtheorem{lemma}[theorem]{Lemma}
\newtheorem{claim}{Claim}
\newtheorem{corollary}[theorem]{Corollary}
\newtheorem{proposition}[theorem]{Proposition}

\theoremstyle{definition}
\newtheorem{definition}[theorem]{Definition}

\providecommand{\institute}[1]{\date{#1}}
\providecommand{\email}[1]{\texttt{#1}}

\begin{document}

\title{Improved Smoothed Analysis of Multiobjective Optimization\thanks{The paper appeared in a preliminary version in the proceedings of STOC~2012 and will appear in JACM.}}

\author{Tobias Brunsch \and Heiko R\"oglin}

\institute
{%
Department of Computer Science\\
University of Bonn, Germany\\
\email{\small\{brunsch,roeglin\}@cs.uni-bonn.de}
}%

\maketitle

\begin{abstract}
We present several new results about smoothed analysis of multiobjective optimization problems.
Motivated by the discrepancy between worst-case analysis and practical experience, this line of research
has gained a lot of attention in the last decade. We consider problems in which~$d$ linear and one arbitrary
objective function are to be optimized over a set $\S \subseteq \SET{ 0, 1 }^n$ of feasible solutions.
We improve the previously best known bound for the smoothed number of Pareto-optimal solutions
to $O(n^{2d} \phi^d)$, where~$\phi$ denotes the perturbation parameter.  
Additionally, we show that for any constant~$c$ the $c\th$ moment of the smoothed number of Pareto-optimal
solutions is bounded by $O((n^{2d} \phi^d)^c)$. This improves the previously best known bounds significantly.

Furthermore, we address the criticism that the perturbations in smoothed analysis destroy the zero-structure
of problems by showing that the smoothed number of Pareto-optimal solutions remains polynomially bounded even
for zero-preserving perturbations. This broadens the class of problems captured by smoothed analysis and it
has consequences for non-linear objective functions. One corollary of our result is that the smoothed number
of Pareto-optimal solutions is polynomially bounded for polynomial objective functions. Our results also extend
to integer optimization problems.
\end{abstract}

\section{Introduction}

In most real-life decision-making problems there is more than one objective to
be optimized. For example, when booking a train ticket, one wishes to minimize
the travel time, the fare, and the number of train changes. As different
objectives are often conflicting, usually no solution is simultaneously optimal
in all criteria and one has to make a trade-off between different objectives.
The most common way to filter out unreasonable trade-offs and to reduce the
number of solutions the decision maker has to choose from is to determine the set
of \emph{Pareto-optimal solutions}, where a solution is called Pareto-optimal if
no other solution is simultaneously better in all criteria.

Multiobjective optimization problems have been studied extensively in operations
research and theoretical computer science (see, e.g., \cite{Ehrgott05} for a
comprehensive survey). In particular, many algorithms for generating the set of Pareto-optimal
solutions for various optimization problems such as the (bounded) knapsack
problem~\cite{NemhauserU69,KlamrothW00}, the multiobjective shortest path
problem~\cite{CorleyM85,Hansen80,SkriverA00}, and the multiobjective network
flow problem~\cite{Ehrgott99,MustafaG98} have been proposed. Enumerating the set
of Pareto-optimal solutions is not only used as a preprocessing step to
eliminate unreasonable trade-offs, but often it is also used as an intermediate
step in algorithms for solving optimization problems. For example, the
Nemhauser--Ullmann algorithm~\cite{NemhauserU69} treats the single-criterion knapsack problem as
a bicriteria optimization problem in which a solution with small weight and
large profit is sought, and it generates the set of Pareto-optimal solutions,
ignoring the given capacity of the knapsack. After this set has been generated,
the algorithm returns the solution with the highest profit among all Pareto-optimal solutions
with weight not exceeding the knapsack capacity. This solution is optimal for the
given instance of the knapsack problem.

Generating the set of Pareto-optimal solutions (a.k.a.\ the \emph{Pareto set}) only makes sense if few solutions
are Pareto-optimal. Otherwise, it is too costly and it does not provide enough
guidance to the decision maker. While, in many applications, it has been
observed that the Pareto set is indeed usually small (see,
e.g., \cite{Mueller-HannemannW01} for an experimental study of the multiobjective
shortest path problem), one can, for almost every problem with more than one objective
function, find instances with an exponential number of Pareto-optimal
solutions (see, e.g., \cite{Ehrgott05}).

Motivated by the discrepancy between worst-case analysis and practical observations, 
\emph{smoothed analysis} of multiobjective optimization problems has gained a lot of attention
in the last decade. Smoothed analysis is a framework for judging the performance of
algorithms that has been proposed in 2001 by Spielman and Teng~\cite{SpielmanT04}
in order to explain why the simplex algorithm is efficient in practice even though it has an
exponential worst-case running time. 
In this framework, inputs are generated in two steps: first, an adversary chooses an
arbitrary instance, and then this instance is slightly perturbed at random. The
smoothed performance of an algorithm is defined to be the worst expected
performance the adversary can achieve. This model can be viewed as a less
pessimistic worst-case analysis, in which the randomness rules out pathological
worst-case instances that are rarely observed in practice but dominate the
worst-case analysis. If the smoothed running time of an algorithm is low
and inputs are subject to a small amount of random noise then it is unlikely to encounter
an instance on which the algorithm performs poorly. In practice, random noise
can stem from measurement errors, numerical imprecision or rounding errors. It
can also model arbitrary influences, which we cannot quantify exactly, but for
which there is also no reason to believe that they are adversarial.

After its invention in 2001, smoothed analysis has been successfully
applied in a variety of contexts, e.g., to explain the
practical success of local search methods, heuristics
for the knapsack problem, online algorithms, and clustering.
A recent survey by Spielman and Teng~\cite{SpielmanT09} summarizes some of these results.
One of the areas in which smoothed analysis has been applied extensively is
multiobjective optimization. In 2003 Beier and V\"ocking~\cite{BeierV04} initiated this line
of research by showing that the smoothed number of Pareto-optimal solutions is polynomially bounded
for all linear binary optimization problems with two objective functions. This was the first
rigorous explanation why heuristics for generating the set of Pareto-optimal solutions are successful
in practice despite their bad worst-case behavior.
In the last years, Beier and V\"ocking's original result has been improved and extended significantly
in a series of papers. A discussion of this work follows in the next section after the formal description
of the model.

\subsection{Model and Previous Work}

We consider a very general model of multiobjective optimization problems.
An instance of such a problem consists of $d+1$ objective functions $V^1, \ldots, V^{d+1}$
that are to be optimized over a set $\S \subseteq \SET{ 0, \ldots, \K }^n$ of feasible solutions
for some integer~$\K$. While the set~$\S$ and the last objective function $V^{d+1} \colon \S \to \RR$ can be
arbitrary, the first~$d$ objective functions have to be linear of the
form $V^t(x) = V^t_1 x_1 + \ldots + V^t_n x_n$ for $x = (x_1, \ldots, x_n) \in \S$ and $t \in \SET{ 1, \ldots, d }$.
We assume without loss of generality that all objectives are to be minimized and we call
a solution $x \in \S$ \emph{Pareto-optimal} if there is no solution $y \in \S$ which is at least as
good as~$x$ in all of the objectives and even better than~$x$ in at least one. We will introduce
this notion formally in Section~\ref{sec:notation}. The set of Pareto-optimal solutions is called
the \emph{Pareto set}. We are interested in the size of this set. As a convention, we count distinct
Pareto-optimal solutions that coincide in all objective values only once. Since we compare solutions
based on their objective values, there is no need to consider more than one
solution with exactly the same values.

If one is allowed to choose the set~$\S$, the objective function~$V^{d+1}$, and the coefficients of the
linear objective functions arbitrarily, then even for $d=1$, one can construct instances with an
exponential number of Pareto-optimal solutions. For this reason Beier and V\"ocking introduced
the model of \emph{$\phi$-smooth instances}~\cite{BeierV04}, in which an adversary can choose the set~$\S$
and the objective function~$V^{d+1}$ arbitrarily while he can only specify a probability density
function $f^t_i\colon[-1,1]\to[0,\phi]$ for each coefficient~$V^t_i$ according to which it is chosen
independently of the other coefficients. This model is more general than Spielman and Teng's original
two-step model in which the adversary first chooses coefficients which are afterwards subject to 
Gaussian perturbations. In $\phi$-smooth instances the adversary can additionally determine the type of noise.
He could, for example, specify for each coefficient
an interval of length $1/\phi$ from which it is chosen uniformly at random.
The parameter $\phi \ge 1/2$ can be seen as a measure
for the power of the adversary: the larger~$\phi$ the more precisely he can specify the coefficients
of the linear objective functions. The aforementioned example of uniform distributions in intervals of length $1/\phi$ shows that for $\phi\to\infty$
smoothed analysis becomes a worst-case analysis. 

The \emph{smoothed number of Pareto-optimal solutions} depends on the number~$n$ of integer variables, the maximum integer~$\K$, and the perturbation
parameter~$\phi$. It is defined to be the largest expected number of Pareto-optimal solutions the adversary can
achieve by any choice of $\S \subseteq \SET{ 0, \ldots, \K }^n$, $V^{d+1} \colon \S \to \RR$, and the densities $f^t_i \colon [-1,1] \to [0,\phi]$.
In the following we assume that the adversary has made arbitrary fixed choices for these entities.
Then we can associate with every matrix $V \in \RR^{d \times n}$ the number $\PO(V)$ of Pareto-optimal solutions in~$\S$
when the coefficients~$V^t_i$ of the~$d$ linear objective functions take the values given in~$V$.
Assuming that the adversary has made worst-case choices for~$\S$, $V^{d+1}$, and the densities~$f^t_i$,
the \emph{smoothed number of Pareto-optimal solutions} is the expected value $\Ex[V]{\PO(V)}$, where the coefficients in~$V$
are chosen according to the densities~$f^t_i$.
For $c\ge 1$, we call $\Ex[V]{\PO^c(V)}$ the \emph{$c$-th moment of the smoothed number of Pareto-optimal solutions}.
Here we assume that the adversary has made worst-case choices for~$\S$, $V^{d+1}$, and the densities~$f^t_i$
that maximize~$\Ex[V]{\PO^c(V)}$ (in general, these are different from the choices that maximize~$\Ex[V]{\PO(V)}$). 
   
Beier and V\"ocking~\cite{BeierV04} showed that for the binary bicriteria case (i.e., $\K=d=1$) 
the smoothed number of Pareto-optimal solutions is $O(n^4\phi)$ and $\Omega(n^2)$.
The upper bound was later simplified and improved by Beier et al.~\cite{BeierRV07}
to $O(n^2\phi)$. In his PhD thesis~\cite{BeierPhD}, Beier conjectured that the smoothed
number of Pareto-optimal solutions is polynomially bounded in~$n$ and~$\phi$ for $\K=1$ and every constant~$d$.
This conjecture was proven by R\"oglin and Teng~\cite{RoeglinT09}, who showed that for
binary solutions and for any fixed $d \geq 1$, the smoothed number of Pareto-optimal solutions is $O((n^2\phi)^{f(d)})$,
where the function~$f$ is roughly $f(d)=2^{d}d!$. They also proved that for any constant~$c$
the $c$-th moment of the smoothed number of Pareto-optimal solutions is bounded
by $O((n^2\phi)^{c\cdot f(d)})$. Moitra and O'Donnell~\cite{MoitraO11} improved the
bound for the smoothed number of Pareto-optimal solutions significantly to $O(n^{2d}\phi^{d(d+1)/2})$. 
However, it remained unclear how to improve the bound for the moments by their methods.
Recently a lower bound of $\Omega(n^{d-1.5}\phi^d)$ for the smoothed number of Pareto-optimal solutions was proven~\cite{BrunschGRR14}.

\subsection{Our Results}

In this article, we present several new results about smoothed analysis of multiobjective binary and integer optimization problems.
Besides general $\phi$-smooth instances, we additionally consider the special case of 
\emph{quasiconcave density functions}. This means that we assume that every coefficient~$V^t_i$ is chosen
independently according to its own density function $f^t_i \colon [-1,1] \to [0,\phi]$ with the additional requirement that
for every density~$f^t_i$ there is a value $x^t_i \in [-1,1]$ such that~$f^t_i$ is non-decreasing in the interval $[-1,x^t_i]$
and non-increasing in the interval $[x^t_i,1]$. We do not think that this is a severe restriction because all natural
perturbation models, like Gaussian or uniform perturbations, use quasiconcave density functions. Furthermore, quasiconcave
densities capture the essence of a perturbation: each coefficient~$V^t_i$ has an unperturbed value~$x^t_i$ and the probability
that the perturbed coefficient takes a value~$z$ becomes smaller with increasing distance $|z - x^t_i|$. We will call
these instances \emph{quasiconcave $\phi$-smooth instances} in the following.

Beier and V\"ocking originally only considered $\phi$-smooth instances for binary bicriteria optimization problems
(i.e., for the case $\K=d=1$). The above described canonical generalization of this model to binary multiobjective optimization
problems, on which R\"oglin and Teng's~\cite{RoeglinT09} and Moitra and O'Donnell's results~\cite{MoitraO11}
are based, appears to be very general and flexible at the first glance. However, one aspect
limits its applicability severely and makes it impossible to formulate certain multiobjective linear
optimization problems in this model. The weak point of the model is that it assumes that every binary variable~$x_i$
appears in every linear objective function as it is not possible to set some coefficients~$V^t_i$ deterministically to~$0$.

Already Spielman and Teng~\cite{SpielmanT04} and Beier and V\"ocking~\cite{BeierV06} observed that the zeros often
encode an essential part of the combinatorial structure of a problem and they suggested to analyze \emph{zero-preserving perturbations}
in which it is possible to either choose a density~$f^t_i$ according to which the coefficient~$V^t_i$ is chosen or to
set it deterministically to~$0$. Zero-preserving perturbations have been studied in~\cite{SankarST06} and~\cite{BeierV06} for analyzing
smoothed condition numbers of matrices and the smoothed complexity of binary optimization problems. 
For the smoothed number of Pareto-optimal solutions no upper bounds are known that are valid
for zero-preserving perturbations (except trivial worst-case bounds), and in particular the bounds  
proven in~\cite{RoeglinT09} and~\cite{MoitraO11} do not seem to generalize easily to zero-preserving perturbations.
In this article, we develop new techniques for analyzing the smoothed number of Pareto-optimal solutions that can also be used
for analyzing zero-preserving perturbations.

\begin{theorem}\label{thm:MainZeroPreserving}
For any $d \geq 1$, the smoothed number of Pareto-optimal solutions is $\K^{(d+1)^5} \cdot O(n^{d^3+d^2+d} \phi^d)$ 
for quasiconcave $\phi$-smooth instances with zero-preserving perturbations and $\K^{(d+1)^5} \cdot O((n\phi)^{d^3+d^2+d})$ for general $\phi$-smooth instances with zero-preserving perturbations.
\end{theorem}

Let us remark that the bounds stated in Theorem~\ref{thm:MainZeroPreserving} hold for any $\K \geq 1$ and not only for sufficiently large values of~$\K$. This is why the factor $\K^{(d+1)^5}$ is outside of the $O$-notation. The $O$-notation only refers to the parameters~$n$ and~$\phi$. For constant~$\K$ like in the binary case the factor $\K^{(d+1)^5}$ is a constant for fixed~$d$. In Section~\ref{subsec:ZeroPreserving} we will present some applications of zero-preserving perturbations.
We will see that they allow us not only to extend the smoothed analysis to linear multiobjective optimization problems
that are not captured by the previous model without zero-preserving perturbations, but that they also enable us
to bound the smoothed number of Pareto-optimal solutions in problems with \emph{non-linear objective functions}.
In particular, the number of Pareto-optimal solutions for multivariate polynomial objective functions can be bounded by
Theorem~\ref{thm:MainZeroPreserving}. We say that a $\phi$-smooth instance has polynomial objective functions if  
every objective function~$V^t$, $t \in \SET{ 1,\ldots, d }$, is the weighted sum of monomials, where the
adversary can specify a $\phi$-bounded density on $[-1, 1]$ for every weight according to which it is chosen. Denote the total number of monomials by~$m$ and let~$\Delta$ denote the maximum degree of the monomials. Then the following corollary holds.

\begin{corollary}
For any $d \geq 1$, the smoothed number of Pareto-optimal solutions is $\K^{(d+1)^5 \Delta} \cdot O(m^{d^3+d^2+d} \phi^d)$ 
for quasiconcave $\phi$-smooth instances with polynomial objective functions. 
For general $\phi$-smooth instances with polynomial objective functions
the smoothed number of Pareto-optimal solutions is $\K^{(d+1)^5 \Delta} \cdot O((m\phi)^{d^3+d^2+d})$.
\end{corollary}

In addition to zero-preserving perturbations we also study the standard model of $\phi$-smooth instances.
We present significantly improved bounds for the smoothed number of Pareto-optimal solutions and the moments, answering 
two questions posed by Moitra and O'Donnell~\cite{MoitraO11}.

\begin{theorem}\label{thm:MainFirstMoment}
For any $d \geq 1$, the smoothed number of Pareto-optimal solutions
is $\K^{2(d+1)^2} \cdot O(n^{2d} \phi^d)$ for quasiconcave $\phi$-smooth instances and
$\K^{2(d+1)^2} \cdot O(n^{2d} \phi^{d(d+1)})$ for general $\phi$-smooth instances.
\end{theorem}  

The bound of Theorem~\ref{thm:MainFirstMoment} for quasiconcave $\phi$-smooth instances improves the previously
best known bound of $O(n^{2d}\phi^{d(d+1)/2})$ in the binary case (which is, however, valid also for non-qua\-si\-con\-cave densities)
and it answers a question posed by Moitra and O'Donnell whether it is
possible to improve the factor of $\phi^{d(d+1)/2}$ in their bound~\cite{MoitraO11}.
Together with the recent lower bound of $\Omega(n^{d-1.5}\phi^d)$~\cite{BrunschGRR14}, which is also valid for
quasiconcave density functions, this shows that the exponents of both~$n$ and~$\phi$ are linear in~$d$.

\begin{theorem}\label{thm:HigherMoment}
For any $d \geq 1$ and any constant $c \in \NN$, the $c$-th moment of the smoothed number of Pareto-optimal solutions
is $\K^{(c+1)^2 (d+1)^2} \cdot O((n^{2d} \phi^d)^c)$ for quasiconcave $\phi$-smooth instances and $\K^{(c+1)^2 (d+1)^2} \cdot O((n^{2d} \phi^{d(d+1)})^c)$ for general $\phi$-smooth instances.
\end{theorem}  

This answers a question in~\cite{MoitraO11} whether it is possible to improve the bounds for the 
moments in~\cite{RoeglinT09} and it yields better concentration bounds for the smoothed number of Pareto-optimal solutions. 
Our results also have immediate consequences for the expected running times of various algorithms because  
most heuristics for generating the Pareto set of some problem (including the ones mentioned at the beginning of the introduction) 
have a running time that depends linearly or quadratically on the size of the Pareto set.

The straightforward extension of the Nemhauser-Ullmann algorithm~\cite{NemhauserU69} to the multiobjective knapsack problem
has, for example, a running time of~$\Theta(\sum_{i=0}^{n-1}|P_i|^2)$ on instances with~$n$ items where~$P_i$ denotes the
Pareto set of the instance that consists only of the first~$i$ items. (For~$d=1$ the running time can be made linear in~$|P_i|$ if the sets~$P_i$
are stored in sorted order.) Other examples are the extensions of the Bellman-Ford algorithm and the Floyd-Warshall algorithm
to multiobjective shortest path problems (see, e.g.,~\cite{EhrgottG02}) whose running times depend linearly (for $d=1$) or
quadratically (for~$d>1$) on the number of Pareto-optimal solutions in certain subproblems.
The improved bounds on the smoothed number of Pareto-optimal solutions and the second moment of this number
yield improved bounds on the smoothed running times of these and various other algorithms.

Note that our analysis also covers the general case when the set~$\S$ is an arbitrary subset of $\SET{ -\K, \ldots, \K }^n$.
In this case, consider the shifted set $\S' = \SET{ x + u \WHERE x \in \S } \subseteq \SET{ 0, \ldots, 2\K }$
for $u = (\K, \ldots, \K)$ and the functions $W^1, \ldots, W^{d+1} \colon \S' \to \RR$, defined as $W^t x = V^t x$
for $t = 1, \ldots, d$ and $W^{d+1} x = V^{d+1}(x - u)$. The Pareto set with respect to~$\S$ and
$\SET{ V^1, \ldots, V^{d+1} }$ and the Pareto set with respect to~$\S'$ and $\SET{ W^1, \ldots, W^{d+1} }$ are
identical except for a shift of $(V^1 u, \ldots, V^d u, 0)$ in the image space. Hence, the sizes of both sets
are equal. All aforementioned results can be applied for~$\S'$ and $\SET{ W^1, \ldots, W^{d+1} }$, so they also
hold for~$\S$ and $\SET{ V^1, \ldots, V^{d+1} }$ if one replaces~$\K$ by $2\K$.

\subsection{Applications of Zero-preserving Perturbations}\label{subsec:ZeroPreserving}

Let us first of all remark that we can assume that the adversarial objective~$V^{d+1}$ is injective. If not,
then let $v_1, \ldots, v_\ell$ be the values taken by~$V^{d+1}$ and let $\Delta = \min_{i \neq j} |v_i - v_j|$.
Now, consider an arbitrary injective function $\delta \colon \S \to [0, \Delta)$ and define the new adversarial
objective as $W^{d+1} x = V^{d+1} x + \delta(x)$. Obviously, this function is injective and it preserves
the order of the solutions in~$\S$. This means that if $V^{d+1} x < V^{d+1} y$ for $x, y \in \S$, then also
$W^{d+1} x < W^{d+1} y$. Let~$x$ be a Pareto optimum with respect to~$\S$ and $\bSET{ V^1, \ldots, V^{d+1} }$
and let $x_2, \ldots, x_m$, $m \geq 1$, be all the other solutions for which $V^k x_i = V^k x$ for all $k \in \SET{ 1, \ldots, d+1 }$.
These are all Pareto optima but, due to our convention, we only count them once. Without loss of generality let~$x$
be the solution that minimizes~$W^{d+1}$ among these solutions. Then~$x$ is also Pareto-optimal with
respect to~$\S$ and $\bSET{ V^1, \ldots, V^d, W^{d+1} }$.

Before we give some applications of zero-preserving perturbations let us remark that in the bicriteria case,
which was studied in~\cite{BeierV04}, zero-preserving
perturbations are not more powerful than other perturbations because they can be simulated by the
right choice of $\S \subseteq \SET{ 0, \ldots, \K }^n$ and the objective function $V^2 \colon \S \to \RR$.

Assume, for example, that the adversary has chosen~$\S$ and~$V^2$ and has decided that the first
coefficient~$V^1_1$ of the first objective function should be deterministically set to~$0$. Also
assume without loss of generality that~$V^2$ is injective. We can partition the set~$\S$ into classes of
solutions that agree in all components except for the first one.
This means that two solutions $x \in \S$ and $y \in \S$ belong to the same class if $x_i = y_i$ for all $i \in \SET{ 2, \ldots, n }$.
All solutions in the same class have the same value in the first objective~$V^1$ as they differ only in the
binary variable~$x_1$, whose coefficient has been set to~$0$. We construct a new set of solutions~$\S'$ that
contains for every class only the solution with smallest value in~$V^2$. One can verify that the number of
Pareto-optimal solutions is the same with respect to~$\S$ and with respect to~$\S'$ because all solutions
in $\S \setminus \S'$ are dominated by solutions in~$\S'$. Then we transform the set $\S' \subseteq \SET{ 0, \ldots, \K }^n$ into a set $\S'' \subseteq \SET{ 0, \ldots, \K }^{n-1}$ by dropping the first component of
every solution. Furthermore, we define a function $W^2 \colon \S'' \to \RR$ that assigns to every solution $x \in \S''$
the same value that~$V^2$ assigns to the corresponding solution in~$\S'$. One can verify that the Pareto set
with respect to~$\S'$ and~$V^2$ is identical with the Pareto set with respect to~$\S''$ and~$W^2$. The only difference is
that in the latter problem we have eliminated the coefficient that is deterministically set to~$0$.
Such an easy reduction of zero-preserving perturbations to other perturbations does not seem to be possible
for $d \ge 2$ anymore.

\subsubsection*{Path Trading}

Berger et al.~\cite{BergerRZ11} study a model for routing in networks. In their model there is a
graph $G = (V, E)$ whose vertex set~$V$ is partitioned into mutually disjoint
sets $V_1, \ldots, V_k$. We can think of~$G$ as the Internet graph whose
vertices are owned and controlled by~$k$ different autonomous systems (ASs). We
denote by $E_i \subseteq E$ the set of edges inside~$V_i$.
The graph~$G$ is undirected, and each edge $e \in E$ has a length
$\ell_e \in \RR_{\geq 0}$. The traffic is modeled by a set of requests, where each
request is characterized by its source node $s \in V$ and its target node $t \in
V$. The Border Gateway Protocol (BGP) determines for each request $(s, t)$ the order in which
it has to be routed through the ASs. We say that a path~$P$ from~$s$ to~$t$ is valid if it
connects~$s$ to~$t$ and visits the ASs in the order specified by the BGP protocol. This means that the first AS
has to choose a path~$P_1$ inside~$V_1$ from~$s$ to some node in~$V_1$ that is connected to some
node $v_2 \in V_2$. Then the second AS has to choose a path~$P_2$ inside~$V_2$ from~$v_2$ to some
node in~$V_2$ that is connected to some node $v_3 \in V_3$ and so on.
For simplicity, the costs of routing a packet between two ASs are assumed to be~$0$, whereas
AS~$i$ incurs costs of $\sum_{e \in P_i} \ell_e$ for routing the packet inside~$V_i$ along path~$P_i$.  
In the common \emph{hot-potato routing}, every AS is only interested in minimizing its own costs
for each request.
To model this, there are~$k$ objective functions that map each valid path~$P$ to a cost 
vector $(C_1(P), \ldots, C_k(P))$, where
\[
   C_i(P) = \sum_{e \in P \cap E_i} \ell_e \quad \text{for $i \in \SET{ 1, \ldots, k }$} \DOT
\]

\begin{figure}
  \begin{center}
    \includegraphics[width=0.5\textwidth]{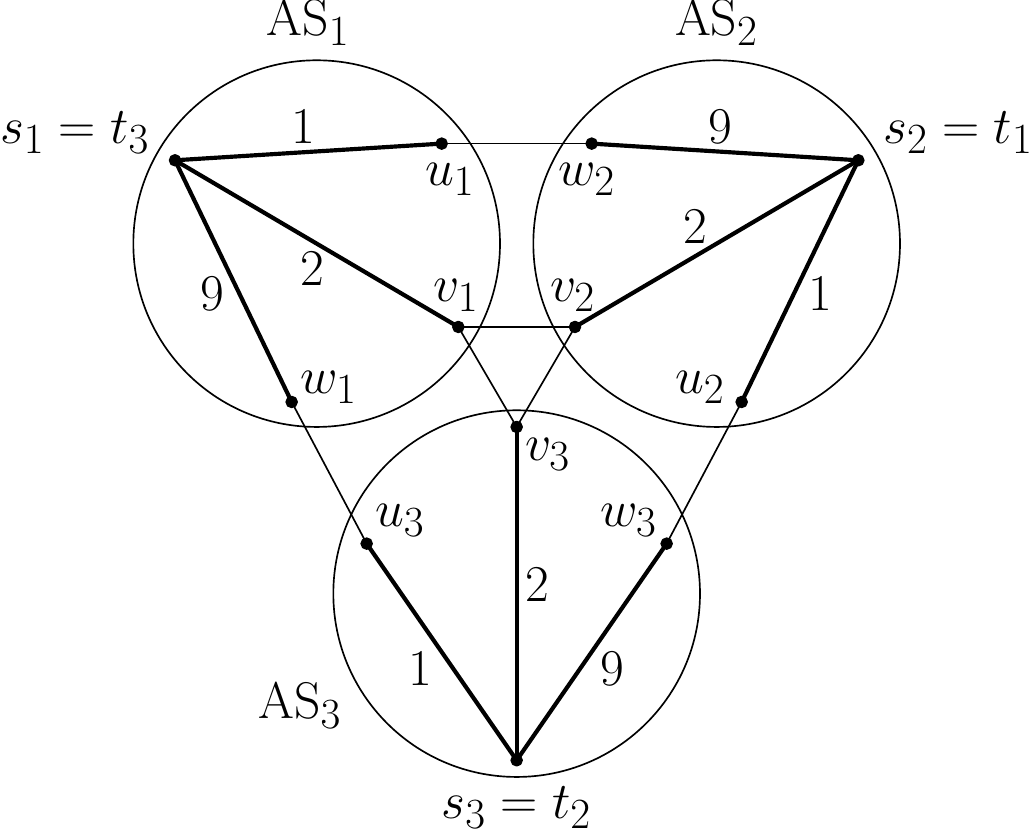}
    \caption{A network graph with three autonomous systems}
    \label{fig:path trading}
  \end{center}
\end{figure}

In~\cite{BergerRZ11} the problem of \emph{path trading} is considered. If there is only
one request, then no AS has an incentive to deviate from the hot-potato strategy.
The problem becomes more interesting if there are multiple requests that have to
be satisfied. Consider, for example, the three ASs depicted in Figure~\ref{fig:path trading}
and assume that there are three requests $(s_1, t_1)$, $(s_2, t_2)$, and $(s_3, t_3)$.
Moreover, assume that the BGP specifies that all requests from $s \in V_i$ to $t \in V_j$
shall be routed directly from AS~$i$ to AS~$j$. If all ASs follow the hot-potato
strategy, then they decide for the routes $(s_1, u_1, w_2, t_1)$, $(s_2, u_2, w_3, t_2)$,
and $(s_3, u_3, w_1, t_3)$. Each AS~$i$ incurs costs of~$1$ for the request
$(s_i, t_i)$ and costs of~$9$ for the request $(s_j, t_j)$ for which $t_j \in V_i$.

Now assume that AS~$i$ routes request $(s_i, t_i)$ from~$s_i$ to~$v_i$. Then it
incurs costs of~$2$ (instead of~$1$) for this route, which is worse than if it had chosen the
hot-potato route. However, if all ASs agree on this new strategy, then each AS~$i$
only incurs costs of~$2$ (instead of~$9$) for the request $(s_j, t_j)$ for which
$t_j \in V_i$. Hence, the total costs of each AS for satisfying the three requests $(s_i, t_i)$
is~$4$ instead of~$10$.

The path trading problem asks whether there exist routes for given requests $(s_i, t_i)$ such
that the total costs of each involved AS is less than or equal to the total costs it would
incur if all would follow the hot-potato strategy. Such routes are called \emph{feasible
path trades}.

Consider~$\K$ requests $(s_1, t_1), \ldots, (s_\K, t_\K)$ and $s_i$-$t_i$-paths $P_1, \ldots, P_\K$
that comply with the BGP. For an edge $e \in E$ let $x_e \in \SET{ 0, \ldots, \K }$ be the
number of paths $P_1, \ldots, P_\K$ that contain~$e$. We can encode the routes $P_1, \ldots, P_\K$ by an
integer vector $x \in \SET{ 0, \ldots, \K }^{|E|}$ consisting of the values~$x_e$. Let~$\S$ denote
the set of encodings of all valid routes $P_1, \ldots, P_\K$. The question whether there is
a feasible path trade for the requests $(s_i, t_i)$ reduces to the question whether the
vector~$x^\star$ that encodes the hot-potato routes $P^\star_1, \ldots, P^\star_\K$ is not
Pareto-optimal with respect to~$\S$ and $\SET{ C_1, \ldots, C_k }$, where the objectives $C_i \colon \S \to \RR$,
\[
  C_i(x) = \sum_{e \in E_i} \ell_e x_e \COMMA
\]
describe the total costs of AS~$i$ for the routes encoded by~$x$. As the Pareto set can be
exponentially large in the worst case, Berger et al.~\cite{BergerRZ11} proposed
to study $\phi$-smooth instances in which an adversary chooses the graph~$G$ and a density
$f_e \colon [0,1] \to [0,\phi]$ for every edge length~$\ell_e$ according to which it is chosen.
It seems as if we could easily apply the
results in~\cite{RoeglinT09} and~\cite{MoitraO11} to bound the smoothed number of Pareto-optimal paths because
all objective functions~$C_i$ are linear in the binary variables~$x_e$, $e \in E$. However,
note that different objective functions contain different variables~$x_e$ because the
coefficients of all~$x_e$ with~$e\notin E_i$ are set to~$0$ in~$C_i$. This is an important combinatorial
property of the path trading problem that has to be obeyed.
In the model in~\cite{RoeglinT09} and~\cite{MoitraO11} it is not possible
to set coefficients deterministically to~$0$. In their model, an AS would, with a probability of~$1$, incur positive costs for all edges and not only for its own edges that are used, which does not resemble the structure of the problem.
Theorem~\ref{thm:MainZeroPreserving}, which allows zero-preserving perturbations, yields immediately the following result. 

\begin{corollary}
The smoothed number of Pareto-optimal valid paths is polynomially bounded in~$|E|$, $\phi$, and~$\K$ for any constant~$k$.
\end{corollary}

\subsubsection*{Non-linear Objective Functions}

Even though we assumed above that the objective functions $V^1, \ldots, V^d$ are linear, we can also extend the smoothed analysis
to non-linear objective functions. We consider first the bicriteria case $d=1$. As above, we assume that the adversary has chosen
an arbitrary set~$\S$ of feasible solutions and an arbitrary injective objective function $V^2 \colon \S \to \RR$. In addition
to that the adversary can choose~$m_1$ arbitrary functions $I^1_i \colon \S \to \SET{ 0, \ldots, \K }$, $i \in \SET{ 1, \ldots, m_1 }$. The objective
function $V^1 \colon \S \to \RR$ is defined to be a weighted sum of the functions~$I_i^1$:
\[
   V^1(x) = \sum_{i=1}^{m_1} w^1_i I^1_i(x) \COMMA
\]
where each weight~$w^1_i$ is randomly chosen according to a density $f^1_i \colon [-1, 1] \to [0, \phi]$ given by the adversary.
There is a wide variety of functions $V^1(x)$ that can be expressed in this way. We can, for example, express every polynomial
if we let $I^1_1, \ldots, I^1_{m_1}$ be its monomials. Note that the value~$\K$ then depends on the set~$\S$ and the maximum degree of the monomials.

We can linearize the problem by introducing a binary variable for every function~$I^1_i$. Using the function $\pi \colon \S \to \SET{ 0, \ldots, \K }^{m_1}$, defined by $\pi(x) = (I^1_1(x), \ldots, I^1_{m_1}(x))$, the set of feasible solutions 
becomes $\S' = \SET{ \pi(x) \WHERE x \in \S } \subseteq \SET{ 0, \ldots, \K }^{m_1}$. For this set of feasible solutions we define $W^1 \colon \S' \to \RR$ and $W^2 \colon \S' \to \RR$ as follows:
\[
   W^1(y) = \sum_{i=1}^{m_1} w^1_i y_i
   \quad \text{and} \quad
   W^2(y) = \min \SET{ V^2(x) \WHERE x \in \S\ \text{and}\ \pi(x) = y } \DOT
\]
The problem defined by~$\S$, $V^1$, and~$V^2$ and the problem defined by~$\S'$, $W^1$, and~$W^2$
are equivalent and have the same number of Pareto-optimal solutions. The latter problem is linear and hence we can apply
the result by Beier et al.~\cite{BeierRV07}, which yields that the
smoothed number of Pareto-optimal solutions is bounded by $\poly(\K) \cdot O(m_1^2 \phi)$.
This shows in particular that the smoothed number of Pareto-optimal solutions is polynomially bounded 
in the number of monomials, the maximum integer in the monomials' ranges, and the density parameter for every polynomial objective function~$V^1$.

We can easily extend these considerations to multiobjective problems with $d \geq 2$.
For these problems the adversary chooses an arbitrary set~$\S$, numbers $m_1, \ldots, m_d \in \NN$, and
an arbitrary injective objective function $V^{d+1} \colon \S \to \RR$. 
In addition to that he chooses arbitrary functions $I^t_i \colon \S \to \SET{ 0, \ldots, \K }$ for $t \in \SET{ 1, \ldots, d }$ and $i \in \SET{ 1, \ldots, m_t }$.
Every objective function $V^t \colon \S \to \RR$ is a weighted sum
\[
   V^t(x) = \sum_{i=1}^{m_t} w^t_i I^t_i(x)
\]
of the functions~$I^t_i$, where each weight~$w^t_i$ is randomly chosen according to a density $f^t_i \colon [-1, 1] \to [0, \phi]$
chosen by the adversary. Similar to the bicriteria case, also this problem can be linearized.
However, the previous results about the smoothed number of Pareto-optimal solutions can only be applied if
every objective function~$V^t$ is composed of exactly the same functions~$I^t_i$. Theorem~\ref{thm:MainZeroPreserving}
implies that the smoothed number of Pareto-optimal solutions is polynomially bounded in $\sum m_i$, $\K$, and~$\phi$, for any choice of the~$I^t_i$.

\subsection*{Outline}

After introducing some notation in the next section, we present
an outline of our approach and our methods in Section~\ref{sec:Approach}. In our
analysis we will frequently draw upon fundamental properties of Pareto-optimal
solutions. These are stated and proven in Section~\ref{sec:Pareto properties}.
In Section~\ref{sec:ZeroDestroying} we prove Theorems~\ref{thm:MainFirstMoment}
and~\ref{thm:HigherMoment}. In Section~\ref{sec:ZeroPreserving} we consider
zero-preserving perturbations and prove Theorem~\ref{thm:MainZeroPreserving}.
We conclude the article with some open questions.


\section{Notation}
\label{sec:notation}

For the sake of simplicity we write $V^t x$ instead of $V^t(x)$, even for the
adversarial objective~$V^{d+1}$. With $V^{k_1 \ldots k_t}x$ we refer to the
vector $(V^{k_1}x, \ldots, V^{k_t}x)$. In our analysis, we will shift the
solutions $x \in \S$ by a certain vector $u \in \SET{ 0, \ldots, \K }^n$ and consider the
values $V^t \cdot (x-u)$. For the linear objectives we mean the value $V^tx -
V^tu$, where $V^t u$ is well-defined even for a shift vector $u \in \SET{ 0, \ldots, \K
}^n \setminus \S$. For the adversarial objective, however, we define $V^{d+1}
\cdot (x - u) := V^{d+1} x$. It should not be confused with $V^{d+1} y$ for $y =
x-u$. Note that for Pareto-optimality only the ordering of the solutions
with respect to $V^{d+1}$ and not the values $V^{d+1} x$ themselves are of interest.
By the definition of $V^{d+1} \cdot (x - u)$, the ordering of the vectors $x - u$,
$x \in \S$, equals the ordering of the vectors $x \in \S$ when considering $V^{d+1}$.

In the whole article let $\e>0$ be an arbitrary real for which $1/\e$ is integral.
Our analyses are valid for all such choices of~$\e$, but to obtain our results we will
consider the limit $\e \to 0$. Thus, think of~$\e$ as a very small real. 
Let $b = (b_1, \ldots, b_d) \in \RR^d$ be a vector such that~$b_k$ is an
integral multiple of~$\e$ for all~$k$. We will call the set $B = \SET{ (y_1, \ldots,
y_d) \in \RR^d \WHERE y_k \in (b_k, b_k + \e] \ \mbox{for all} \ k }$ an
\emph{$\e$-box} and~$b$ \emph{the corner} of~$B$. For a vector $x \in \SET{ -\K,
\ldots, \K }^n$ the expression $B_V(x)$ denotes the unique $\e$-box~$B$ for which $V^{1
\ldots d} x \in B$. We call~$B$ the $\e$-box of~$x$ and say that~$x$ lies
in~$B$. With~$\B_\e$ we denote the set of all $\e$-boxes having corners~$b$ for
which $b \in \SET{ -n\K, -n\K + \e, \ldots, n\K - 2\e, n\K - \e}^d$. Hence, $|\B_\e| = (2n\K/\e)^d$. If
all coefficients~$V^k_i$ of~$V$ are from $[-1, 1]$, which is true for all
models considered in this article, and if for all $k = 1, \ldots, d$ there is an
index~$i$ such that $|V^k_i| < 1$, which holds with probability~$1$ in all of our
models, then the $\e$-box of any vector $x \in \SET{ -\K, \ldots, \K }^n$ belongs
to~$\B_\e$. Note that all vectors~$x$ constructed in this article are from $\SET{
-\K, \ldots, \K }^n$. Hence, without any further explanation we will assume that $B_V(x)
\in \B_\e$.

In this article we extensively use tuples instead of sets. The reason for this is
that we are not only interested in certain components of a vector or matrix,
but we also want to describe in which order they are considered. This will be
clear after the introduction of the following notation. Let $n, m$ be positive
integers and let $a_1, \ldots, a_n$, $b_1, \ldots, b_m$ be arbitrary and not
necessarily pairwise distinct reals. We
define $[n] = (1, \ldots, n)$, $[n]_0 = (0, 1, \ldots, n)$, $|(a_1, \ldots,
a_n)| = n$ and $(a_1, \ldots, a_n) \cup (b_1, \ldots, b_m) = (a_1, \ldots, a_n,
b_1, \ldots, b_m)$. By $(a_1, \ldots, a_n) \setminus (b_1, \ldots, b_m)$ and
$(a_1, \ldots, a_n) \cap (b_1, \ldots, b_m)$ we denote the tuples we obtain by
removing all occurrences of elements from $(a_1, \ldots, a_n)$ that do/do not belong to $(b_1,
\ldots, b_m)$. We write $(a_1, \ldots, a_n) \subseteq (b_1, \ldots, b_m)$
if $m \geq n$ and if $(a_1, \ldots, a_n)$ can be obtained from $(b_1, \ldots, b_m)$
by removing $m-n$ elements.

Let~$x$ be a vector
and let~$A$ be a matrix. By $x|_{i_1 \ldots i_n} = x|_{(i_1, \ldots, i_n)}$ we
denote the column vector $(x_{i_1}, \ldots, x_{i_n})^\T$, by $A|_{(i_1, \ldots,
i_n)}$ we denote the matrix consisting of the rows $i_1, \ldots, i_n$ of
matrix~$A$ (in this order).

For an index set $I \subseteq [n]$ and a vector $y \in \SET{ 0, \ldots, \K }^n$
let $\S_I(y)$ denote the set of all solutions $z \in \S$ such that $z_i = y_i$
for all indices $i \in I$. For the sake of simplicity we also use the
notation $\S_I(\hat{y})$ to describe the set $\SET{ z \in \S \WHERE z_i =
\hat{y}_i\ \text{for all}\ i \in I }$ for a vector $\hat{y} \in \SET{ 0, \ldots, \K }^{|I|}$ when the components
of~$y$ are labeled by $y_{i_1}, \ldots, y_{i_{|I|}}$ where $I = (i_1, \ldots,
i_{|I|})$.

With $\ID{n}$ we refer to the $n \times n$-identity matrix
and with $\ZERO{m}{n}$ to  the $m \times n$-matrix whose entries are all~$0$.
If the number of rows and columns are clear, then we drop the indices.

For a set $M \subseteq \RR^n$ and a vector $y \in \RR^n$ we define $M + y \DEF \SET{ x + y \WHERE x \in M }$,
the Minkowski sum of~$M$ and~$\SET{y}$.

\begin{definition}
\label{def:Pareto optimality}
Let $\S \subseteq \RR^n$ be a set of solutions and let $f_1, \ldots, f_d \colon \S \to \RR$ be functions.
\begin{enumerate}

  \item Let $x, y \in \RR^n$ be vectors. We say that~$x$ \emph{dominates}~$y$ (with respect to $\SET{ f_1, \ldots, f_d }$), if $f_i(x) \leq f_i(y)$ for all $i \in [d]$ and $f_i(x) < f_i(y)$ for at least one $i \in [d]$. We say that~$x$ \emph{dominates}~$y$ \emph{strongly} (with respect to $\SET{ f_1, \ldots, f_d }$), if $f_i(x) < f_i(y)$ for all $i \in [d]$.

  \item Let $x \in \RR^n$ be a vector. We call~$x$ \emph{Pareto-optimal} or a \emph{Pareto-optimum} (with respect to~$\S$ and $\SET{ f_1, \ldots, f_d }$), if~$x$ is an element of~$\S$ and if no solution $y \in \S$ dominates~$x$. We call~$x$ \emph{weakly Pareto-optimal} or a \emph{weak Pareto-optimum} (with respect to~$\S$ and $\SET{ f_1, \ldots, f_d }$), if~$x$ is an element of~$\S$ and if no solution $y \in \S$ dominates~$x$ strongly.

\end{enumerate}
\end{definition}

We focus on Pareto-optimal solutions. The notions of strong dominance and weak Pareto-optimality are merely used for zero-preserving perturbations.

\section{Outline of our Approach}
\label{sec:Approach}

To prove our results we adapt and improve methods from the previous analyses by Moitra and O'Donnell~\cite{MoitraO11}
and by R\"oglin and Teng~\cite{RoeglinT09} and combine them in a novel way. Since all coefficients of the linear
objective functions lie in the interval $[-1,1]$, for every solution $x \in \S$ the vector $V^{1\ldots d}x$ lies 
in the hypercube $[-n\K,n\K]^d$. The first step is to partition this hypercube into $\e$-boxes. If~$\e$ is very small (exponentially small in~$n$),
then it is unlikely that there are two different solutions $x \in \S$ and $y \in \S$ that lie in the same $\e$-box~$B$ unless~$x$
and~$y$ differ only in positions that are not perturbed in any of the objective functions, in which case we consider them as the
same solution. In the remainder of this section we assume that no two solutions lie in the
same $\e$-box. Then, in order to bound the number of Pareto-optimal solutions, it suffices to count the number of non-empty
$\e$-boxes.

In order to prove Theorem~\ref{thm:MainFirstMoment} we show that for each fixed $\e$-box the probability that
it contains a Pareto-optimal solution is bounded by $k \cdot \K^{2d^2+2d+1} n^d\phi^d\e^d$ for the constant $k = 2^{2d^2+3d+1} \cdot (d \cdot (d+1))^{d^2}$ that is hidden in the $O$-notation. This implies
the theorem as the number of $\e$-boxes is $(2n\K/\e)^d$ and the exponent of~$\K$ is $2d^2 + 3d + 1 \leq 2(d+1)^2$. Fix an arbitrary $\e$-box~$B$. In the following we will call
a solution $x \in \S$ a \emph{candidate} if there is a realization of~$V$ such that~$x$ is Pareto-optimal and lies in~$B$. 
If there was only a single candidate $x \in \S$,
then we could bound the probability that there is a Pareto-optimal solution in~$B$ by the probability that 
this particular solution~$x$ lies in~$B$. This probability can easily be bounded from above by $\e^d\phi^d$ in the non-zero-preserving case.
However, in principle, every solution $x\in\S$ can be a candidate and a union bound over all of them leads to a factor
of~$|\S|$ in the bound, which we have to avoid. 

Following ideas of Moitra and O'Donnell, we divide the draw of the random matrix~$V$ into two steps. In the first step
some information about~$V$ is revealed that suffices to limit the set of candidates to a single solution $x \in \S$.
The exact position $V^{1 \ldots d}x$ of this solution is determined in the second step. If the information that is
revealed in these two steps is chosen carefully, then there is enough randomness left in the second step to bound the
probability that~$x$ lies in the $\e$-box~$B$. In Moitra and O'Donnell's analysis the coefficients in the matrix~$V$ are
partitioned into two groups. In the first step the first group of coefficients is drawn, which suffices to determine
the unique candidate~$x$, and in the second step the remaining coefficients are drawn, which suffices to bound the probability
that~$x$ lies in~$B$. The second part consists essentially of $d(d+1)/2$ coefficients, which causes the
factor of $\phi^{d(d+1)/2}$ in their bound.

We improve the analysis by a different choice of how to break the draw of~$V$ into two parts. As in the previous analysis,
most coefficients are drawn in the first step. Only~$d^2$ coefficients of~$V$ are drawn in the second step.
However, these coefficients are not left completely random as in~\cite{MoitraO11} because
after the other coefficients have been drawn there can still be multiple candidates
for Pareto-optimal solutions in~$B$. Instead,
the randomness is reduced further by drawing $d(d-1)$ linear combinations of these
random variables in the first step. These linear combinations have the property that, 
after they have been drawn, there is a unique candidate~$x$ whose
position can be described by~$d$ linear combinations that are linearly independent of
the linear combinations already drawn in the first step. In~\cite{RoeglinT09} it was
observed that linearly independent linear combinations of independent random variables
behave in some respect similar to independent random variables. With this insight one can
argue that in the second step there is still enough randomness to bound the probability
that~$x$ lies in~$B$. While the analysis in~\cite{RoeglinT09} yields only a bound proportional
to $\phi^{d^2}\e^d$, we prove an improved result for quasiconcave densities that
yields the desired bound proportional to $\phi^d\e^d$ (see Theorem~\ref{theorem.Prob:enough randomness}).

In order to bound the $c\th$ moment, we sum the probability that all $\e$-boxes $B_1,\ldots,B_c$ simultaneously contain a Pareto-optimal solution over all $c$-tuples $(B_1,\ldots,B_c)$ of $\e$-boxes.
We bound this probability from above by $k \cdot \K^{c^2(d+1)^2 + cd^2} n^{cd}\phi^{cd}\e^{cd}$ for the constant $k = 2^{c^2(d+1)^2 + cd^2 + cd} \cdot (cd(d+1))^{cd^2}$ that is hidden in the $O$-notation.
Since there are $(2n\K/\e)^{cd}$ different $c$-tuples of $\e$-boxes and the exponent of~$\K$ is $c^2(d+1)^2 + cd^2 + cd \leq (c+1)^2 (d+1)^2$, this implies the bound of $\K^{(c+1)^2(d+1)^2} \cdot O((n^2\phi)^{cd})$ for the smoothed $c\th$ moment of the number of Pareto-optimal solutions.

Let us fix a $c$-tuple $(B_1,\ldots,B_c)$ of $\e$-boxes.
The approach to bound the probability that all of these $\e$-boxes contain simultaneously a Pareto-optimal
solution is similar to the approach for the first moment. We divide the draw of~$V$ into two steps.
In the first step enough information is revealed to identify for each of the $\e$-boxes~$B_i$ a unique
candidate $x_i \in \S$ for a Pareto-optimal solution in~$B_i$. 
If we do this carefully, then there is enough randomness left in the second step to bound the probability
that $V^{1\ldots d}x_i \in B_i$ for every $i \in [c]$.
Again most coefficients are drawn in the first step and some linear combinations of the other $cd^2$ coefficients are
also drawn in the first step. However, we cannot simply repeat the construction for the first moment
independently $c$ times because then there might be dependencies between the events $V^{1\ldots d}x_i \in B_i$
for different~$i$. In order to bound the probability that in the second step all~$x_i$ lie in their corresponding
$\e$-boxes~$B_i$, we need to ensure that the events $V^{1\ldots d}x_i \in B_i$ are (almost) independent
after the information from the first step has been revealed.

The general approach to handle zero-preserving perturbations is closely related to the approach for bounding the
first moment for non-zero-preserving perturbations. However, additional complications have to be handled.
The main problem is that we cannot easily guarantee anymore that the linear combinations in the second step
are linearly independent of the linear combinations revealed in the first step. Essentially, the revealed linear
combinations describe the positions of some solutions, which we will call~\emph{auxiliary solutions} in the following. For non-zero-preserving perturbations revealing
this information is not critical as no solution has in any objective function exactly the same value as~$x$. 
For zero-preserving solutions it can, however, happen that the auxiliary solutions take exactly the same value
as~$x$ in one of the objective functions. Then there is not enough randomness
left in the second step anymore to bound the probability that~$x$ lies in this objective in the $\e$-interval described
by the $\e$-box~$B$.

In the remainder of this section we will present some more details on our analysis. We first present a simplified
argument to bound the smoothed number of Pareto-optimal solutions. Afterwards we will briefly discuss which changes
to this argument are necessary to bound higher moments and to analyze zero-preserving perturbations.

\paragraph{Smoothed Number of Pareto-optimal Solutions}

As an important building block in the proof of Theorem~\ref{thm:MainFirstMoment}
we use an insight from~\cite{MoitraO11} about how to test whether a given $\e$-box
contains a Pareto-optimal solution. Let us fix an $\e$-box $B = (b_1, b_1+\e] \times \ldots \times (b_d, b_d+\e]$
with corner $b = (b_1, \ldots, b_d)$.
The following algorithm takes as parameters the matrix~$V$ and the $\e$-box~$B$ and it
returns a solution~$x^{(0)}$.

\vspace{1ex}
\noindent\Witness{V, B}
\begin{algorithmic}[1]
  \STATE Set $\R_{d+1} = \S$.
  \FOR{$t = d, d-1, \ldots, 0$}
    \STATE Set $\C_t = \bSET{ z \in \R_{t+1} \WHERE V^{1 \ldots t} z \le b|_{1 \ldots t} }$. 
    \STATE Set $x^{(t)} = \argmin \bSET{ V^{t+1} z \WHERE z \in \C_t }$. 
    \STATE Set $\R_t = \bSET{ z \in \R_{t+1} \WHERE V^{t+1} z < V^{t+1} x^{(t)} }$.
  \ENDFOR
  \RETURN $x^{(0)}$  
\end{algorithmic}

\vspace{5pt}

The actual \tWitness function that we use in the proof of Theorem~\ref{thm:MainFirstMoment}
is more complex because it has to deal with some technicalities. In particular, the case
that some set~$\C_t$ is empty, in which~$x^{(t)}$ and~$\R_t$ would be undefined in
the function above, has to be handled.
For the purpose of illustration we ignore
these technicalities here and assume that~$\C_t$ is never empty. The crucial
observation that has been made by Moitra and O'Donnell is that if there is 
a Pareto-optimal solution $x \in \S$ that lies in~$B$, then $x^{(0)} = x$ (assuming
that no two solutions lie in the same $\e$-box). Hence, the solution~$x^{(0)}$ returned
by the \tWitness function is the only candidate for a Pareto-optimal solution in~$B$.
Our goal is to execute the \tWitness function and to obtain the solution~$x^{(0)}$ 
without revealing the entire matrix~$V$. We will see that it is indeed possible to
divide the draw of~$V$ into two steps such that in the first step enough information 
is revealed to execute the \tWitness function and such that in the second step there
is still enough randomness left to bound the probability that~$x^{(0)}$ lies in~$B$.

We want to illustrate the case $d=2$, in which there are one adversarial and two linear objective functions
(even though the following reasoning is true for all $d\in\NN$).
For this, assume that~$B$ contains a single solution~$x$ which is Pareto-optimal and that~$x$
is very close to the corner~$b$ of~$B$ which can be assumed if~$B$ is very small. Then $V^t z \leq b_t$ is
equivalent to $V^t z < V^t x$ for each $t \in [d]$.

Consider the situation depicted in Figure~\ref{fig:Witness Step I}. The first and the second objective value of each solution determine a point in the Euclidean plane. The additional value depicted next to this point represents the third objective value of each solution. Let us consider the situation before entering the loop. All points in Figure~\ref{fig:Witness Step I} are encircled meaning that $\R_3$ contains all solutions, i.e., $\R_3 = \S$. Now let us analyze the loop. The set~$\C_2$ contains all solutions that have smaller first and second objective values than~$x$ (gray area in Figure~\ref{fig:Witness Step II}). Among these solutions we pick the one with the smallest third objective value and denote it by~$x^{(2)}$. Set~$\R_2$ contains all solutions with a smaller third objective value (encircled points in Figure~\ref{fig:Witness Step III}). Note that in particular no solution of the gray region is considered anymore. On the other hand, $x$ belongs to~$\R_2$ due to Pareto-optimality.

The set~$\C_1$ contains all solutions from~$\R_2$ that have a smaller first objective value than~$x$ (encircled points in the gray area in Figure~\ref{fig:Witness Step IV}). Among these solutions~$x^{(1)}$ is the one with the smallest second objective value. Set~$\R_1$ contains all solutions from~$\R_2$ with a smaller second objective value (encircled points in Figure~\ref{fig:Witness Step V}). This set still contains~$x$, but no points from the gray region.

In the final iteration $t=0$ we obtain $\C_0 = \R_1$ since there is no restriction in the construction of~$\C_0$ anymore and $\C_0 \neq \emptyset$ since $x \in \R_1$. Solution~$x^{(0)}$ is among the remaining solutions the one with the smallest first objective value (Figure~\ref{fig:Witness Step VI}). This solution equals~$x$ and is now returned.

\begin{figure}[t]\newcommand{\figWidth}{0.30\textwidth}
  \begin{center}
    \subfloat[Initial situation]
    {
      \includegraphics[page=1, width=\figWidth]{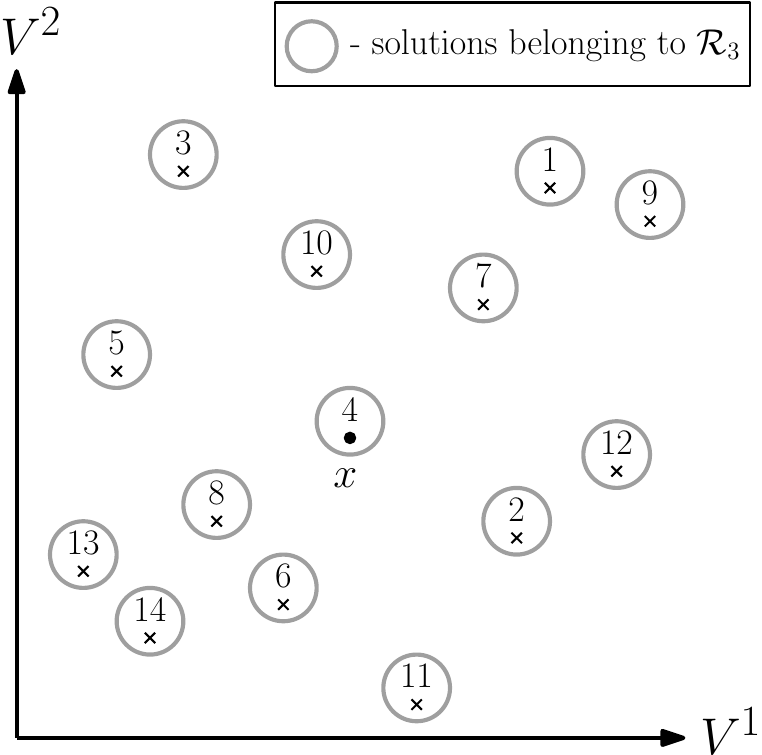}
      \label{fig:Witness Step I}
    }
    \subfloat[Determination of $x^{(2)}$]
    {
      \includegraphics[page=2, width=\figWidth]{Images/Witness.pdf}
      \label{fig:Witness Step II}
    }
    \subfloat[Determination of $\R_2$]
    {
      \includegraphics[page=3, width=\figWidth]{Images/Witness.pdf}
      \label{fig:Witness Step III}
    }
  \end{center}

  \begin{center}
    \subfloat[Determination of $x^{(1)}$]
    {
      \includegraphics[page=4, width=\figWidth]{Images/Witness.pdf}
      \label{fig:Witness Step IV}
    }
    \subfloat[Determination of $\R_1$]
    {
      \includegraphics[page=5, width=\figWidth]{Images/Witness.pdf}
      \label{fig:Witness Step V}
    }
    \subfloat[Determination of $x^{(0)}$]
    {
      \includegraphics[page=6, width=\figWidth]{Images/Witness.pdf}
      \label{fig:Witness Step VI}
    }
  \end{center}

  \caption{Execution of the \tWitness function for three objectives}
  \label{fig:Witness}
\end{figure}

Let us now discuss how the draw of~$V$ can be divided into two steps such that in the first step enough information 
is revealed to execute the \tWitness function and such that in the second step there
is still enough randomness left to bound the probability that~$x^{(0)}$ lies in~$B$.
For this let $I \subseteq [n]$ be a set of indices and assume that we know in advance which
values the solutions $x^{(0)}, \ldots, x^{(d)}$ take at these indices, i.e., assume
that we know $a^{(0)} = x^{(0)}|_I, \ldots, a^{(d)} = x^{(d)}|_I$ before executing the
\tWitness function. Then we can reconstruct $x^{(0)}, \ldots, x^{(d)}$ without having
to reveal the entire matrix~$V$. This can be done by the following algorithm, which
gets as additional parameters the set~$I$ and the matrix $A = [a^{(0)}, \ldots, a^{(d)}]$.

\vspace{1ex}
\noindent\Witness{V, I, A, B}
\begin{algorithmic}[1]
  \STATE Set $\R_{d+1} = \bigcup_{t'=0}^{d} \S_{I} \big( a^{(t')} \big)$. \label{line:Approach1}
  \FOR{$t = d, d-1, \ldots, 0$}
    \STATE Set $\C_t = \bSET{ z \in \R_{t+1} \WHERE V^{1 \ldots t} z \le b|_{1 \ldots t} } \cap \S_I \big( a^{(t)} \big)$. \label{line:Approach2} 
    \STATE Set $x^{(t)} = \argmin \bSET{ V^{t+1} z \WHERE z \in \C_t }$. 
    \STATE Set $\R_t = \bSET{ z \in \R_{t+1} \WHERE V^{t+1} z < V^{t+1} x^{(t)} } \cap \bigcup_{t'=0}^{t-1} \S_{I} \big( a^{(t')} \big)$. \label{line:Approach3}
  \ENDFOR
  \RETURN $(x^{(0)},\ldots,x^{(d)})$  
\end{algorithmic}

\vspace{5pt}

The additional restriction of the set~$\R_{d+1}$
does not change the outcome of the \tWitness function as all solutions $x^{(0)}, \ldots, x^{(d)}$
generated by the first \tWitness function are contained in the set~$\R_{d+1}$ defined
in line~\ref{line:Approach1} of the second \tWitness function. Similarly one can
argue that the additional restrictions in lines~\ref{line:Approach2} and~\ref{line:Approach3}
do not change the outcome of the algorithm because all solutions~$x^{(t)}$ generated by the first 
\tWitness function satisfy the restrictions that are made in the second \tWitness function.
Hence, if $a^{(0)} = x^{(0)}|_I, \ldots, a^{(d)} = x^{(d)}|_I$,
then both \tWitness functions generate the same~$x^{(0)}$.

We will now discuss how much information about~$V$ needs to be revealed 
in order to execute the second \tWitness function, assuming that the additional
parameters~$I$ and~$A$ are given. 
We assume that the coefficients~$V^t_i$ are revealed
for every $t \in [d]$ and $i \notin I$. For the remaining coefficients only certain linear
combinations need to be known in order to be able to execute the \tWitness function.
By carefully looking at the \tWitness function, one can deduce that for $t \in [d]$ only the following linear
combinations need to be known:
\begin{gather*}
    V^t|_I \cdot x^{(t)}|_I, \ldots, V^t|_I \cdot x^{(d)}|_I \COMMA \\
    V^t|_I \cdot (x^{(t-1)}-x^{(0)})|_I, \ldots, V^t|_I \cdot (x^{(t-1)}-x^{(t-2)})|_I \DOT 
\end{gather*}
These terms can be viewed as linear combinations of the random variables~$V^t_i$, $t \in [d]$, $i \in I$,
with coefficients from $\SET{ -\K, \ldots, \K }$. In addition to the already fixed random variables~$V^t_i$, $t \in[d]$, $i \notin I$, the following~$d$ linear combinations determine the position $V^{1 \ldots d}x$ of $x=x^{(0)}$:
\[
   V^{1}_I \cdot x^{(0)}|_I, \ldots, V^{d}_I \cdot x^{(0)}|_I \DOT
\]

An important observation on which our analysis is based is that if the
vectors $x^{(0)}|_I, \ldots, x^{(d)}|_I$ are linearly independent, then also all of
the above mentioned linear combinations are linearly independent. In particular,
the~$d$ linear combinations that determine the position of~$x$ cannot be expressed
by the other linear combinations. 
Usually, however, it is not possible to find a tuple $I \subseteq [n]$ of indices such
that the vectors $x^{(0)}|_I, \ldots, x^{(d)}|_I$ are linearly independent.
By certain technical modifications of the \tWitness function we will ensure that there always
exists such a tuple~$I$ with $|I| = d+1$ and that the last index of~$I$ is determined by the other~$d$ indices.
Since we do not know the tuple~$I$ and the matrix~$A$
in advance, we apply a union bound over all valid choices for these parameters, which yields
a factor of $(\K+1)^{(d+1)^2} n^d \leq 2^{(d+1)^2} \K^{(d+1)^2} n^d$ in the bound for the probability that there exists a Pareto-optimal
solution in~$B$.

R\"oglin and Teng~\cite{RoeglinT09} observed that the linear independence of the linear combinations  
implies that even if the linear combinations needed to execute the \tWitness function
are revealed in the first step, there is still enough randomness in the second step to
prove an upper bound on the probability that $V^{1 \ldots d}x$ lies in a fixed $\e$-box~$B$
that is proportional to~$\e^d$. The bound proven in~\cite{RoeglinT09} is, however, not strong
enough to improve Moitra and O'Donnell's result~\cite{MoitraO11} because the dependence on~$\phi$ is in the order
of $\Theta(\phi^{d^2})$ which is worse than the dependence of $\Theta(\phi^{d(d+1)/2})$ proven by
Moitra and O'Donnell.
We show that for quasiconcave density functions the
dependence in~\cite{RoeglinT09} can be improved significantly to $\Theta(\phi^{d})$, which 
yields the improved bound of $O(n^{2d}\phi^{d})$ in Theorem~\ref{thm:MainFirstMoment} for the binary case.

\paragraph{Higher Moments}

The analysis of higher moments is based on running the \tWitness function multiple times.
Let us fix a $c$-tuple $(B_1,\ldots,B_c)$ of $\e$-boxes. As described above, we bound the probability that
all of them contain a Pareto-optimal solution. For this, we run the \tWitness function $c$ times. In this way,
we get for every $j \in [c]$ a sequence $x^{(j,0)}, \ldots, x^{(j,d)}$ of solutions such that~$x^{(j,0)}$ is the
unique candidate for a Pareto-optimal solution in~$B_j$.

As above, we would like to execute the $c$ calls of the \tWitness function without having to reveal the
entire matrix~$V$. Again if we know for a subset $I \subseteq [n]$ the values that the solutions~$x^{(j,t)}$,
$j\in[c]$, $t\in[d]$, take at these positions, then we do not need to reveal the coefficients~$V^t_i$ with $i \in I$
to be able to execute the calls of the \tWitness function. As in the case of the first moment, it suffices to
reveal some linear combinations of these coefficients.

In order to guarantee that these linear combinations are linearly independent of the linear combinations that
determine the positions of the solutions~$x^{(j,0)}$, $j \in [c]$, we need to coordinate the calls of the \tWitness function.
Otherwise it might happen that, for example, the linear combinations revealed for executing the first call of the
witness function determine already the position of~$x^{(2,0)}$, the candidate for a Pareto-optimal solution in~$B_2$. 
Assume that the first call of the \tWitness function returns a sequence $x^{(1,0)}, \ldots, x^{(1,d)}$ of solutions
and that $I_1\subseteq [n]$ is a set of indices that satisfies the desired property that $x^{(1,0)}|_{I_1}, \ldots, x^{(1,d)}|_{I_1}$
are linearly independent. In order to achieve that all solutions generated in the following calls of the \tWitness function
are linearly independent of these linear combinations, we do not start a second independent call of the
\tWitness function, but we restrict the set of feasible solutions first. Instead of choosing $x^{(2,0)}, \ldots, x^{(2,d)}$
among all solutions from~$\S$, we restrict the set of feasible solutions for the second call of the \tWitness function
to $\S' = S_{I_1}(x^{(2,0)})$. Although we do not know $x^{(2,0)}$ in advance, we can assume to know some of its entries
due to a technical trick. Essentially, all solutions generated in call~$r$ of the \tWitness function
have to coincide with~$x^{(r,0)}$ in all positions that have been selected in one of the previous calls.

This and some additional tricks allow us to ensure that in the end there is a set $I \subseteq [n]$ with $|I| \le (d+1)c$ 
such that all vectors~$x^{(j,t)}|_I$, $j \in [c]$, $t \in [d]$ are linearly independent. Then we can again use 
the bound proven in~\cite{RoeglinT09} to bound the probability that $V^{1 \ldots d}x^{(j,0)} \in B_j$ simultaneously 
for every $j \in [c]$ from above by a term proportional to $\e^{cd} \phi^{cd^2}$. With our improved bound for quasiconcave
density functions, we obtain a bound proportional to $\e^{cd} \phi^{cd}$. Together with a union bound over all valid
choices for~$I$ and the values~$x^{(j,t)}|_I$, $j \in [c]$, $t \in [d]$, we obtain a bound of $k \cdot \K^{c^2(d+1)^2+cd^2} n^{cd} \phi^{cd} \e^cd$ on the probability that all candidates~$x^{(j,0)}$ lie in their corresponding $\e$-boxes for the constant $k = 2^{c^2(d+1)^2+cd^2+cd} \cdot (cd(d+1))^{cd^2}$ that is hidden in the $O$-notation. Together with the bound of $O((n\K)^{cd}/\e^{cd})$ for the number of $c$-tuples $(B_1, \ldots, B_c)$ this implies Theorem~\ref{thm:HigherMoment} as the exponent of~$\K$ is $c^2(d+1)^2 + cd^2 + cd \leq (c+1)^2 (d+1)^2$. 
 
\paragraph{Zero-preserving Perturbations}

If we use the same \tWitness function as above also for zero-preserving perturbations, then it can happen that there is a Pareto-optimal solution~$x$ in the
$\e$-box~$B$ that does not coincide with the solution~$x^{(0)}$ returned by the \tWitness function. This problem occurs,
for example, if $V^d \cdot x^{(d-1)} = V^d \cdot x^{(0)}$, which we cannot exclude if we allow zero-preserving perturbations.
We recommend to visualize this case for $d = 2$. On the other hand if we knew in advance that $V^d \cdot x^{(d-1)} = V^d \cdot x^{(0)}$,
then we could bound the probability of $V^dx^{(0)} \in (b_d, b_d+\e]$ already after the solution~$x^{(d-1)}$ has been generated.
Hence, if we were only interested in bounding this probability, we could terminate the \tWitness function already after~$x^{(d-1)}$
has been generated. Instead of terminating the \tWitness function at this point entirely, we keep in mind that $V^d \cdot x^{(0)}$
has already been determined and we restart the \tWitness function with the remaining objective functions only. 

Let us make this a bit more precise. As long as the solutions~$x^{(t)}$ generated by the \tWitness function differ in all
objective functions from~$x$, we execute the \tWitness function without any modification. Only if a solution~$x^{(t)}$
is generated that agrees with~$x$ in some objective functions, we deviate from the original \tWitness function. Let $K \subseteq [d]$
denote the objective functions in which~$x^{(t)}$ coincides with~$x$. At this point we can bound the probability that
$V^t\cdot x\in(b_t,b_t+\e]$ simultaneously for all $t\in K$. In order to also deal with the other objectives $t \notin K$,
we restart the \tWitness function. In this restart, we ignore all objective functions in~$K$ and we execute the
\tWitness function as if only objectives $t \notin K$ were present. Additionally we restrict in the restart the set of feasible solutions
to those that coincide in the objectives $t \in K$ with~$x$, i.e., to $\SET{ y \in \S \WHERE V^t \cdot y = V^t \cdot x \ \text{for all $t \in K$}}$.
With similar techniques as in the analysis of higher moments we ensure that different restarts lead to linearly independent
linear combinations.

This exploits that every Pareto-optimal solution~$x$ is also Pareto-optimal with respect to only the 
objective functions~$V^t$ with $t \notin K$ if the set~$\S$ is restricted to solutions
that agree with~$x$ in all objective functions~$V^t$ with $t \in K$. 
This property guarantees that whenever the \tWitness function is restarted, $x$ is still a Pareto-optimal solution
with respect to the restricted solution set and the remaining objective functions. 

It can happen that we have to restart the \tWitness function $d$ times before a unique candidate for a Pareto-optimal solution
in~$B$ is identified. As in each of these restarts at most~$d$ solutions are generated, the total number of solutions that is
generated can increase from $d+1$, as in the case of non-zero-preserving perturbations, to roughly~$d^2$. The set $I \subseteq [n]$
of indices restricted to which these solutions are linearly independent has a cardinality of at most~$d^3$. The reason for this
increase is that we have to choose more indices to obtain linear independence due to the fixed zeros. Taking a union bound over all
valid choices of~$I$, of the values that the generated solutions take at these positions, and of the possibilities when and due to
which objectives the restarts occur, yields Theorem~\ref{thm:MainZeroPreserving}. 
This theorem relies again on the result about linearly independent linear combinations of independent random variables
from~\cite{RoeglinT09} and its improved version for quasiconcave densities that we show in this article.

\section{Properties of (Weak) Pareto-optimal Solutions}
\label{sec:Pareto properties}

In this section we will identify the main properties of (weakly) Pareto-optimal solutions that lay the foundation for all variants of the \tWitness function. In the model without zero-preserving perturbations we only need properties of Pareto optima. In the model with zero-preserving perturbations, however, much more work has to be done and there we need the notion of weak Pareto optimality.

We start with an observation that is valid for both Pareto-optimal solutions and weak Pareto-optimal solutions.

\begin{proposition}
\label{Subset Proposition}
Let $\S \subseteq \RR^n$ be a set of solutions, let $f_1, \ldots, f_d \colon \S \to \RR$ be functions, let~$x^\star$ be a (weak) Pareto optimum with respect to~$\S$ and $\SET{ f_1, \ldots, f_d }$, and let $\S' \subseteq \S$ be a subset of solutions that contains~$x^\star$. Then~$x^\star$ is (weakly) Pareto-optimal with respect to~$\S'$ and $\SET{ f_1, \ldots, f_d }$.
\end{proposition}

The core idea of the \tWitness functions is given by the following lemma and Corollary~\ref{Recursion Corollary}. It implies that if~$x$ is Pareto-optimal with respect to~$\R_{t+1}$ and $\SET{ V^1, \ldots, V^{t+1} }$, then~$x$ is also Pareto-optimal with respect to~$\R_t$ and $\SET{ V^1, \ldots, V^t }$ (cf.\ function \Witness{V, B} described in Section~\ref{sec:Approach}). Given this as the induction step, it yields that~$x$ is Pareto-optimal with respect to~$\R_1$ and $\SET{ V^1 }$. This means that in iteration $t = 0$ we obtain $x^{(0)} = \argmin \bSET{ V^1 z \WHERE z \in \C_0 } = x$ because $\C_0 = \R_1$.

\begin{lemma}
\label{Recursion Lemma}
Let $\S \subseteq \RR^n$ be a set of solutions, let $f_1, \ldots, f_{t+1} \colon \S \to \RR$, $t \geq 1$, be functions, and let~$x^\star$ be a weak Pareto optimum with respect to~$\S$ and $\SET{ f_1, \ldots, f_{t+1} }$. We consider the set $\C \subseteq \S$ of solutions that dominate~$x^\star$ strongly with respect to $\SET{ f_1, \ldots, f_t }$.
\begin{enumerate}\renewcommand{\labelenumi}{(\Roman{enumi})}

  \item If $\C = \emptyset$, then~$x^\star$ is weakly Pareto-optimal with respect to~$\S$ and $\SET{ f_1, \ldots, f_t }$.

  \item If $\C \neq \emptyset$, then let $\hat{f} = \min_{x \in \C} f_{t+1}(x)$. Then $f_{t+1}(x^\star) \leq \hat{f}$. Furthermore, if $f_{t+1}(x^\star) < \hat{f}$, then~$x^\star$ is weakly Pareto-optimal with respect to $\R \mathop{:=} \{ x \in \S \WHERE f_{t+1}(x) < \hat{f} \}$ and $\SET{ f_1, \ldots, f_t }$.

\end{enumerate}
\end{lemma}

\begin{proof}
Claim~(I) holds due to the definition of weak Pareto optimality. Let us consider Claim~(II). If the inequality $f_{t+1}(x^\star) \leq \hat{f}$ does not hold, then $\hat{x} = \argmin_{x \in \C} f_{t+1}(x)$ dominates~$x^\star$ strongly with respect to $\SET{ f_1, \ldots, f_{t+1} }$. This is a contradiction since~$x^\star$ is weakly Pareto-optimal with respect to~$\S$ and $\SET{ f_1, \ldots, f_{t+1} }$.

Now let us show that~$x^\star$ is weakly Pareto-optimal with respect to~$\R$ and $\SET{ f_1, \ldots, f_t }$ if $f_{t+1}(x^\star) < \hat{f}$. The condition ensures that $x^\star \in \R$. Assume to the contrary that there exists a $y \in \R$ that dominates~$x^\star$ strongly with respect to $\SET{ f_1, \ldots, f_t }$. Since $\R \subseteq \S$, this implies $y \in \C$. Due to $y \in \R$ we obtain the contradiction $f_{t+1}(y) < \hat{f} \leq f_{t+1}(y)$, where the second inequality follows from the definition of~$\hat{f}$ and $y \in \C$.
\end{proof}

If the functions $f_1, \ldots, f_t$ in Lemma~\ref{Recursion Lemma} are injective, we can also obtain a statement about Pareto optima.

\begin{corollary}
\label{Recursion Corollary}
Let $\S \subseteq \RR^n$ be a set of solutions, let $f_1, \ldots, f_{t+1} \colon \S \to \RR$, $t \geq 1$, be functions, where $f_1, \ldots, f_t$ are injective, and let~$x^\star$ be a Pareto optimum with respect to~$\S$ and $\SET{ f_1, \ldots, f_{t+1} }$. We consider the set $\C \subseteq \S$ of solutions that dominate~$x^\star$ strongly with respect to $\SET{ f_1, \ldots, f_t }$.

\begin{enumerate}\renewcommand{\labelenumi}{(\Roman{enumi})}

  \item If $\C = \emptyset$, then~$x^\star$ is Pareto-optimal with respect to~$\S$ and $\SET{ f_1, \ldots, f_t }$.

  \item If $\C \neq \emptyset$, then let $\hat{f} = \min_{x \in \C} f_{t+1}(x)$. Then $f_{t+1}(x^\star) < \hat{f}$. Furthermore, $x^\star$ is Pareto-optimal with respect to $\R \mathop{:=} \{ x \in \S \WHERE f_{t+1}(x) < \hat{f} \}$ and $\SET{ f_1, \ldots, f_t }$.

\end{enumerate}
\end{corollary}

\begin{proof}
First of all we observe that a solution $y \in \S$ dominates~$x^\star$ with respect to $\SET{ f_1, \ldots, f_t }$ if and only if~$y$ dominates~$x^\star$ strongly with respect to $\SET{ f_1, \ldots, f_t }$. This is due to the injectivity of the functions $f_1, \ldots, f_t$. Consequently, Claim~(I) follows from the definition of Pareto optimality. Let us consider Claim~(II). Assume to the contrary that $\hat{f} \leq f_{t+1}(x^\star)$. In this case, the solution $\hat{x} = \argmin_{x \in \C} f_{t+1}(x)$ would dominate~$x^\star$ with respect to $\SET{ f_1, \ldots, f_{t+1} }$ contradicting the assumption that~$x^\star$ is Pareto-optimal. Hence, $f_{t+1}(x^\star) < \hat{f}$.

Due to Lemma~\ref{Recursion Lemma}, $x^\star$ is weakly Pareto-optimal with respect to~$\R$ and $\SET{ f_1, \ldots, f_t }$ because every Pareto optimum is also a weak Pareto optimum. As these functions are injective, $x^\star$ is even Pareto-optimal with respect to~$\R$ and $\SET{ f_1, \ldots, f_t }$.
\end{proof}

For the model with zero-preserving perturbations we need one more lemma that allows us to handle non-injectivity appropriately.

\begin{lemma}
\label{Restart Lemma}
Let $\S \subseteq \RR^n$ be a set of solutions, let $f_1, \ldots, f_{t+1} \colon \S \to \RR$, $t \geq 1$, be functions, and let~$x^\star$ be a Pareto optimum with respect to~$\S$ and $\SET{ f_1, \ldots, f_{t+1} }$. Furthermore, let $K \subseteq [t+1]$ be a tuple of indices and let~$\S'$ be a subset of $\SET{ x \in \S \WHERE f_k(x) = f_k(x^\star)\ \text{for all}\ k \in K }$. Then~$x^\star$ is Pareto-optimal with respect to~$\S'$ and $\SET{ f_k \WHERE k \in [t+1] \setminus K }$.
\end{lemma}

\begin{proof}
Assume to the contrary that~$x^\star$ is not Pareto-optimal. Then there exists a solution $y \in \S'$ such that~$y$ dominates~$x^\star$ with respect to $\SET{ f_k \WHERE k \in [t+1] \setminus K }$. Since $f_k(y) = f_k(x^\star)$ for all $k \in K$, solution~$y$ also dominates~$x^\star$ with respect to $\SET{ f_1, \ldots, f_{t+1} }$. This contradicts the assumption that~$x^\star$ is Pareto-optimal.
\end{proof}

\section{Non-zero-preserving Perturbations}

\label{sec:ZeroDestroying}

\subsection{Smoothed Number of Pareto-optimal Solutions}

\label{subsec:FirstMoment}

To prove Theorem~\ref{thm:MainFirstMoment} we assume without loss of generality that $n \geq d+1$ and consider the function given as Algorithm~\ref{algorithm:Witness} which we call the \tWitness function. It is very similar to the one suggested by Moitra and O'Donnell, but with an additional parameter~$I$. This parameter is a tuple of forbidden indices: it restricts the set of indices we are allowed to choose from. For the analysis of the smoothed number of Pareto-optimal solutions we will set $I = ()$. The parameter becomes important in the next section when we analyze higher moments.

\IfUseAlgorithmIIeStyle
{
\LinesNumbered
\begin{algorithm*}[h!t]
  \caption{\Witness{V, x, I}}
  \label{algorithm:Witness}
  set $I_{d+1} = I$\;
  set $\R_{d+1} = \S_{I_{d+1}}(x)$\label{l:initial loser set} \;
  \For{$t = d, d-1, \ldots, 0$}
  {
    set $\C_t = \bSET{ z \in \R_{t+1} \WHERE V^{1 \ldots t} z < V^{1 \ldots t} x }$\label{l:winner set} \;
    \uIf{$\C_t \neq \emptyset$}
    {\label{l:interesting case}
      set $x^{(t)} = \argmin \bSET{ V^{t+1} z \WHERE z \in \C_t }$\label{l:winner} \;
      \lIf{$t = 0$}{\Return{$x^{(t)}$} \;}
      set $i_t = \min \bSET{ i \in [n] \WHERE x^{(t)}_{i_t} \neq x_{i_t} }$\label{l:index} \;
      set $I_t = I_{t+1} \cup (i_t)$\;
      set $\R_t = \bSET{ z \in \R_{t+1} \WHERE V^{t+1} z < V^{t+1} x^{(t)} } \cap \S_{I_t}(x)$\label{l:loser set} \;
    }
    \Else
    {\label{l:technical case}
      set $i_t = \min ([n] \setminus I_{t+1})$\label{l:trivial index} \;
      set $I_t = I_{t+1} \cup (i_t)$\;
      set $x^{(t)}_i = \begin{cases}
                         \min (\SET{ 0, \ldots, \K } \setminus \SET{ x_i }) & \tIF i = i_t \cr
                         x_i                                                & \OTHERWISE
                       \end{cases}$\label{l:trivial winner} \;
      set $\R_t = \R_{t+1} \cap \S_{I_t}(x)$\label{l:trivial loser set} \;
    }
  }
  \Return{$(\bot, \ldots, \bot)$} \;
\end{algorithm*}%
}
{
\begin{algorithm*}[h!t]
  \caption{\Witness{V, x, I}}
  \label{algorithm:Witness}
  \begin{algorithmic}[1]  
    \STATE Set $I_{d+1} = I$.
    \STATE Set $\R_{d+1} = \S_{I_{d+1}}(x)$. \label{l:initial loser set}
    \FOR{$t = d, d-1, \ldots, 0$}
      \STATE Set $\C_t = \bSET{ z \in \R_{t+1} \WHERE V^{1 \ldots t} z < V^{1 \ldots t} x }$. \label{l:winner set}
      \IF{$\C_t \neq \emptyset$} \label{l:interesting case}
        \STATE Set $x^{(t)} = \argmin \bSET{ V^{t+1} z \WHERE z \in \C_t }$. \label{l:winner}
        \STATE \textbf{if} $t = 0$ \textbf{then} \textbf{return} $x^{(t)}$
        \STATE Set $i_t = \min \bSET{ i \in [n] \WHERE x^{(t)}_{i_t} \neq x_{i_t} }$. \label{l:index}
        \STATE Set $I_t = I_{t+1} \cup (i_t)$.
        \STATE Set $\R_t = \bSET{ z \in \R_{t+1} \WHERE V^{t+1} z < V^{t+1} x^{(t)} } \cap \S_{I_t}(x)$. \label{l:loser set}
      \ELSE \label{l:technical case}
        \STATE Set $i_t = \min ([n] \setminus I_{t+1})$. \label{l:trivial index}
        \STATE Set $I_t = I_{t+1} \cup (i_t)$.
        \STATE Set $x^{(t)}_i = \begin{cases}
                  \min (\SET{ 0, \ldots, \K } \setminus \SET{ x_i }) & \tIF i = i_t \COMMA \cr
                  x_i                                                & \OTHERWISE \DOT
                \end{cases}$ \label{l:trivial winner}
        \STATE Set $\R_t = \R_{t+1} \cap \S_{I_t}(x)$. \label{l:trivial loser set}
      \ENDIF
    \ENDFOR
    \RETURN $(\bot, \ldots, \bot)$    
  \end{algorithmic}
\end{algorithm*}%
}%

Let us give some remarks about the \tWitness function. Note that $\C_0 = \R_1$
since $V^{1 \ldots t} z < V^{1 \ldots t} x$ is no
restriction if $t = 0$. In Line~\ref{l:winner} ties are broken by taking
the lexicographically first solution~$x^{(t)}$. For $t \geq 1$ the index~$i_t$ in Line~\ref{l:index} 
exists because $V^1 x^{(t)} < V^1 x$ which implies $x^{(t)} \neq x$.

Unless stated otherwise, we assume that the following \emph{$\OK$-event} $\OK(V)$ occurs. This event occurs if
$|V^k \cdot (y - z)| \geq \e$ for every $k \in [d]$ and for arbitrary two distinct
solutions $y \neq z \in \S$ and if for all $k \in [d]$ there is an index $i \in [n]$
for which $|V^k_i| < 1$. Amongst others, the first property ensures that there is a unique
$\arg \min$ in Line~\ref{l:winner} and that the functions $V^1, \ldots, V^d$ are injective.
The latter property, which holds with probability~$1$,
ensures that $B_V(x) \in \B_\e$ for all vectors $x \in \SET{ -\K, \ldots, \K }^n$. Later we will
see that the $\OK$-event occurs with sufficiently high probability.

Before we start to analyze the \tWitness function, let us discuss the differences between the function
\Witness{V, B} described in Section~\ref{sec:Approach} and the function \Witness{V, x, I} given as
Algorithm~\ref{algorithm:Witness}. As described in Section~\ref{sec:Approach} for the illustrative case $d=2$,
the parameters~$B$ and~$x$ play exactly the same role if $B = B_V(x)$ assuming that the $\OK$-event
holds. As stated earlier, the additional parameter~$I$ in the function \Witness{V, x, I} has no meaning
for the analysis of the first moment. To prove Theorem~\ref{thm:MainFirstMoment}, we simply set it to the empty tuple. The case $\C_t \neq \emptyset$ (Line~\ref{l:interesting case}) is the interesting case, which is also captured by the function \Witness{V, B}. The case $\C_t = \emptyset$ (Line~\ref{l:technical case}) is the technical case. Here it is only important that we choose an index~$i_t$ that is not an element of~$I_{t+1}$ and that the vector~$x^{(t)}$ is defined such that $x^{(t)}$ coincides with~$x$ in all components $i \in I_{t+1}$ and that it does not coincide with~$x$ in component~$i_t$. Note that the vector~$x^{(t)}$ as we define it in Line~\ref{l:trivial winner} is not necessarily a solution from~$\S$.

In the remainder of this section we only consider the case that~$x$ is Pareto-optimal,
that~$I$ is an arbitrary index tuple with pairwise distinct indices, and that
the number~$|I|$ of indices contained in~$I$ is at most $n-(d+1)$. This ensures that the indices $i_0, \ldots, i_d$ exist.

\begin{lemma}
\label{lemma:Witness}
The call \Witness{V, x, I} returns the vector $x^{(0)} = x$.
\end{lemma}

\begin{proof}
We show the following claim by induction on~$t$.

\begin{claim}
\label{claim:correctness zero destroying}
For all $t \in [d+1]$, solution~$x$ is Pareto-optimal with respect to~$\R_t$ and $\SET{ V^1, \ldots, V^t }$.
\end{claim}

\begin{proof}[Proof of Claim~\ref{claim:correctness zero destroying}]
Note that the functions $V^1, \ldots, V^d$ are injective due to the assumption that the $\OK$-event occurs. This allows us to apply Corollary~\ref{Recursion Corollary}. Recalling that $x \in \S_{I'}(x)$ for every index tuple~$I'$, Claim~\ref{claim:correctness zero destroying} is true for $t = d+1$ by assumption and due to Proposition~\ref{Subset Proposition}.

Now let us assume that the claim holds for some value $t+1$ and consider set~$\C_t$. We distinguish between two cases. If $\C_t = \emptyset$, then $\R_t = \R_{t+1} \cap \S_{I_t}(x)$ and the claim follows from the induction hypothesis, from Corollary~\ref{Recursion Corollary}~(I), and from Proposition~\ref{Subset Proposition}. If $\C_t \neq \emptyset$, then $\R_t = \R'_t \cap \S_{I_t}(x)$ for $\R'_t = \SET{ z \in \R_{t+1} \WHERE V^{t+1} z < V^{t+1} x^{(t)} }$. Hence, the claim follows from the induction hypothesis, from Corollary~\ref{Recursion Corollary}~(II), and from Proposition~\ref{Subset Proposition}.
\end{proof}

In accordance with Claim~\ref{claim:correctness zero destroying}, we obtain for $t=1$ that~$x$ is Pareto-optimal with respect to~$\R_1$ and $\tSET{ V^1 }$. In particular, $x \in \R_1 = \C_0 \neq \emptyset$, i.e., $x^{(0)} = x$. This solution will be returned in iteration $t = 0$. \qed
\end{proof}

At a first glance it seems odd to compute a solution~$x$ by calling a function with~$x$ as parameter. However, we will see that not all information about~$x$ is required to execute the call \Witness{V, x, I}. To be a bit more precise, the indices $i_1, \ldots, i_d$ and the entries at the positions $i \in I \cup (i_1, \ldots, i_d)$ of the vectors~$x^{(t)}$ constructed during the execution of the \tWitness function suffice to simulate the execution of \tWitness without knowing~$x$ completely (see Lemma~\ref{lemma:same behavior}). We will call these information a \emph{certificate} (see Definition~\ref{definition:certificate}). For technical reasons we will assume that we also know the entries of the vectors~$x^{(t)}$ at position $i_0 = \min ([n] \setminus (I \cup (i_1, \ldots, i_d)))$ and, for the analyis of higher moments, also at further positions.

For our purpose it is not necessary to know how to obtain the required information about~$x$ to reconstruct it. It suffices to know that the set of possible certificates is sufficiently small (see Lemma~\ref{lemma:certificate space}) and that for at least one of them the simulation of the execution of \tWitness returns~$x$ (see Lemma~\ref{lemma:same behavior}). This is one crucial property which will help us to bound the expected number of Pareto-optimal solutions.

\begin{definition}
\label{definition:certificate}
Let $x^{(0)}, \ldots, x^{(d)}$ be the vectors and~$I_1$ be the index tuple constructed during the call \Witness{V, x, I} and set $i_0 = \min \big( [n] \setminus I_1 \big)$ and $I_0 = I_1 \cup (i_0)$. We call the pair $(I_0, A_0)$ for $A_0 = \big[ x^{(d)}, \ldots, x^{(0)} \big]$ the \emph{$(V, I)$-certificate} of~$x$. The pair $(I_0, A)$ for $A = A_0|_{I_0}$ is called the \emph{restricted $(V, I)$-certificate} of~$x$. We call a pair $(I', A')$ a \emph{(restricted) $I$-certificate}, if there exist a realization~$V$ such that the $\OK$-event occurs and a Pareto-optimal solution $x \in \S$ such that $(I', A')$ is the (restricted) $(V, I)$-certificate of~$x$. By~$\CS(I)$ we denote the set of all restricted $I$-certificates.
\end{definition}

The notation used in this section is summarized in Table~\ref{table:Notation}.

\begin{table}[ht]
\begin{tabularx}{\textwidth}{|lX|} \hline
$\S \subseteq \SET{ 0, \ldots, \K }^n$\phantom{${A^A}^A$} & set of feasible solutions\\
$V^1,\ldots,V^d$ & linear objective functions\\
$V^{d+1} \colon \S \to \RR$ & adversarial objective function\\
$V^t_1,\ldots,V^t_n$ & coefficients of~$V^t$ for~$t\in[d]$\\
$f^t_i\colon[-1,1]\to[0,\phi]$ & probability density of~$V^t_i$ for~$t\in[d]$ and~$i\in[n]$\\
$V \in \RR^{d \times n}$ & matrix of coefficients of~$V^1,\ldots,V^d$\\ 
$\PO(V)$ & number of Pareto-optimal solutions for~$V$\\
$\OK(V)$ & event that $|V^k \cdot (y - z)| \geq \e$ for every $k \in [d]$ and for arbitrary two distinct
solutions $y \neq z \in \S$ and that for all $k \in [d]$ there is an index $i \in [n]$ for which $|V^k_i| < 1$\\
$\B_\e$ & set of all $\e$-boxes having corners~$b$ for which \\
        & $b \in \SET{ -n\K, -n\K + \e, \ldots, n\K - 2\e, n\K - \e}^d$\\
$B_V(x)$ & $\e$-box~$B$ for which $V^{1 \ldots d} x \in B$\\
$x^{(0)}, \ldots, x^{(d)}$ & vectors constructed during the call of Algorithm~\ref{algorithm:Witness}\\
 $\R_{d+1},\ldots,\R_0$ & sets constructed during the call of Algorithm~\ref{algorithm:Witness}\\
 $\C_{d},\ldots,\C_0$ & sets constructed during the call of Algorithm~\ref{algorithm:Witness}\\ 
$I_1$ & index tuple constructed during the call of Algorithm~\ref{algorithm:Witness}\\ 
$i_0$ & $\min \big( [n] \setminus I_1 \big)$\\
$I_0$ & $I_1 \cup (i_0)$\\
$(I_0, A_0)$ & $(V, I)$-certificate of~$x$ where $A_0 = \big[ x^{(d)}, \ldots, x^{(0)} \big]$\\
 $(I_0, A)$ & restricted $(V, I)$-certificate of~$x$ where $A = A_0|_{I_0}$\\
 $(I', A')$ & (restricted) $I$-certificate, i.e., there exist a realization~$V$ such that the $\OK$-event occurs and a Pareto-optimal solution $x \in \S$ such that $(I', A')$ is the (restricted) $(V, I)$-certificate of~$x$\\
 $\CS(I)$ & set of all restricted $I$-certificates\\
 $u \in \SET{ 0, \ldots, \K }^n$ & shift vector used in Algorithm~\ref{algorithm:Witness simulation}
\\ \hline
\end{tabularx}
\caption{Notation used in Section~\ref{subsec:FirstMoment}}
\label{table:Notation}
\end{table}

For the analysis of the first moment we only need restricted $I$-certificates. Our analysis of higher moments requires more knowledge about the vectors~$x^{(t)}$ than just the values~$x_i$ for $i \in I_0$. The additional indices are, however, depending on further calls of the \tWitness function which we do not know a priori. This is why we have to define two types of certificates. For the sake of reusability we formulate some statements more general than necessary for this section.

\begin{lemma}\newcommand{\AST}{`$*$'}
\label{lemma:certificate form}
Let~$V$ be an arbitrary realization for which the $\OK$-event occurs, let~$x$ be a Pareto-optimal solution with respect to~$\S$ and~$V$, and let $(I_0, A)$ be the restricted $(V, I)$-certificate of~$x$. Then $I_0 = (j_1, \ldots, j_{|I|+d+1})$ consists of pairwise distinct indices and
\[
  A = \begin{bmatrix}
    x_{j_1} & \ldots & x_{j_{|I|}} & \overline{x_{j_{|I|+1}}} & *      & \ldots                   & * \cr
    \vdots  &        & \vdots      & x_{j_{|I|+1}}            & \ddots & \ddots                   & \vdots \cr
    \vdots  &        & \vdots      & \vdots                   & \ddots & \overline{x_{j_{|I|+d}}} & * \cr
    x_{j_1} & \ldots & x_{j_{|I|}} & x_{j_{|I|+1}}            & \ldots & x_{j_{|I|+d}}            & x_{j_{|I|+d+1}}
  \end{bmatrix}^\T
  \in \SET{ 0, \ldots, \K }^{(|I|+d+1) \times (d+1)} \COMMA
\]
where each \AST\ can be an arbitrary value from $\SET{ 0, \ldots, \K }$ (different \AST-entries can represent different values) and where $\overline{z}$ for a value $z \in \SET{ 0, \ldots, \K }$ can be an arbitrary value from $\SET{ 0, \ldots, \K } \setminus \SET{ z }$.
\end{lemma}

\begin{proof}
Lemma~\ref{lemma:Witness} implies that the last column $(x_{j_1}, \ldots, x_{j_{|I|+d}+1})^\T$ of~$A$ equals $x^{(0)}_{I_0} = x|_{I_0}$.
Hence, we just have to consider the first~$d$ columns of~$A$. Note that $I = (j_1, \ldots, j_{|I|})$ and $j_{|I|+1}, \ldots, j_{|I|+d+1} = i_d, \ldots, i_0$. The construction of the sets~$\R_t$ yields $\R_t \subseteq \S_{I_t}(x)$ (see Lines~\ref{l:initial loser set}, \ref{l:loser set}, and~\ref{l:trivial loser set}). Index~$i_t$ is always chosen such that $i_t \notin I_{t+1}$: If it is constructed in Line~\ref{l:index}, then $x^{(t)}_{i_t} \neq x_{i_t}$. Since in this case we have
\[
	x^{(t)}
	\in \C_t
	\subseteq \R_{t+1}
	\subseteq \S_{I_{t+1}}(x) \COMMA
\]
index~$i_t$ cannot be an element of~$I_{t+1}$. In Line~\ref{l:trivial index}, index~$i_t$ is explicitely constructed such that $i_t \notin I_{t+1}$. The same argument holds for index~$i_0$. Hence, the indices of~$I_0$ are pairwise distinct.

Now, consider the column of~$A$ corresponding to vector~$x^{(t)}$ for $t \in [d]$. If $\C_t = \emptyset$, then the form of the column follows directly from the construction of~$x^{(t)}$ in Line~\ref{l:trivial winner} and from the fact that the indices of~$I_0$ are pairwise distinct. If $\C_t \neq \emptyset$, then
\[
	x^{(t)}
	\in \C_t
	\subseteq \R_{t+1}
	\subseteq \S_{I_{t+1}}(x) \COMMA
\]
i.e., $x^{(t)}$ coincides with~$x$ in all indices $i \in I_{t+1}$. By the choice of~$i_t$ in Line~\ref{l:index} we get $x^{(t)}_{i_t} \in \SET{ 0, \ldots, \K } \setminus \SET{x_{i_t}}$. This concludes the proof.
\end{proof}

Let $(I_0, A_0)$ be the $(V, I)$-certificate of~$x$ and let $J \supseteq I_0$ be a tuple of pairwise distinct indices.
As mentioned before, our goal is to execute the \tWitness function without revealing the entire matrix~$V$.
For this we consider the following variant of the \tWitness function given as Algorithm~\ref{algorithm:Witness simulation} that uses information about~$x$ given by the index tuple~$J$, the matrix $A = A_0|_J$ with columns $a^{(d)}, \ldots, a^{(0)}$, a shift vector~$u$ and the $\e$-box $B = B_V(x-u)$ instead of vector~$x$ itself. The meaning of the shift vector will become clear when we analyze the probability of certain events. We will see that not all information about~$V$ needs to be revealed to execute the new \tWitness function, i.e., we have some randomness left which we can use later. With the choice of the shift vector we can control which information has to be revealed for executing the \tWitness function.

\IfUseAlgorithmIIeStyle
{
\LinesNumbered
\begin{algorithm*}[h!t]
  \caption{\Witness{V, J, A, B, u}}
  \label{algorithm:Witness simulation}
  let~$b$ be the corner of~$B$ \;
  set $\R_{d+1} = \bigcup_{s=0}^{d} \S_J \big( a^{(s)} \big)$ \;
  \For{$t = d, d-1, \ldots, 0$}
  {
    set $\C_t = \bSET{ z \in \R_{t+1} \WHERE V^{1 \ldots t} \cdot (z - u) \leq b|_{1 \ldots t} } \cap \S_J \big( a^{(t)} \big)$\label{l.II:winner set} \;
    \uIf{$\C_t \neq \emptyset$}
    {
      set $x^{(t)} = \argmin \bSET{ V^{t+1} z \WHERE z \in \C_t }$\label{l.II:winner} \;
      \lIf{$t = 0$}{\Return{$x^{(t)}$} \;}
      set $\R_t = \bSET{ z \in \R_{t+1} \WHERE V^{t+1} z < V^{t+1} x^{(t)} } \cap \bigcup_{s=0}^{t-1} \S_J \big( a^{(s)} \big)$\label{l.II:loser set} \;
    }
    \Else
    {
      set $x^{(t)} = (\bot, \ldots, \bot)$ \;
      set $\R_t = \R_{t+1} \cap \bigcup_{s=0}^{t-1} \S_J \big( a^{(s)} \big)$ \;
    }
  }
  \Return{$x^{(0)}$} \;
\end{algorithm*}%
}
{
\begin{algorithm*}[h!t]
  \caption{\Witness{V, J, A, B, u}}
  \label{algorithm:Witness simulation}
  \begin{algorithmic}[1]
    \STATE Let~$b$ be the corner of~$B$.
    \STATE Set $\R_{d+1} = \bigcup_{s=0}^{d} \S_J \big( a^{(s)} \big)$.
    \FOR{$t = d, d-1, \ldots, 0$}
      \STATE Set $\C_t = \bSET{ z \in \R_{t+1} \WHERE V^{1 \ldots t} \cdot (z - u) \leq b|_{1 \ldots t} } \cap \S_J \big( a^{(t)} \big)$. \label{l.II:winner set}
      \IF{$\C_t \neq \emptyset$}
        \STATE Set $x^{(t)} = \argmin \bSET{ V^{t+1} z \WHERE z \in \C_t }$. \label{l.II:winner}
        \STATE \textbf{if} $t = 0$ \textbf{then} \textbf{return} $x^{(t)}$
        \STATE Set $\R_t = \bSET{ z \in \R_{t+1} \WHERE V^{t+1} z < V^{t+1} x^{(t)} } \cap \bigcup_{s=0}^{t-1} \S_J \big( a^{(s)} \big)$. \label{l.II:loser set}
      \ELSE
        \STATE Set $x^{(t)} = (\bot, \ldots, \bot)$.
        \STATE Set $\R_t = \R_{t+1} \cap \bigcup_{s=0}^{t-1} \S_J \big( a^{(s)} \big)$.
      \ENDIF
    \ENDFOR
    \RETURN $x^{(0)}$
  \end{algorithmic}
\end{algorithm*}%
}

\begin{lemma}
\label{lemma:same behavior}
Let $(I_0, A_0)$ be the $(V, I)$-certificate of~$x$, let $J \supseteq I_0$ be an arbitrary tuple of pairwise distinct indices, let $A = A_0|_J$, let $u \in \SET{ 0, \ldots, \K }^n$ be an arbitrary vector, and let $B = B_V(x - u)$. Then the call \Witness{V, J, A, B, u} returns vector~$x$.
\end{lemma}

Before we give a formal proof of Lemma~\ref{lemma:same behavior} we try to give some intuition for it. Instead of considering the whole set~$\S$ of solutions we restrict it to vectors that look like the vectors we want to reconstruct in the next iterations, i.e., we intersect the current set with the set $\bigcup_{s=0}^{t-1} \S_J \big( a^{(s)} \big)$ in iteration~$t$. In this way we only deal with subsets of the original sets, but we do not lose the vectors we want to reconstruct since $J \supseteq I_0$. This restriction to the essential candidates of solutions allows us to execute this variant of the \tWitness function with only partial information about~$V$.

\begin{proof}
Let~$\R'_t$, $\C'_t$, and~$x'^{(t)}$ denote the sets and vectors constructed during the execution of the call \Witness{V, J, A, B, u} and let $\R_t$, $\C_t$, and~$x^{(t)}$ denote the sets and vectors constructed during the execution of call \Witness{V, x, I}. We prove the following claims simultaneously by induction.

\begin{claim}
\label{claim:same behavior I}
$\R'_t \subseteq \R_t$ for all $t \in [d+1]$.
\end{claim}

\begin{claim}
\label{claim:same behavior II}
$x'^{(t)} = x^{(t)}$ for all $t \in [d]_0$ for which $\C_t \neq \emptyset$.
\end{claim}

\begin{claim}
\label{claim:same behavior III}
$x^{(s)} \in \R'_t$ for all $t \in [d+1]$ and all $s \in [t-1]_0$ for which $\C_s \neq \emptyset$.
\end{claim}

\begin{proof}[Proof of Claim~\ref{claim:same behavior I}, Claim~\ref{claim:same behavior II}, and Claim~\ref{claim:same behavior III}]
Let us first focus on the shift vector~$u$ and compare Line~\ref{l:winner set} of the first \tWitness function (Algorithm~\ref{algorithm:Witness}) with Line~\ref{l.II:winner set} of the second \tWitness function (Algorithm~\ref{algorithm:Witness simulation}). The main difference is that in the first version we have the restriction $V^{1 \ldots t} z < V^{1 \ldots t} x$, whereas in the second version we seek for solutions~$z$ such that $V^{1 \ldots t} \cdot (z - u) \leq b|_{1 \ldots t}$. As~$b$ is the corner of the $\e$-box $B = B_V(x - u)$, those restrictions are equivalent for solutions $z \in \S$ since
\[
  V^{1 \ldots t} \cdot (z - u) \leq b|_{1 \ldots t}
  \iff V^{1 \ldots t} \cdot (z - u) < V^{1 \ldots t} \cdot (x - u)
  \iff V^{1 \ldots t} z < V^{1 \ldots t} x \DOT
\]
The first inequality is due to the occurrence of the $\OK$-event.

Now we prove the statements by downward induction over~$t$. Let $t = d+1$. Lemma~\ref{lemma:certificate form} yields $a^{(s)} \big|_I = x|_I$ for all $s \in [d]_0$, i.e., $\bigcup_{s=0}^d \S_J \big( a^{(s)} \big) \subseteq \S_I(x)$ because $I \subseteq I_0 \subseteq J$. Consequently, $\R'_{d+1} \subseteq \R_{d+1}$ (Claim~\ref{claim:same behavior I}). Consider an arbitrary index $s \in [(d+1)-1]_0$ for which $\C_s \neq \emptyset$. Then
\[
	x^{(s)}
	\in \C_s
	\subseteq \R_{s+1}
	\subseteq \S
\]
(see Line~\ref{l:winner}) and, thus, $x^{(s)} \in \S_J \big( a^{(s)} \big)$. Hence, $x^{(s)} \in \R'_{d+1}$ (Claim~\ref{claim:same behavior III}).

For the induction step let $t \leq d$. By the observation above we have
\begin{align*}
	\C'_t &= \SET{ z \in \R'_{t+1} \WHERE V^{1 \ldots t}z < V^{1 \ldots t} x } \cap \S_J \big( a^{(t)} \big) \quad \text{and} \cr
	\C_t &= \SET{ z \in \R_{t+1} \WHERE V^{1 \ldots t} z < V^{1 \ldots t} x } \DOT
\end{align*}
Since $\R'_{t+1} \subseteq \R_{t+1}$, we obtain $\C'_t \subseteq \C_t$. We first consider the case $\C_t = \emptyset$ which implies $\C'_t = \emptyset$ and $t \geq 1$ in accordance with Lemma~\ref{lemma:Witness} since $\C_0 \neq \emptyset$. Then
\begin{align*}
	\R'_t &= \R'_{t+1} \cap \bigcup_{s=0}^{t-1} \S_J \big( a^{(s)} \big) \quad \text{and} \cr
	\R_t &= \R_{t+1} \cap \S_{I_t}(x) \DOT
\end{align*}
According to Lemma~\ref{lemma:certificate form}, all vectors $x^{(0)}, \ldots, x^{(t-1)}$ coincide with~$x$ on the indices $i \in I_t$ as $I_t \subseteq I_0 \subseteq J$. Thus, $\bigcup_{s=0}^{t-1} \S_J \big( a^{(s)} \big) \subseteq \S_{I_t}(x)$. As $\R'_{t+1} \subseteq \R_{t+1}$ due to Claim~\ref{claim:same behavior I} of the induction hypothesis, we obtain $\R'_t \subseteq \R_t$ (Claim~\ref{claim:same behavior I}). For Claim~\ref{claim:same behavior II} nothing has to be shown here. Let $s \in [t-1]_0$ be an index for which $\C_s \neq \emptyset$. Then $x^{(s)} \in \R'_{t+1}$ by Claim~\ref{claim:same behavior III} of the induction hypothesis, $x^{(s)} \in \S_J \big( a^{(s)} \big)$, and consequently $x^{(s)} \in \R'_t$ (Claim~\ref{claim:same behavior III}).

Finally, let us consider the case $\C_t \neq \emptyset$. Claim~\ref{claim:same behavior III} of the induction hypothesis yields $x^{(t)} \in \R'_{t+1}$. Since $x^{(t)} \in \S_J \big( a^{(t)} \big)$ and $V^{1 \ldots t} x^{(t)} < V^{1 \ldots t} x$, also $x^{(t)} \in \C'_t$ and, thus, $\C'_t \neq \emptyset$. Hence, $x'^{(t)} = x^{(t)}$ as $\C'_t \subseteq \C_t$ (Claim~\ref{claim:same behavior II}). The remaining claims have only to be validated if $t \geq 1$. Then
\[
	\R'_t = \bSET{ z \in \R'_{t+1} \WHERE V^{t+1} z < V^{t+1} x^{(t)} } \cap \bigcup_{s=0}^{t-1} \S_J \big( a^{(s)} \big)
\]
because $x'^{(t)} = x^{(t)}$, and
\[
	\R_t = \bSET{ z \in \R_{t+1} \WHERE V^{t+1} z < V^{t+1} x^{(t)} } \cap \S_{I_t}(x) \DOT
\]
With the same argument used for the case $\C_t = \emptyset$ we obtain $\R'_{t+1} \cap \bigcup_{s=0}^{t-1} \S_J \big( a^{(s)} \big) \subseteq \R_{t+1} \cap \S_{I_t}(x)$ and, hence, $\R'_t \subseteq \R_t$ (Claim~\ref{claim:same behavior I}). Consider an arbitrary index $s \in [t-1]_0$ for which $\C_s \neq \emptyset$. Then
\[
	x^{(s)}
	\in \C_s
	\subseteq \R_{s+1}
	\subseteq \R_t \DOT
\]
In particular, $V^{t+1} x^{(s)} < V^{t+1} x^{(t)}$ (see Line~\ref{l:loser set}) and, hence, $V^{t+1} x^{(s)} < V^{t+1} x'^{(t)}$ because $x'^{(t)} = x^{(t)}$. Furthermore, $x^{(s)} \in \R'_{t+1}$ due to the induction hypothesis, Claim~\ref{claim:same behavior III}, and $x^{(s)} \in \S_J \big( a^{(s)} \big)$. Consequently, $x^{(s)} \in \R'_t$ (Claim~\ref{claim:same behavior III}).
\end{proof}

With the claims above Lemma~\ref{lemma:same behavior} follows immediately: Since $x^{(0)} = x$ and $\C_0 \neq \emptyset$ due to Lemma~\ref{lemma:Witness}, we obtain $x'^{(0)} = x^{(0)}$ (Claim~\ref{claim:same behavior II}). Hence, the call \Witness{V, J, A, B, u} returns $x'^{(0)} = x$. \qed 
\end{proof}

As mentioned earlier, with the shift vector~$u$ we control which information of~$V$ has to be revealed to execute the call \Witness{V, J, A, B, u}. While Lemma~\ref{lemma:same behavior} holds for every vector~$u$, we have to choose~$u$ carefully for our probabilistic analysis to work. We will see that the choice $u^\star = u^\star(J, A)$, given by
\begin{equation}
\label{eq:shift vector}
u^\star_i = \begin{cases}
  |x_i - 1| & \tIF i = i_0 \COMMA \cr
  x_i       & \tIF i \in J \setminus (i_0) \COMMA \cr
  0         & \OTHERWISE \COMMA
\end{cases}
\end{equation}
is appropriate since $x_i - u^\star_i = 0$ for all $i \in J \setminus (i_0)$ and $|x_{i_0} - u^\star_{i_0}| = 1$ (cf.\ Lemma~\ref{lemma:shifted certificate form}). Recall that~$i_0$ is the index that has been added to~$I_1$ in the definition of the $(V, I)$-certificate to obtain~$I_0$ and note that $u^\star \in \SET{ 0, \ldots, \K }^n$. Moreover, for every index $i \in J$ the value~$x_i$ is given in the last column of~$A$ (see Lemma~\ref{lemma:certificate form}). Hence, if $(I_0, A_0)$ is the $(V, I)$-certificate of~$x$, then vector~$u^\star$ can be defined when a tuple $J \supseteq I_0$ and the matrix $A = A_0|_J$ are known; we do not have to know the solution~$x$ itself.

For bounding the number of Pareto-optimal solutions consider the functions $\chi_{I_0, A, B}(V)$ parameterized by an arbitrary restricted $I$-certificate $(I_0, A)$, and an arbitrary $\e$-box $B \in \B_\e$, defined as follows: $\chi_{I_0, A, B}(V) = 1$ if the call \Witness{V, I_0, A, B, u^\star(I_0, A)} returns a solution $x' \in \S$ for which $B_V \big( x' - u^\star(I_0, A) \big) = B$, and $\chi_{I_0, A, B}(V) = 0$ otherwise.

\begin{corollary}
\label{corollary:realization bound}
Assume that the $\OK$-event occurs. Then the number $\PO(V)$ of Pareto-optimal solutions is at most
\[
	\sum \limits_{(I_0, A) \in \CS(I)} \sum \limits_{B \in \B_\e} \chi_{I_0, A, B}(V) \DOT
\]
\end{corollary}

\begin{proof}
Let~$x$ be a Pareto-optimal solution, let $(I_0, A)$ be the restricted $(V, I)$-certificate of~$x$, and let $B = B_V \big( x - u^\star(I_0, A) \big) \in \B_\e$. Due to Lemma~\ref{lemma:same behavior}, \Witness{V, I_0, A, B, u^\star(I_0, A)} returns vector~$x$. Hence, $\chi_{I_0, A, B}(V) = 1$. It remains to show that the assignment $x \mapsto (I_0, A, B)$ given in the previous lines is injective. Otherwise we would count the occurence of two distinct Pareto-optimal solutions~$x_1$ and~$x_2$ only once in the sum stated in Corollary~\ref{corollary:realization bound}.

Let~$x_1$ and~$x_2$ be distinct Pareto-optimal solutions and let $(I_0^{(1)}, A_1)$ and $(I_0^{(2)}, A_2)$ be the restricted $(V, I)$-certificates of~$x_1$ and~$x_2$, respectively. If $(I_0^{(1)}, A_1) \neq (I_0^{(2)}, A_2)$, then~$x_1$ and~$x_2$ are mapped to distinct triplets. Otherwise, $u^\star(I_0^{(1)}, A_1) = u^\star(I_0^{(2)}, A_2)$ and, hence, $B_V \big( x_1 - u^\star(I_0^{(1)}, A_1) \big) \neq B_V \big( x_2 - u^\star(I_0^{(2)}, A_2) \big)$ because of the $\OK$-event and $x_1 \neq x_2$. Consequently, also in this case~$x_1$ and~$x_2$ are mapped to distinct triplets.
\end{proof}

Corollary~\ref{corollary:realization bound} immediately implies a bound on
the expected number of Pareto-optimal solutions.

\begin{corollary}
\label{corollary:expectation bound}
The expected number of Pareto-optimal solutions is bounded by
\[
  \Ex[V]{\PO(V)} \leq \sum_{(I_0, A) \in \CS(I)} \sum_{B \in \B_\e} \Pr[V]{E_{I_0, A, B}} + (\K+1)^n \cdot \Pr[V]{\overline{\OK(V)}} \COMMA
\]
where $E_{I_0, A, B}$ denotes the event that the call \Witness{V, I_0, A, B, u^\star(I_0, A)} returns a vector~$x'$ such that $B_V \big( x' - u^\star(I_0, A) \big) = B$.
\end{corollary}

\begin{proof}
By applying Corollary~\ref{corollary:realization bound}, we obtain
\begin{align*}
&\Ex[V]{\PO(V)} \cr
  &= \Ex[V]{\PO(V) \,\big|\, \OK(V)} \cdot \Pr[V]{\OK(V)} + \Ex[V]{\PO(V) \,\big|\, \overline{\OK(V)}} \cdot \Pr[V]{\overline{\OK(V)}} \cr
  &\leq \Ex[V]{\left. \sum_{(I_0, A) \in \CS(I)} \sum_{B \in \B_\e} \chi_{I_0, A, B}(V) \right| \OK(V)} \cdot \Pr[V]{\OK(V)} + |\S| \cdot \Pr[V]{\overline{\OK(V)}} \cr
  &\leq \Ex[V]{\sum_{(I_0, A) \in \CS(I)} \sum_{B \in \B_\e} \chi_{I_0, A, B}(V)} + (\K+1)^n \cdot \Pr[V]{\overline{\OK(V)}} \cr
  &= \sum_{(I_0, A) \in \CS(I)} \sum_{B \in \B_\e} \Pr[V]{E_{I_0, A, B}} + (\K+1)^n \cdot \Pr[V]{\overline{\OK(V)}} \DOT \qedhere
\end{align*}
\end{proof}

We will see that the first term of the sum in
Corollary~\ref{corollary:expectation bound} can be bounded independently
of~$\e$ and that the limit of the second term tends to~$0$ for $\e \to 0$.
First of all, we analyze the size of the restricted certificate space.

\begin{lemma}
\label{lemma:certificate space}
The size of the restricted certificate space $\CS(I)$ for $I = ()$ is bounded by
\[
	|\CS(I)| \leq (\K+1)^{(d+1)^2} n^d \DOT
\]
\end{lemma}

\begin{proof}
Exactly~$d$ indices $i_1, \ldots, i_d$ are created during the execution of the call \Witness{V, x, I} if the $\OK$-event occurs and if~$x$ is Pareto-optimal with respect to~$V$. The index~$i_0$ is determined deterministically depending on the indices $i_1, \ldots, i_d$. Matrix~$A$ of every restricted $I$-certificate $(I_0, A)$ is a $(d+1) \times (d+1)$-matrix with entries from $\SET{ 0, \ldots, \K }$. Hence, the number of possible restricted $I$-certificates is bounded by $(\K+1)^{(d+1)^2} n^d$.
\end{proof}

Let us now fix an arbitrary $I$-certificate $(I_0, A_0)$, a tuple $J \supseteq I_0$, and an $\e$-box $B \in \B_\e$. We want to analyze the probability $\Pr[V]{E_{J, A, B}}$ where $A = A_0|_J$. By~$V_J$ and~$V_{\overline{J}}$ we denote the part of the matrix $V^{1 \ldots d}$ that belongs to the indices $i \in J$ and to the indices $i \notin J$, respectively. We apply the principle of deferred decisions and assume that~$V_{\overline{J}}$ is fixed as well, i.e., we will only exploit the randomness of~$V_J$.

As motivated above, the call \Witness{V, J, A, B, u} can be executed without the full knowledge of~$V_J$. To formalize this, we introduce matrices~$Q_k$ that describe the linear combinations of~$V^k_J$ that suffice to be known:
\begin{equation}
\label{eq:linear combinations}
  Q_k = \big[ p^{(d)}, \ldots, p^{(k)}, p^{(k-2)} - p^{(k-1)}, \ldots,  p^{(0)} - p^{(k-1)} \big] \in \SET{ -\K, \ldots, \K }^{|J| \times d}
\end{equation}
for $p^{(t)} = p^{(t)}(J, A, u) = a^{(t)} - u|_J$ where~$a^{(t)}$ are the columns of matrix $A = \big[ a^{(d)}, \ldots, a^{(0)} \big]$ and $t \in [d]_0$. Note that the matrices $Q_k = Q_k(J, A, u)$ depend on the pair $(J, A)$ and on the vector~$u$.

\begin{lemma}
\label{lemma:output determination}
Let $u \in \SET{ 0, \ldots, \K }^n$ be an arbitrary shift vector and let~$U$ and~$W$ be two realizations of~$V$ such that $U_{\overline{J}} = W_{\overline{J}}$ and $U^k_J \cdot q = W^k_J \cdot q$ for all indices $k \in [d]$ and all columns~$q$ of the matrix $Q_k(J, A, u)$. Then the calls \Witness{U, J, A, B, u} and \Witness{W, J, A, B, u} return the same result.
\end{lemma}

Lemma~\ref{lemma:output determination} states that for different realizations~$U_J$ and~$W_J$ of~$V_J$ the modified \tWitness function outputs the same result. Actually, in the proof we will even see that the complete execution of both calls is identical. This means that solution~$x$ is already determined if these realizations are known. However, there is still randomness left in the objective values $V^1 x, \ldots, V^d x$ which allows us to bound the probability that~$x$ falls into box~$B$ (see Corollary~\ref{corollary:success unlikely}).

\begin{proof}
We fix an index $k \in [d]$ and analyze which information of~$V^k_J$ is required for the execution of the call \Witness{V, J, A, B, u}. For the execution of Line~\ref{l.II:winner set} we need to know $V^k \cdot (z - u)$ for solutions $z \in \S_J \big( a^{(t)} \big)$ in all iterations $t \geq k$. Since we assume~$V^k_{\overline{J}}$ to be known, this means that
\[
	V^k_J \cdot \big( z|_J - u|_J \big)
	= V^k_J \cdot \big( a^{(t)} - u|_J \big)
	= V^k_J \cdot p^{(t)}
\]
must be revealed. For $t \geq k$ vector~$p^{(t)}$ is a column of~$Q_k$. The execution of Line~\ref{l.II:winner} does not require further information about~$V^k_J$: The only iteration where we might need information about~$V^k_J$ is iteration $t=k-1$. However, as $\C_t \subseteq \S_J \big( a^{(t)} \big)$, we obtain
\[
  x^{(t)}
  = \argmin \bSET{ V^{t+1} z \WHERE z \in \C_t }
  = \argmin \bSET{ V^{t+1}_{\overline{J}} z|_{\overline{J}} \WHERE z \in \C_t }
\]
because all solutions $z \in \C_t$ agree on the entries with indices $i \in J$. Since $V^{t+1}_{\overline{J}} = V^k_{\overline{J}}$ is known, $x^{(t)}$ can be determined without any further information. Note that this does not imply that $V^{t+1} x^{(t)}$ is already specified.

It remains to consider Line~\ref{l.II:loser set}. Only in iteration $t = k-1$ we need information about~$V^k$. In that iteration it suffices to know $V^k_J \cdot \big( z|_J - x^{(t)}|_J \big)$ for every solution $z \in \bigcup_{s=0}^{t-1} \S_J \big( a^{(s)} \big)$. Hence, for $z \in \S_J \big( a^{(s)} \big)$, $s \in [t-1]_0 = [k-2]_0$, the linear combinations
\[
  V^k_J \cdot \big( z|_J - x^{(t)}|_J \big)
  = V^k_J \cdot \big( \big( a^{(s)} - u|_J \big) - \big( a^{(k-1)} - u|_J \big) \big)
  = V^k_J \cdot \big( p^{(s)} - p^{(k-1)} \big)
\]
must be revealed. For $s \in [k-2]_0$, vector $p^{(s)} - p^{(k-1)}$ is a column of~$Q_k$.

As~$U$ and~$W$ agree on all necessary information, both calls return the same result.
\end{proof}

We will now see why $u^\star = u^\star(J, A)$ defined in Equation~\eqref{eq:shift vector} is a good shift vector.

\begin{lemma}\newcommand{\AST}{`$*$'}\newcommand{\PLUS}{`$+$'}
\label{lemma:shifted certificate form}
Let $Q = \big[ \hat{p}^{(d)}, \ldots, \hat{p}^{(0)} \big]$ where $\hat{p}^{(t)} = p^{(t)} \big( J, A, u^\star(J, A) \big) \big|_{I_0}$. Then
\[
  |Q| = \begin{bmatrix}
    0 & \ldots & 0 & + & * & \ldots & * \cr
    \vdots & & \vdots & 0 & \ddots & \ddots & \vdots \cr
    \vdots & & \vdots & \vdots & \ddots & + & * \cr
    0 & \ldots & 0 & 0 & \ldots & 0 & 1
  \end{bmatrix}^\T \in \SET{ 0, \ldots, \K }^{(|I|+d+1) \times (d+1)} \COMMA
\]
where~$|Q|$ denotes the matrix~$Q'$ for which $q'_{ij} = |q_{ij}|$. Each \AST-entry can be an arbitrary value from $\SET{ 0, \ldots, \K }$ and each \PLUS-entry can be an arbitrary value from $\SET{ 1, \ldots, \K }$. Different \AST-entries as well as different \PLUS-entries can represent different values.
\end{lemma}

\begin{proof}\newcommand{\CR}{\\[0.15em]}
Let $I_0 = (j_1, \ldots, j_{|I|+d+1})$, i.e., $i_0 = j_{|I|+d+1}$. According to Lemma~\ref{lemma:certificate form} and the construction of vector~$u^\star$ in Equation~\eqref{eq:shift vector} we obtain
\[
  Q = \begin{bmatrix}
    x_{j_1}                  & \ldots        & \ldots                   & x_{j_1}         \cr
    \vdots                   &               &                          & \vdots          \cr
    x_{j_{|I|}}              & \ldots        & \ldots                   & x_{j_{|I|}}     \cr
    \overline{x_{j_{|I|+1}}} & x_{j_{|I|+1}} & \ldots                   & x_{j_{|I|+1}}   \cr
    *                        & \ddots        & \ddots                   & \vdots          \cr
    \vdots                   & \ddots        & \overline{x_{j_{|I|+d}}} & x_{j_{|I|+d}}   \cr
    *                        & \ldots        & *                        & x_{j_{|I|+d+1}}
  \end{bmatrix}
  -
  \begin{bmatrix}
    x_{j_1}             & \ldots & x_{j_1}             \cr
    \vdots              &        & \vdots              \cr
    x_{j_{|I|}}         & \ldots & x_{j_{|I|}}         \cr
    x_{j_{|I|+1}}       & \ldots & x_{j_{|I|+1}}       \CR
    \vdots              &        & \vdots              \cr
    x_{j_{|I|+d}}       & \ldots & x_{j_{|I|+d}}       \CR
    |x_{j_{|I|+d+1}}-1| & \ldots & |x_{j_{|I|+d+1}}-1|
  \end{bmatrix} \DOT
\]
The claim follows since $|a-b| \leq \K$, $\overline{a} - a \neq 0$, and $\big| a - |a-1| \big| = 1$ for all $a, b \in \SET{ 0, \ldots, \K }$.
\end{proof}

\begin{lemma}
\label{lemma:linear independence}
For all $k \in [d]$ the columns of the matrix $Q_k \big( J, A, u^\star(J, A) \big)$ and the vector~$p^{(0)}$ are linearly independent.
\end{lemma}

\begin{proof}
Let $\hat{p}^{(t)} = p^{(t)} \big|_{I_0}$ for all $t \in [d]_0$. It suffices to show that the columns of the submatrix $\hat{Q}_k = Q_k \big|_{I_0}$ and the vector~$\hat{p}^{(0)}$ are linearly independent. Consider the matrix $Q = \big[ \hat{p}^{(d)}, \ldots, \hat{p}^{(0)} \big]$. Due to Lemma~\ref{lemma:shifted certificate form} the last~$d+1$ rows of~$Q$ form a lower triangular matrix and the entries on the principal diagonal are from the set $\SET{ -\K, \ldots, \K } \setminus \SET{0}$. Consequently, the vectors $\hat{p}^{(t)}$ are linearly independent. As these vectors are the same as the columns of matrix~$\hat{Q}_1$ plus vector~$\hat{p}^{(0)}$ (see Equation~\ref{eq:linear combinations}), the claim holds for $k = 1$. Now let $k \geq 2$. We consider an arbitrary linear combination of the columns of~$\hat{Q}_k$ and the vector~$\hat{p}^{(0)}$ and show that it is~$0$ if and only if all coefficients are~$0$.
\begin{align*}
0
&= \sum_{t=k}^d \lambda_t \cdot \hat{p}^{(t)} + \sum_{t=0}^{k-2} \lambda_t \cdot \big( \hat{p}^{(t)} - \hat{p}^{(k-1)} \big) + \mu \cdot \hat{p}^{(0)} \cr
&= \sum_{t=k}^d \lambda_t \cdot \hat{p}^{(t)} + \sum_{t=1}^{k-2} \lambda_t \cdot \hat{p}^{(t)} - \left( \sum_{t=0}^{k-2} \lambda_t \right) \cdot \hat{p}^{(k-1)} + (\lambda_0 + \mu) \cdot \hat{p}^{(0)} \DOT
\end{align*}
As the vectors~$\hat{p}^{(t)}$ are linearly independent, we first get $\lambda_t = 0$ for $t \in [d] \setminus \SET{ k-1 }$, which yields $\lambda_0 = 0$ due to $\sum_{t=0}^{k-2} \lambda_t = 0$ and, finally, $\mu = 0$ because of $\lambda_0 + \mu = 0$. This concludes the proof.
\end{proof}

\begin{corollary}
\label{corollary:success unlikely}
Let $\gamma = d(d+1)$. For an arbitrary restricted $I$-certificate $(I_0, A)$ the probability of the event $E_{I_0, A, B}$ is bounded by
\[
  \Pr[V]{E_{I_0, A, B}} \leq (2\gamma\K)^{\gamma-d} \phi^\gamma \e^d
\]
and by
\[
  \Pr[V]{E_{I_0, A, B}} \leq 2^d (\gamma\K)^{\gamma-d} \phi^d \e^d
\]
if all densities are quasiconcave.
\end{corollary}

\begin{proof}
Event $E_{I_0, A, B}$ occurs if the output of the call \Witness{V, I_0, A, B, u^\star(I_0, A)} is a vector~$x'$ for which $B_V \big( x'-u^\star(I_0, A) \big) = B$. We apply the principle of deferred decisions and assume that $V|_{\overline{I_0}}$ is fixed arbitrarily. Now let us further assume that the linear combinations of~$V^k_{I_0}$ given by the columns of matrix $Q_k = Q_k(I_0, A, u^\star(I_0, A))$ are known for all $k \in [d]$. This means that for some fixed values we consider all realizations of~$V$ for which the linear combinations of~$V^k_{I_0}$ given by the columns of~$Q_k$ equal these values. In accordance with Lemma~\ref{lemma:output determination}, vector~$x'$ is therefore already determined, i.e., it is the same for all realizations of~$V$ that are still under consideration.

The equality $B_V \big( x'-u^\star(I_0, A) \big) = B$ holds if and only if
\[
  V^k \cdot \big( x' - u^\star(I_0, A) \big)
  = V^k_{\overline{I_0}} \cdot \big( x' - u^\star(I_0, A) \big) \big|_{\overline{I_0}} + V^k_{I_0} \cdot \big( x' - u^\star(I_0, A) \big) \big|_{I_0} \in (b_k, b_k + \e]
\]
holds for all $k \in [d]$, where $b = (b_1, \ldots, b_d)$ is the corner of~$B$. Since
\[
  \big( x' - u^\star(I_0, A) \big) \big|_{I_0} = a^{(0)} - u^\star(I_0, A)|_{I_0} = p^{(0)}
\]
for the vector $p^{(0)} = p^{(0)} \big( I_0, A, u^\star(I_0, A) \big)$, this is equivalent to the event that
\[
  V^k_{I_0} \cdot p^{(0)} \in (b_k, b_k + \e] - V^k_{\overline{I_0}} \cdot \big( x' - u^\star(I_0, A) \big) \big|_{\overline{I_0}} \FED C_k \COMMA
\]
where~$C_k$ is an interval of length~$\e$ depending on~$x'$ and hence on the linear combinations of~$V_{I_0}$ given by the matrices~$Q_k$. By~$C$ we denote the $d$-dimensional hypercube $C = \prod_{k=1}^d C_k$ with side length~$\e$ defined by the intervals~$C_k$.

For all $k \in [d]$ let $Q'_k \in \SET{ -\K, \ldots, \K}^{|I_0| \times (d+1)}$ be the matrix consisting of the columns of~$Q_k$ and the vector~$p^{(0)}$. These matrices form the diagonal blocks of the matrix
\[
  Q' = \begin{bmatrix}
    Q'_1 & \mathbb{O} & \ldots & \mathbb{O} \cr
    \mathbb{O} & \ddots & \ddots & \vdots \cr
    \vdots & \ddots & \ddots & \mathbb{O} \cr
    \mathbb{O} & \ldots & \mathbb{O} & Q'_d
  \end{bmatrix} \in \SET{ -\K, \ldots, \K }^{d \cdot (|I|+d+1) \times d \cdot (d+1)} \DOT
\]
Lemma~\ref{lemma:linear independence}, applied for $J = I_0$, implies that matrix~$Q'$ has full rank. We permute the columns of~$Q'$ to obtain a matrix~$Q$ whose last~$d$ columns belong to the last column of one of the matrices~$Q'_k$. This means that the last~$d$ columns of~$Q'$ are $(p^{(0)}, 0^{|I_0|}, \ldots, 0^{|I_0|}), \ldots, (0^{|I_0|}, \ldots, 0^{|I_0|}, p^{(0)})$. Let the rows of~$Q$ be labeled by $Q_{j_1, 1}, \ldots, Q_{j_m, 1}, \ldots, Q_{j_1, d}, \ldots, Q_{j_m, d}$ assuming that $I_0 = (j_1, \ldots, j_m)$. We introduce random variables $X_{j,k} = V^k_j$, $j \in I_0$, $k \in [d]$, labeled in the same fashion as the rows of~$Q$. Event $E_{I_0, A, B}$ holds if and only if the~$d$ linear combinations of the variables~$X_{j,k}$ given by the last~$d$ columns of~$Q$ fall into the $d$-dimensional hypercube~$C$ depending on the linear combinations of the variables~$X_{j,k}$ given by the remaining columns of~$Q$. The claim follows by applying Theorem~\ref{theorem.Prob:enough randomness} for the matrix~$Q^\T$ and $k=d$ and due to the fact that the number of columns of~$Q$ is $\gamma = d \cdot (d+1)$. Hence,
\[
  \Pr[V]{E_{I_0, A, B}}
  \leq (2\gamma\K)^{\gamma-d} \phi^\gamma \e^d
\]
in general and
\[
  \Pr[V]{E_{I_0, A, B}}
  \leq 2^d (\gamma\K)^{\gamma-d} \phi^d \e^d
\]
if all densities are quasiconcave. The different bounds for general densities and quasiconcave densities come solely from Theorem~\ref{theorem.Prob:enough randomness}.
\end{proof}

\begin{proof}[Proof of Theorem~\ref{thm:MainFirstMoment}]
We begin the proof by showing that the $\OK$-event is likely to happen.
For all indices $t \in [d]$ and all solutions $x \neq y \in \S$ the probability
that $\big| V^t x - V^t y \big| \leq \e$ is bounded by $2\phi\e$. To see this, choose one index $i \in [n]$ such that $x_i \neq y_i$ and apply the principle of deferred decisions by fixing all coefficients~$V^t_j$ for $j \neq i$. Then the value~$V^t_i$ must fall into an interval of length $2\e/|x_i-y_i| \leq 2\e$. The probability for this is bounded from above by $2\e \phi$. A union bound over all indices $t \in [d]$ and over all pairs $(x, y) \in \S \times \S$ for which $x \neq y$ yields $\Pr[V]{\overline{\OK(V)}} \leq 2(\K+1)^{2n} d \phi \e$.

Let $\gamma = d \cdot (d+1)$. We set
\[
	s = \begin{cases}
		(2\gamma\K)^{\gamma-d} \phi^\gamma & \text{for general density functions} \COMMA \cr
		2^d (\gamma\K)^{\gamma-d} \phi^d   & \text{for quasiconcave density functions} \DOT
	\end{cases}
\]
With $I = ()$ we obtain
\begin{align*}
  \Ex[V]{\PO(V)}
  &\leq \sum_{(I_0, A) \in \CS(I)} \sum_{B \in \B_\e} \Pr[V]{E_{I_0, A, B}} + (\K+1)^n \cdot \Pr[V]{\overline{\OK(V)}} \cr
  &\leq \sum_{(I_0, A) \in \CS(I)} \sum_{B \in \B_\e} s \cdot \e^d  + (\K+1)^n \cdot 2(\K+1)^{2n} d \phi \e \cr
  &= |\CS(I)| \cdot |\B_\e| \cdot s \cdot \e^d  + 2(\K+1)^{3n} d \phi \e \cr
  &\leq (\K+1)^{(d+1)^2} n^d \cdot \left( \frac{2n\K}{\e} \right)^d \cdot s \cdot \e^d  + 2(\K+1)^{3n} d \phi \e \cr
  &= 2^d (\K+1)^{(d+1)^2} \K^d n^{2d} \cdot s  + 2(\K+1)^{3n} d \phi \e \DOT
\end{align*}
The first inequality is due to Corollary~\ref{corollary:expectation bound}.
The second inequality is due to
Corollary~\ref{corollary:success unlikely}. The third inequality stems from Lemma~\ref{lemma:certificate space}. Since this bound is true for every $\e > 0$ for which $1/\e$ is integral, it also holds for the limit $\e \to 0$. Hence, we obtain
\[
  \Ex[V]{\PO(V)} \leq 2^d (\K+1)^{(d+1)^2} \K^d n^{2d} \cdot s \DOT
\]
Substituting~$s$ and~$\gamma$ by their definitions yields
\begin{align*}
  \Ex[V]{\PO(V)}
  &\leq 2^d (\K+1)^{(d+1)^2} \K^d n^{2d} \cdot (2d(d+1)\K)^{d(d+1)-d} \phi^{d(d+1)} \cr
  &= \K^{2(d+1)^2} \cdot O \big( n^{2d} \phi^{d(d+1)} \big)
\end{align*}
for general densities and
\begin{align*}
  \Ex[V]{\PO(V)}
  &\leq 2^d (\K+1)^{(d+1)^2} \K^d n^{2d} 2^d (d(d+1)\K)^{d(d+1)-d} \phi^d \cr
  &= \K^{2(d+1)^2} \cdot O \big( n^{2d} \phi^d \big)
\end{align*}
for quasiconcave densities.
\end{proof}

\subsection{Higher Moments}

The basic idea behind our analysis of higher moments is the following: If the $\OK$-event occurs, then we can count the $c\th$ power of the number $\PO(V)$ of Pareto-optimal solutions by counting all $c$-tuples $(B_1, \ldots, B_c)$ of $\e$-boxes where each $\e$-box~$B_i$ contains a Pareto-optimal solution~$x_i$. We can bound this value as follows: First, call \Witness{V, x_1, ()} to obtain a vector~$x'_1$ and consider the index tuple~$I_0^{(1)}$ that contains all indices created in this call and one additional index. In the second step, call \Witness{V, x_2, I_0^{(1)}} to obtain a vector~$x'_2$ and consider the tuple~$I_0^{(2)}$ consisting of the indices of~$I_0^{(1)}$, the indices created in this call, and one additional index. Now, call \Witness{V, x_3, I_0^{(2)}} and so on. For the call \Witness{V, x, I} to be well-defined, in Section~\ref{subsec:FirstMoment} we assumed $|I| \leq n - (d+1)$. Consequently, here we have to ensure that $|I_0^{(c-1)}| \leq n - (d+1)$, i.e., $n \geq c \cdot (d+1)$. We can assume this for fixed integers~$c$ and~$d$ because all of our results are presented in $O$-notation.

If $(x_1, \ldots, x_c)$ is a tuple of Pareto-optimal solutions with $V^{1\ldots d}x_i\in B_i$ for $i\in[c]$, then $(x'_1, \ldots, x'_c) = (x_1, \ldots, x_c)$ due to Lemma~\ref{lemma:Witness}. As in the analysis of the first moment, we use the variant of the \tWitness\ function that uses certificates of the vectors~$x_\ell$ instead of the vectors itself to simulate the calls. Hence we can reuse several statements of Section~\ref{subsec:FirstMoment}.

Let us remark that for bounding the $c\th$ moment of the smoothed number of
Pareto-optimal solutions we also have to consider $c$-tuples $(B_1, \ldots, B_c)$ of $\e$-boxes
for which $B_k = B_\ell$ for some indices $k < \ell$. This might seem critical as both boxes contain
the same Pareto-optimal solution $x_k = x_\ell$ (if such a solution exists) which could cause 
problems due to dependencies. We resolve this problem by using different shift vectors~$u_k$ 
and~$u_\ell$ for reconstructing the vectors~$x_k$ and~$x_\ell$ with the \tWitness function.

Unless stated otherwise, let~$V$ be a realization such that the $\OK$-event $\OK(V)$ occurs and fix arbitrary
solutions $x_1, \ldots, x_c \in \S$ with $V^{1\ldots d}x_i\in B_i$ for $i\in[c]$
that are Pareto-optimal with respect to~$V$.

\begin{definition}
\label{definition.High:certificate}
Let $I_0^{(0)} = ()$ and let $(I_0^{(\ell)}, A_0^{(\ell)})$ be the $(V, I_0^{(\ell-1)})$-certificate of~$x_\ell$ defined in Definition~\ref{definition:certificate}, $\ell = 1, \ldots, c$. We call the pair $(I, A)$ for $I = I_0^{(c)}$, $A = (A^{(1)}, \ldots, A^{(c)})$, and $A^{(\ell)} = A_0^{(\ell)} \big|_I$, the (\emph{restricted}) \emph{$V$-certificate} of $(x_1, \ldots, x_c)$. We call a pair $(I'_0, A')$ a \emph{$c$-certificate}, if there is a realization~$V$ such that the $\OK$-event occurs and if there are Pareto-optimal solutions $x_1, \ldots, x_c \in \S$ such that $(I'_0, A')$ is the $V$-certificate of $(x_1, \ldots, x_c)$. By~$\CS_c$ we denote the set of all $c$-certificates.
\end{definition}

Note that $I_0^{(0)} \subseteq \ldots \subseteq I_0^{(c)}$ and $|I_0^{(\ell)}| = |I_0^{(\ell-1)}| + d+1$ for $\ell \in [c]$.

We now consider the functions $\chi_{I, A, \vec{B}}(V)$, parameterized by an arbitrary $c$-certificate $(I, A) \in \CS_c$ and a vector $\vec{B} = (B_1, \ldots, B_c) \in \B_\e^c$ of $\e$-boxes, which is defined as follows: $\chi_{I, A, \vec{B}}(V) = 1$ if for all $\ell \in [c]$ the call \Witness{V, I, A^{(\ell)}, B_\ell, u^\star(I, A^{(\ell)})} returns a solutions~$x'_\ell$ such that $B_V \big( x'_\ell - u^\star(I, A^{(\ell)}) \big) = B_\ell$, and $\chi_{I, A, \vec{B}}(V) = 0$ otherwise. Recall that the vector $u^\star = u^\star(I, A^{(\ell)})$ is defined in Equation~\ref{eq:shift vector}.

\begin{corollary}
\label{corollary.High:realization bound}
Assume that the $\OK$-event occurs. Then the $c\th$ power of the number $\PO(V)$ of Pareto-optimal solutions is at most
\[
	\sum \limits_{(I, A) \in \CS_c} \sum \limits_{\vec{B} \in \B_\e^c} \chi_{I, A, \vec{B}}(V) \DOT
\]
\end{corollary}

\begin{proof}
The $c\th$ power of the number $\PO(V)$ of Pareto-optimal solutions equals the number of $c$-tuples $(x_1, \ldots, x_c)$ of Pareto-optimal solutions. Let $(x_1, \ldots, x_c)$ be such a $c$-tuple, let $(I, A)$ be the $V$-certificate of $(x_1, \ldots, x_c)$, and let $B_\ell = B_V \big( x_\ell - u^\star(I, A^{(\ell)}) \big) \in \B_\e$. Due to Lemma~\ref{lemma:same behavior}, \Witness{V, I, A^{(\ell)}, B_\ell, u^\star(I, A^{(\ell)})} returns vector~$x_\ell$ for all $\ell \in [c]$. Hence, $\chi_{I, A, \vec{B}}(V) = 1$ for $\vec{B} = (B_1, \ldots, B_c)$. As in the proof of Corollary~\ref{corollary:realization bound} we have to show that this assignment $(x_1, \ldots, x_c) \mapsto (I, A, \vec{B})$ is injective.

Let $(x_1, \ldots, x_c)$ and $(y_1, \ldots, y_c)$ be distinct $c$-tuples of Pareto-optimal solutions, i.e., there is an index $\ell \in [c]$ such that $x_\ell \neq y_\ell$, and let $(I_1, A_1)$ and $(I_2, A_2)$ be their $V$-certificates. If $(I_1, A_1) \neq (I_2, A_2)$, then both tuples are mapped to distinct triplets. Otherwise, $u^\star(I_1, A_1^{(\ell)}) = u^\star(I_2, A_2^{(\ell)})$ and, thus, $B_V \big( x_\ell - u^\star(I_1, A_1^{(\ell)}) \big) \neq B_V \big( y_\ell - u^\star(I_2, A_2^{(\ell)}) \big)$ because of the $\OK$-event and $x_\ell \neq y_\ell$. Consequently, also in this case $(x_1, \ldots, x_c)$ and $(y_1, \ldots, y_c)$ are mapped to distinct triplets.
\end{proof}

Corollary~\ref{corollary.High:realization bound} immediately implies a bound on
the $c\th$ moment of the number of Pareto-optimal solutions.

\begin{corollary}
\label{corollary.High:expectation bound}
The $c\th$ moment of the number of Pareto-optimal solutions is bounded by
\[
  \Ex[V]{\PO^c(V)} \leq \sum_{(I, A) \in \CS_c} \sum_{\vec{B} \in \B_\e^c} \Pr[V]{E_{I, A, \vec{B}}} + (\K+1)^{cn} \cdot \Pr[V]{\overline{\OK(V)}} \COMMA
\]
where $E_{I, A, \vec{B}}$ denotes the event that $\chi_{I, A, \vec{B}}(V) = 1$.
\end{corollary}

We omit the proof since it is exactly the same as the one of Corollary~\ref{corollary:expectation bound}.

\begin{lemma}
\label{lemma.High:certificate space}
The size of the certificate space is bounded by
\[
	|\CS_c| \leq (\K+1)^{c^2 (d+1)^2} n^{cd} \DOT
\]
\end{lemma}

\begin{proof}
Let $(I, A)$ be an arbitrary $c$-certificate. Each matrix~$A^{(\ell)}$ is a $|I| \times (d+1)$-matrix with entries from $\SET{ 0, \ldots, \K }$. The tuple~$I$ can be written as
\[
	I = (i^{(1)}_d, \ldots, i^{(1)}_0, \ldots, i^{(c)}_d, \ldots, i^{(c)}_0) \COMMA
\]
created by~$c$ successive calls of the \tWitness\ function, where the indices~$i^{(\ell)}_0$ are chosen deterministically in Definition~\ref{definition:certificate}. Since $|I| = c \cdot (d+1)$ the claim follows.
\end{proof}

\begin{corollary}
\label{corollary.High:success unlikely}
Let $\gamma = cd(d+1)$. For an arbitrary $c$-certificate $(I, A)$ and an arbitrary vector $\vec{B} \in \B_\e^c$ of $\e$-boxes the probability of the event $E_{I, A, \vec{B}}$ is bounded by
\[
  \Pr[V]{E_{I, A, \vec{B}}} \leq (2\gamma\K)^{\gamma-cd} \phi^\gamma \e^{cd}
\]
and by
\[
  \Pr[V]{E_{I, A, \vec{B}}} \leq 2^{cd} (\gamma\K)^{\gamma-cd} \phi^{cd} \e^{cd}
\]
if all densities are quasiconcave.
\end{corollary}

\begin{proof}
For $k \in [d]$ and $\ell \in [c]$ consider the matrices $Q_k \big( I, A^{(\ell)}, u^\star_\ell \big)$ for $u^\star_\ell = u^\star(I, A^{(\ell)})$ defined in Equation~\eqref{eq:linear combinations}. Due to Lemma~\ref{lemma:output determination} the output of the call \Witness{V, I, A^{(\ell)}, B_\ell, u^\star_\ell} is determined if $V_{\overline{I}}$ and the linear combinations $V^k_{I} \cdot q$ for all indices $k \in [d]$ and all columns~$q$ of the matrix $Q^{(\ell)}_k = Q_k(I, A^{(\ell)}, u^\star_\ell)$ are given. With the same argument as in the proof of Corollary~\ref{corollary.High:success unlikely} event $E_{I, A, \vec{B}}$ occurs if and only if $V_{I} \cdot \big[ p^{(\ell, 1)}, \ldots, p^{(\ell, d)} \big]$ falls into some $d$-dimensional hypercube~$C_\ell$ with side length~$\e$ depending on the linear combinations $V_{I} \cdot Q^{(\ell)}_k$. In this notation, $p^{(\ell, t)}$ is short for $p^{(t)}(I, A^{(\ell)}, u^\star_\ell)$.

Now, consider the matrix
\[
  Q'_k = \left[ Q^{(1)}_k, p^{(1,k)}, \ldots, Q^{(c)}_k, p^{(c,k)} \right] \in \SET{ -\K, \ldots, \K }^{|I| \times c \cdot (d+1)} \DOT
\]
Note that $|I| = c \cdot (d+1) = \gamma/d$. Due to Lemma~\ref{lemma:shifted certificate form}, $Q'_k$ is a lower block triangular matrix, due to Lemma~\ref{lemma:linear independence} the columns of $\big[ Q^{(\ell)}_k, p^{(\ell, k)} \big]$ are linearly independent. Hence, matrix~$Q'_k$ is an invertible matrix and the same holds for the block diagonal matrix
\[
  Q' = \begin{bmatrix}
    Q'_1       & \mathbb{O} & \ldots     & \mathbb{O} \cr
    \mathbb{O} & \ddots     & \ddots     & \vdots     \cr
    \vdots     & \ddots     & \ddots     & \mathbb{O} \cr
    \mathbb{O} & \ldots     & \mathbb{O} & Q'_d
  \end{bmatrix}
  \in \SET{ -\K, \ldots, \K }^{\gamma \times \gamma} \DOT
\]
We permute the columns of~$Q'$ to obtain a matrix~$Q$ where the last $cd$ columns belong to the columns $p^{(1,1)}, \ldots, p^{(1,d)}, \ldots, p^{(c,1)}, \ldots, p^{(c,d)}$. We assume the rows of~$Q$ to be labeled by $Q_{j_1,1}, \ldots, Q_{j_m, 1}, \ldots, Q_{j_1, d}, \ldots, Q_{j_m, d}$ where $I = (j_1, \ldots, j_m)$ and introduce random variables $X_{j,k} = V^k_j$, $j \in I$, $k \in [d]$, indexed the same way as the rows of~$Q$. Event $E_{I, A, \vec{B}}$ holds if and only if the $cd$ linear combinations of the variables~$X_{j,k}$ given by the last $cd$ columns of~$Q$ fall into the $cd$-dimensional hypercube $C = \prod_{\ell = 1}^c C_\ell$ with side length~$\e$ depending on the linear combinations of the variables~$X_{j,k}$ given by the remaining columns of~$Q$. The claim follows by applying Theorem~\ref{theorem.Prob:enough randomness} for the matrix~$Q^\T$ and $k = cd$ and due to the fact that the number of columns of~$Q$ is~$\gamma$.
\end{proof}

\begin{proof}[Proof of Theorem~\ref{thm:HigherMoment}]
In the proof of Theorem~\ref{thm:MainFirstMoment} we showed that the probability that the $\OK$-event does not hold is bounded by $2(\K+1)^{2n}d\phi\e$. Let $\gamma = cd(d+1)$. We set
\[
	s = \begin{cases}
		(2\gamma\K)^{\gamma-cd} \phi^\gamma 		& \text{for general density functions} \COMMA \cr
		2^{cd} (\gamma\K)^{\gamma-cd} \phi^{cd} & \text{for quasiconcave density functions} \DOT
	\end{cases}
\]
Then we obtain
\begin{align*}
  \Ex[V]{\PO^c(V)}
  &\leq \sum_{(I, A) \in \CS_c} \sum_{\vec{B} \in \B_\e^c} \Pr[V]{E_{I, A, \vec{B}}} + (\K+1)^{cn} \cdot \Pr[V]{\overline{\OK(V)}} \cr
  &\leq \sum_{(I, A) \in \CS_c} \sum_{\vec{B} \in \B_\e^c} s \cdot \e^{cd} + (\K+1)^{cn} \cdot 2(\K+1)^{2n}d\phi\e \cr
  &= |\CS_c| \cdot |\B_\e^c| \cdot s \cdot \e^{cd} + 2(\K+1)^{(c+2)n} d\phi\e \cr
  &\leq (\K+1)^{c^2 (d+1)^2} n^{cd} \cdot \left( \frac{2n\K}{\e} \right)^{cd} \cdot s \cdot \e^{cd} + 2(\K+1)^{(c+2)n} d\phi\e \cr
  &= 2^{cd} (\K+1)^{c^2 (d+1)^2} \K^{cd} n^{2cd} \cdot s + 2(\K+1)^{(c+2)n} d\phi\e \DOT
\end{align*}
The first inequality is due to Corollary~\ref{corollary.High:expectation bound}.
The second inequality is due to Corollary~\ref{corollary.High:success unlikely}. The third inequality stems from Lemma~\ref{lemma.High:certificate space}. Since this bound is true for every $\e > 0$ for which $1/\e$ is integral, it also holds for the limit $\e \to 0$. Hence, we obtain
\[
  \Ex[V]{\PO^c(V)} \leq 2^{cd} (\K+1)^{c^2 (d+1)^2} \K^{cd} n^{2cd} \cdot s \DOT
\]
Substituting~$s$ and~$\gamma$ by their definitions yields
\begin{align*}
  \Ex[V]{\PO^c(V)}
  &\leq 2^{cd} (\K+1)^{c^2 (d+1)^2} \K^{cd} n^{2cd} \cdot (2cd(d+1)\K)^{cd(d+1)-cd} \phi^{cd(d+1)} \cr
  &= \K^{(c+1)^2 (d+1)^2} \cdot O \big( (n^{2d} \phi^{d(d+1)})^c \big)
\end{align*}
for general densities and
\begin{align*}
  \Ex[V]{\PO^c(V)}
  &\leq 2^{cd} (\K+1)^{c^2 (d+1)^2} \K^{cd} n^{2cd} \cdot 2^{cd} (cd(d+1)\K)^{cd(d+1)-cd} \phi^{cd} \cr
  &= \K^{(c+1)^2 (d+1)^2} \cdot O \big( (n^{2d} \phi^d)^c \big)
\end{align*}
for quasiconcave densities.
\end{proof}

The proof of Theorem~\ref{thm:HigherMoment} yields that $\Ex[V]{\PO^c(V)} \leq s_c$ for
\[
  s_c \DEF 2^{c(d+1)^2} (cd(d+1))^{cd^2} (\K+1)^{(c+1)^2(d+1)^2} n^{2cd} \phi^{c\beta} \COMMA
\]
where
\[
	\beta = \begin{cases}
		d(d+1) & \text{for general density functions} \COMMA \cr
		d      & \text{for quasiconcave density functions} \DOT
	\end{cases}
\]
With the following Corollary
we bound the probability that $\PO(V)$ exceeds a certain multiple of~$s_1$. We obtain
a significantly better concentration bound than the one we would obtain by applying Markov's inequality
for the first moment.

\begin{corollary}
The probability that the number of Pareto-optimal solutions is at least $\lambda \cdot s_1$ for some $\lambda \geq 1$ is bounded by
\[
  \Pr[V]{\PO(V) \geq \lambda \cdot s_1} \leq \left( \frac{1}{\lambda} \right)^{\frac{1}{2} \cdot \floor{\frac{\log_{\K+1} \lambda}{4(d+1)^2}}} \DOT
\]
\end{corollary}

\begin{proof}
Let~$c^\star$ be the real for which $(\K+1)^{2c^\star(d+1)^2} = \lambda^{1/2}$, i.e.,
\[
	c^\star = \frac{\log_{\K+1} \lambda}{4(d+1)^2} \DOT
\]
Observing that $c \leq 2^c \leq (\K+1)^c$ for all $c \in \RR$ and setting $c = \floor{c^\star}$ yields
\begin{align*}
  &\Pr[V]{\PO(V) \geq \lambda \cdot s_1} \cr
  &= \Pr[V]{\PO^c(V) \geq \lambda^c \cdot s_1^c}
  = \Pr[V]{\PO^c(V) \geq \frac{\lambda^c \cdot s_1^c}{\Ex[V]{\PO^c(V)}} \cdot \Ex[V]{\PO^c(V)}} \cr
  &\leq \frac{\Ex[V]{\PO^c(V)}}{\lambda^c \cdot s_1^c}
  \leq \frac{s_c}{\lambda^c \cdot s_1^c}
  = \frac{2^{c(d+1)^2} (cd(d+1))^{cd^2} (\K+1)^{(c+1)^2(d+1)^2} n^{2cd} \phi^{c\beta}}{\lambda^c \cdot 2^{c(d+1)^2} (d(d+1))^{cd^2} (\K+1)^{4c(d+1)^2} n^{2cd} \phi^{c\beta}} \cr
  &= \frac{c^{cd^2} (\K+1)^{(c-1)^2(d+1)^2}}{\lambda^c}
  \leq \left( \frac{c^{(d+1)^2} (\K+1)^{c(d+1)^2}}{\lambda} \right)^c \cr
  &\leq \left( \frac{(\K+1)^{c(d+1)^2} (\K+1)^{c(d+1)^2}}{\lambda} \right)^c
  \leq \left( \frac{(\K+1)^{2c^\star(d+1)^2}}{\lambda} \right)^c \cr
  &= \left( \frac{1}{\lambda} \right)^{\frac{c}{2}}
  = \left( \frac{1}{\lambda} \right)^{\frac{1}{2} \cdot \floor{\frac{\log_{\K+1} \lambda}{4(d+1)^2}}} \DOT
\end{align*}
The first inequality is Markov's inequality. The second inequality only holds if $c \geq 1$. However, for $c = 0$ the inequality $\Pr[V]{\PO(V) \geq \lambda \cdot s_1} \leq \lambda^{-c/2}$ is trivially true.
\end{proof}

\section{Zero-preserving Perturbations}
\label{sec:ZeroPreserving}

Our analysis of Theorem~\ref{thm:MainZeroPreserving} holds for instances with the following property: There exists a partition $(I_1, \ldots, I_d)$ of $[n]$ such that, for all $t \in [d]$ and for all $i \in [n]$, the coefficient~$V^t_i$ is not deterministically set to~$0$ if and only if $i \in I_t$. This means that exactly~$n$ of the $d \cdot n$ coefficients are perturbed and that the value $V^t x$ only depends on the entries~$x_i$ of~$x$ for which $i \in I_t$. All other objective functions do not depend on these entries. Furthermore, we require $|I_t| \geq (d+1)^3$ for all $t \in [d]$.

With the next two lemmas we show that if Theorem~\ref{thm:MainZeroPreserving} holds for instances that have the form described above, then it also holds for all other instances with a slightly larger constant that is hidden in the $O$-notation.


%

\begin{lemma}
\label{lemma.ZP:enough perturbed coefficients}
Without loss of generality in each objective function except for the adversarial one there are more than $(d+1)^3$ perturbed
coefficients, i.e., coefficients that are not deterministically set to~$0$.
\end{lemma}

\begin{proof}
For an index $k \in [d]$ let~$P_k$ be the tuple of indices~$i$ for which~$V^k_i$ is a perturbed coefficient, i.e., a coefficient
which is not set to~$0$ deterministically.
Let~$K$ be the tuple of indices~$k$ for which $|P_k| \leq (d+1)^3$, let $P = \bigcup_{k \in K} P_k$, and consider the decomposition
of~$\S$ into subsets of solutions $\S_v = \SET{ x \in \S \WHERE x|_P = v }$, $v \in \SET{ 0, \ldots, \K }^{|P|}$. Let $x \in \S_v$ be
an arbitrary solution. If~$x$ is Pareto-optimal with respect to~$\S$ and $\SET{ V^1, \ldots, V^{d+1} }$, then~$x$ is also Pareto-optimal with respect to~$\S_v$ and $\SET{ V^k \WHERE k \in [d+1] \setminus K }$ due to Lemma~\ref{Restart Lemma}.
As all remaining objective functions~$V^k$, $k \in [d+1] \setminus K$, have more than $(d+1)^3$ perturbed coefficients, the instance
with these objective functions and~$\S_v$ as set of feasible solutions has the desired form and we can apply Theorem~\ref{thm:MainZeroPreserving} for each of these instances. Since we now have $(\K+1)^{|P|} \leq (2\K)^{|K| \cdot (d+1)^3}$ instances, each having $d-|K|$ linear and one adversarial objective, we can bound the number of Pareto-optimal solutions by
\[
  (2\K)^{|K| \cdot (d+1)^3} \cdot \K^{(d-|K|+1)^5} \cdot O \left( n^{\alpha(d-|K|)} \cdot \phi^{\beta(d-|K|)} \right) \COMMA
\]
where~$\alpha$ and~$\beta$ denote the exponents of~$n$ and~$\phi$ in the bound stated in Theorem~\ref{thm:MainZeroPreserving}. These exponents depend on the number~$d$ of non-adversarial objectives and whether the densities are quasiconcave or not. Since they are monotonically increasing, which particularly implies $\alpha(d-|K|) \leq \alpha(d)$ and $\beta(d-|K|) \leq \beta(d)$, we can bound the number of Pareto-optima simply by
\[
  (2\K)^{|K| \cdot (d+1)^3} \cdot \K^{(d-|K|+1)^5} \cdot O \left( n^{\alpha(d)} \cdot \phi^{\beta(d)} \right) \DOT
\]
Hence, it suffices to show that
\[
	\K^{|K| \cdot (d+1)^3 + (d-|K|+1)^5}
	\leq \K^{(d+1)^5} \COMMA
\]
as the additional factor $2^{|K| \cdot (d+1)^3} \leq 2^{(d+1)^3 d}$ can be hidden in the $O$-notation. This inequality is equivalent to showing that $b \cdot a^3 + (a-b)^5 \leq a^5$ for $b = |K|$ and $a = d+1$. Note that $0 \leq b = |K| \leq d = a-1$. By a chain of equivalences we obtain
\begin{align*}
  b \cdot a^3 + (a-b)^5 \leq a^5
  &\iff a^3 \leq \frac{1}{b} \cdot (a^5 - (a^5 - 5a^4b + 10a^3b^2 - 10a^2b^3 + 5ab^4 - b^5) \cr
  &\iff a^3 \leq 5a^4 - 10a^3b + 10a^2b^2 - 5ab^3 + b^4 \cr
  &\qquad \qquad = 5a \cdot (a^3 - 2a^2b + 2ab^2 - b^3) + b^4 \cr
  &\qquad \qquad = 5a(a-b) \cdot (a^2 - ab + b^2) + b^4 \FED f(a, b) \DOT
\end{align*}
Applying the inequality $ab \leq (\frac{a+b}{2})^2 = \frac{(a+b)^2}{4}$ yields
\[
  f(a, b)
  \geq 5a(a-b) \cdot \left( a^2 - \frac{(a+b)^2}{4} + b^2 \right)
  \geq 5a \cdot \left( \frac{a^2}{2} + \frac{(a-b)^2}{4} + \frac{b^2}{2} \right)
  \geq \frac{5}{2}a^3
  \geq a^3 \DOT
\]
This concludes the proof.
\end{proof}

\begin{lemma}
\label{lemma.ZP:easy structure}
Without loss of generality for every $i \in [n]$ exactly one of the coefficients $V^1_i, \ldots, V^d_i$ is perturbed, whereas the others are deterministically
set to~$0$.
\end{lemma}

Let us remark that we can transform every instance that has not the form stated in Lemma~\ref{lemma.ZP:easy structure} into an instance with this form. We will show that this transformation does not increase the size of the Pareto-set for any realization of the coefficients. Hence, the bound from Theorem~\ref{thm:MainZeroPreserving} that applies for the modified instance also applies for the original instance. However, our transformation increases the dimension of the set~$\S$ from~$n$ to $d \cdot n$. Hence, we lose a factor of $d^{d^3+d^2+d}$ in the bound which we can hide in the $O$-notation since we have to apply this transformation only once.

\begin{proof}
We first show how to decrease the number of indices~$i$ for which $V^1_i, \ldots, V^d_i$ is perturbed to at most one.
For this, let
\[
	\S' = \SET{ (x, x, \ldots, x) \WHERE x \in \S }
	\subseteq \SET{ 0, \ldots, \K }^{dn}
\]
be the set of feasible solutions that contains for every $x \in \S$ the solution $x^d \in \SET{ 0, \ldots, \K }^{dn}$ that consists of~$d$ copies of~$x$. For $k \in [d]$ we define a linear objective function $W^k \colon \S' \to \RR$ in which all coefficients~$W^k_i$ with $i \notin I_k \DEF \SET{ (k-1)n + 1, \ldots, kn }$ are deterministically set to~$0$. The coefficients $W^k_{(k-1)n+1}, \ldots, W^k_{kn}$ are chosen as the coefficients $V^k_1, \ldots, V^k_n$, i.e.,
either randomly according to a density~$f^k_i$ or~$0$ deterministically. The objective function~$W^{d+1}$ maps every
solution $x^d\in\S'$ to $V^{d+1}(x)$. The instance consisting of~$\S'$ and the objective functions $W^1, \ldots, W^{d+1}$ has the
desired property that every variable appears in at most one of the objective functions $W^1, \ldots, W^d$ and it has the
same smoothed number of Pareto-optimal solutions as the instance consisting of~$\S$ and the objective functions $V^1, \ldots, V^{d+1}$.  
For every $i \in [dn]$ for which none of the coefficients $W^1_i, \ldots, W^d_i$ is perturbed we can eliminate the corresponding
variable from~$\S'$.

This shows that every $\phi$-smooth instance with $\S \subseteq \SET{ 0, \ldots, \K }^n$ can be transformed into another $\phi$-smooth instance 
with $\S \subseteq \SET{ 0, \ldots, \K }^\ell$ with $\ell \leq dn$ in which every variable appears in exactly one objective function
and that has the same smoothed number of Pareto-optimal solutions. As the bound proven in Theorem~\ref{thm:MainZeroPreserving}
depends polynomially on the number of variables, we lose only a constant factor (with respect to~$n$, $\phi$, and~$\K$) by going from $\S \subseteq \SET{ 0, \ldots, \K }^n$ to $\S' \subseteq \SET{ 0, \ldots, \K }^{dn}$. This constant is hidden in the $O$-notation.
\end{proof}

In the remainder of this chapter we focus on instances having the structure described in
Lemma~\ref{lemma.ZP:enough perturbed coefficients} and Lemma~\ref{lemma.ZP:easy structure}. 
Then $(P_1, \ldots, P_d)$ is a partition of~$[n]$, where~$P_t$ denotes the tuple of
indices~$i$ for which~$V^t_i$ is perturbed.

We consider the variant of the \tWitness function given as Algorithm~\ref{algorithm:WitnessZ}, referred to as the \tWitnessZ function, which gets as 
parameters besides the usual~$V$, $x$, and~$I$, a set $K \subseteq [d]$ of indices of objective functions and a
call number $r\in\NN$. In a call of the \tWitnessZ function only the adversarial objective function~$V^{d+1}$ and the objective functions~$V^t$ with $t\in K$ are considered.
The set of solutions is restricted to solutions that agree with~$x$ in all positions~$P_k$ with $k\notin K$.
Additionally, as in the \tWitness function, only solutions are considered that agree with~$x$ in all positions $i\in I$.
By the right choice of~$I$, we can avoid choosing an index multiple times in different calls of the \tWitnessZ function. 
The parameter~$r$ simply corresponds to the number of the current call of the \tWitnessZ function.
The \tWitnessZ function always returns some subset of~$\S$.

\IfUseAlgorithmIIeStyle
{
\LinesNumbered
\begin{algorithm*}[h!t]
  \caption{\WitnessZ{V, x, K, r, I}}
  \label{algorithm:WitnessZ}
  let~$d_r$ be the number of components of~$K$ and let~$K$ be of the form $K = (k_1, \ldots, k_{d_r})$ \;
  set $k_{d_r+1} = d+1$ \;
  set $\R^{(r)}_{d_r+1} = \S_I(x) \cap \bigcap_{k \in [d] \setminus K} \S_{P_k}(x)$\label{l.ZP:initial set} \;
  \lIf{$d_r = 0$}{\Return{$\R^{(r)}_{d_r+1}$} \;}\label{l.ZP:success}
  \For{$t = d_r, d_r-1, \ldots, 0$}
  {\label{l.ZP:for}
    set $\C^{(r)}_t = \bSET{ z \in \R^{(r)}_{t+1} \WHERE V^{k_1 \ldots k_t} z < V^{k_1 \ldots k_t} x }$\label{l.ZP:winner set} \;
    \uIf{$\C^{(r)}_t \neq \emptyset$}
    {
      set $x^{(r, t)} = \argmin \bSET{ V^{k_{t+1}} z \WHERE z \in \C^{(r)}_t }$\label{l.ZP:argmin} \;
      let $K_\EQ \subseteq K$ be the tuple of indices~$k$ for which $x^{(r, t)}|_{P_k} = x|_{P_k}$\label{l.ZP:equal indices} \;
      set $K_\NEQ = K \setminus K_\EQ$\label{l.ZP:not equal indices} \;
      \For{$k \in K$}
      {
        \uIf{$k \in K_\EQ$}
        {
          set $r_k = r$\label{l.ZP:last access} \;
        }
        \Else
        { 
          set $i_k = \min \bSET{i \in P_k \WHERE x^{(r, t)}_i \neq x_i }$\label{l.ZP:index choice} \;
          $I \leftarrow I \cup (i_k)$ \;
        }
      }
      \uIf{$K_\EQ = ()$}
      {
        set $\R^{(r)}_t = \bSET{ z \in \R^{(r)}_{t+1} \WHERE V^{k_{t+1}} z < V^{k_{t+1}} x^{(r, t)} } \cap \S_I(x)$\label{l.ZP:loser set} \;
      }
      \Else
      {
        set $t_r = t$\label{l.ZP:last iteration} \;
        \Return{\WitnessZ{V, x, K_\NEQ, r+1, I}}\label{l.ZP:restart} \;
      }
    }
    \Else
    {
      \For{$k \in K$}
     	{
        set $i_k = \min (P_k \setminus I)$\label{l.ZP:trivial winner index} \;
        $I \leftarrow I \cup (i_k)$ \;
      }
      set $x^{(r, t)}_i = \begin{cases}
                            \min (\SET{ 0, \ldots, \K } \setminus \SET{ x_i }) & \text{if}\ i \in \bSET{ i_{k_1}, \ldots, i_{k_{d_r}} } \cr
                            x_i                                                & \OTHERWISE
             \end{cases}$\label{l.ZP:trivial winner} \;
      set $\R^{(r)}_t = \R^{(r)}_{t+1} \cap \S_I(x)$\label{l.ZP:trivial loser set} \;
    }
  }
  \Return{$\emptyset$}\label{l.ZP:error} \;
\end{algorithm*}%
}
{
\begin{algorithm*}[h!t]
  \caption{\WitnessZ{V, x, K, r, I}}
  \label{algorithm:WitnessZ}
  \begin{algorithmic}[1]
    \STATE Let~$d_r$ be the number of components of~$K$ and let~$K$ be of the form $K = (k_1, \ldots, k_{d_r})$.
    \STATE Set $k_{d_r+1} = d+1$.    
    \STATE Set $\R^{(r)}_{d_r+1} = \S_I(x) \cap \bigcap_{k \in [d] \setminus K} \S_{P_k}(x)$.  \label{l.ZP:initial set} 
    \STATE \textbf{if} $d_r = 0$ \textbf{then} \textbf{return} $\R^{(r)}_{d_r+1}$ \label{l.ZP:success}
    \FOR{$t = d_r, d_r-1, \ldots, 0$} \label{l.ZP:for}
      \STATE Set $\C^{(r)}_t = \bSET{ z \in \R^{(r)}_{t+1} \WHERE V^{k_1 \ldots k_t} z < V^{k_1 \ldots k_t} x }$. \label{l.ZP:winner set}
      \IF{$\C^{(r)}_t \neq \emptyset$}
        \STATE Set $x^{(r, t)} = \argmin \bSET{ V^{k_{t+1}} z \WHERE z \in \C^{(r)}_t }$. \label{l.ZP:argmin}
        \STATE Let $K_\EQ \subseteq K$ be the tuple of indices~$k$ for which $x^{(r, t)}|_{P_k} = x|_{P_k}$. \label{l.ZP:equal indices}
        \STATE Set $K_\NEQ = K \setminus K_\EQ$. \label{l.ZP:not equal indices}
        \FOR{$k \in K$}
          \IF{$k \in K_\EQ$}
            \STATE Set $r_k = r$. \label{l.ZP:last access}
          \ELSE 
            \STATE Set $i_k = \min \bSET{i \in P_k \WHERE x^{(r, t)}_i \neq x_i }$. \label{l.ZP:index choice}
            \STATE $I \lmapsto I \cup (i_k)$
          \ENDIF
        \ENDFOR
        \IF{$K_\EQ = ()$}
          \STATE Set $\R^{(r)}_t = \bSET{ z \in \R^{(r)}_{t+1} \WHERE V^{k_{t+1}} z < V^{k_{t+1}} x^{(r, t)} } \cap \S_I(x)$. \label{l.ZP:loser set}
        \ELSE
          \STATE Set $t_r = t$. \label{l.ZP:last iteration}
          \RETURN \WitnessZ{V, x, K_\NEQ, r+1, I} \label{l.ZP:restart}
        \ENDIF
      \ELSE
        \FOR{$k \in K$}
          \STATE Set $i_k = \min (P_k \setminus I)$. \label{l.ZP:trivial winner index}
          \STATE $I \lmapsto I \cup (i_k)$.
        \ENDFOR
        \STATE Set $x^{(r, t)}_i = \begin{cases}
                \min (\SET{ 0, \ldots, \K } \setminus \SET{ x_i }) & \text{if}\ i \in \bSET{ i_{k_1}, \ldots, i_{k_{d_r}} } \COMMA \cr
                x_i                                                & \OTHERWISE \DOT
               \end{cases}$ \label{l.ZP:trivial winner}
        \STATE Set $\R^{(r)}_t = \R^{(r)}_{t+1} \cap \S_I(x)$. \label{l.ZP:trivial loser set}
      \ENDIF
    \ENDFOR
    \RETURN $\emptyset$ \label{l.ZP:error}
  \end{algorithmic}
\end{algorithm*}%
}

Let us give some remarks about the \tWitnessZ function. As a convention we set $\bigcap_{k \in ()} \S_{P_k}(x) = \S$ (cf.\ Line~\ref{l.ZP:initial set}). This is only important in the case where $K = [d]$, i.e., in the first call.

As in the \tWitness function, if iteration $t=0$ is reached in a certain call~$r$ (this does not have to be the case), then we obtain $\C^{(r)}_0 = \R^{(r)}_1$ since $V^{k_1 \ldots k_t} z < V^{k_1 \ldots k_t} x$ (see Line~\ref{l.ZP:winner set}) is no restriction for $t=0$. For the definition of~$x^{(r,t)}$ in Line~\ref{l.ZP:argmin}, ties are broken by taking the lexicographically first solution. Though we did the same in the \tWitness function, it is much more important here. In the model without zero-preserving perturbations the functions $V^1, \ldots, V^d$ are injective with probability~$1$. If this is the case, then no ties have to be broken. In the model with zero-preserving perturbations, the functions $V^1, \ldots, V^d$ can be non-injective with probability~$1$: If there are two distinct solutions $x, y \in \S$ for which $x|_{P_k} = y|_{P_k}$, then $V^k x = V^k y$.

The index~$r_k$ defined in Line~\ref{l.ZP:last access} is the number of the last call in which the objective function~$V^k$ has been considered. The index~$t_r$ defined in Line~\ref{l.ZP:last iteration} is the number of the iteration in call number~$r$ of \tWitnessZ in which the next  recursive call of \tWitnessZ was made. We will see that, if the last call of the \tWitnessZ function is the call with number~$r^\star+1$, then $r_k \in [r^\star]$ for each $k \in [d]$ and there is at least one index $k \in [d]$ for which $r_k = r^\star$, i.e., the objective function~$V^k$ has been considered until the end. Furthermore, the indices~$t_r$ are defined for $r = 1, \ldots, r^\star$. For the simulation of the \tWitness function information about the solutions~$x^{(t)}$ and the indices~$i_t$ are required. For the simulation of the \tWitnessZ function we additionally need the values~$r_k$ and~$t_r$ to know when to make a new recursive call (in iteration $t=t_r$ in the call with number~$r$) and which objectives to consider (objective~$V^k$ will be considered in the call with number~$r$ if and only if $r \leq r_k$).

In Line~\ref{l.ZP:index choice} it is always possible to find an index $i \in P_k$ on which the current vector~$x^{(r, t)}$ and~$x$ disagree because this line is only reached if $k \in K_\NEQ$, i.e., if $x^{(r, t)}|_{P_k} \neq x|_{P_k}$. In order for Line~\ref{l.ZP:trivial winner index} to be feasible, we have to guarantee that $P_k \setminus I \neq ()$.
This follows since we assumed $|P_k|>(d+1)^3 > d(d+1)$ in accordance with Lemma~\ref{lemma.ZP:enough perturbed coefficients} and because there are at most~$d$ calls of \tWitnessZ with non-empty~$K$ with at most $d+1$ iterations each, and in each iteration at most one index from~$P_k$ is added to~$I$. Note that it would be more precise to introduce the notation~$i^{(r,t)}_k$ rather than~$i_k$ (cf.\ Line~\ref{l.ZP:index choice} and Line~\ref{l.ZP:trivial winner index}). Furthermore, we could also write~$I^{(r,t)}_k$ instead of~$I$. For the sake of readability we decided to drop these additional indices and refer to index~$i_k$ and tuple~$I$ of iteration~$t$ of call~$r$ in our proofs.

Before we analyze the \tWitnessZ function, let us discuss similarities and differences to the \tWitness function. The initial calls \WitnessZ{V, x, [d], 1, I} and \Witness{V, x, I} are very similar. All objectives $V^1, \ldots, V^{d+1}$ are considered. Furthermore, $d_1 = d$ and $\R_{d+1}^{(1)} = \S_I(x)$. Line~\ref{l.ZP:success} can be ignored in this call since $d_1 = d \geq 1$. Also the loop of the \tWitnessZ function is very similar to the loop of the \tWitness function. The sets~$\C^{(r)}_t$ and $\R^{(r)}_t$ and the solution~$x^{(r,t)}$ are defined the same way as the sets~$\C_t$ and~$\R_t$ and the solution~$x^{(t)}$ in the \tWitness function.

However, there are two main differences to the \tWitness function in the body of the loop that are due to the two additional issues we have to deal with when considering zero-preserving perturbations. First it is possible that $V^k x^{(r,t)} = V^k x$ for some of the indices~$k$. This happens if $x^{(r,t)}|_{P_k} = x|_{P_k}$ (otherwise, it happens with probability~$0$) and is a fundamental issue. If we would proceed running the loop as we do it in  the \tWitness function, then we might lose the crucial property that the function returns~$\SET{x}$ if~$x$ is Pareto-optimal. The tuple~$K_\EQ$ contains the problematic indices~$k$ for which $x^{(r,t)}|_{P_k} = x|_{P_k}$. If $K_\EQ = ()$, then we proceed more or less as we did in the \tWitness function (see Line~\ref{l.ZP:index choice} and Line~\ref{l.ZP:loser set}). As discussed above, the case $K_\EQ \neq ()$ has to be treated differently. In this case we make use of Lemma~\ref{Restart Lemma} which implies that, if~$x$ is Pareto-optimal, then it is also Pareto-optimal with respect to $\bigcap_{k \in K_\EQ} \S_{P_k}(x)$ and $\SET{ V^k \WHERE k \in K_\NEQ \cup (d+1) }$ (cf.\ Line~\ref{l.ZP:restart} of the current call and Line~\ref{l.ZP:initial set} of the next call).

The second difference due to another issue with zero-preserving perturbations can be sketched as follows. In each iteration~$t$ of the \tWitness function one index~$i_t$ is chosen. Since in the model without zero-preserving perturbations all coefficients are perturbed, we can then exploit the randomness in the coefficients $V^1_{i_t}, \ldots, V^d_{i_t}$. In the model with zero-preserving perturbations under the assumption given by Lemma~\ref{lemma.ZP:easy structure}, for each index $i \in [n]$ exactly one of the coefficients $V^1_i, \ldots, V^d_i$ is perturbed while the others are~$0$. Hence, we choose one index $i_k \in P_k$ for each objective~$V^k$ to ensure that we have one perturbed coefficient $V^k_{i_k}$ per objective. These indices~$i_k$ are chosen only for $k \in K_\NEQ$ (see Line~\ref{l.ZP:index choice}). This is due to the fact that for our analysis to work we need the additional property that $x^{(r,t)}_{i_k} \neq x_{i_k}$ which is impossible for $k \in K_\EQ$ by the definition of~$K_\EQ$ and the requirement $i_k \in P_k$. However, as from now on we do not consider the objectives~$V^k$ for $k \in K_\EQ$ anymore, we do not have to choose indices~$i_k$ for $k \in K_\EQ$.

In the remainder of this section we only consider the case that~$x$ is Pareto-optimal. Unless stated otherwise, we assume that the following \emph{$\OKZ$-event} $\OKZ(V)$ occurs. This event occurs if $|V^k \cdot (y - z)| \geq \e$ for every $k \in [d]$ and for every two solutions $y, z \in \S$ for which $y|_{P_k} \neq z|_{P_k}$. We will later see that the $\OKZ$-event occurs with sufficiently high probability.

\begin{lemma}
\label{lemma.ZP:Witness}
The call \WitnessZ{V, x, [d], 1, ()} returns the set $\SET{x^{(r^\star, t_{r^\star})}} = \SET{x}$, where $r^\star = \max \SET{ r_1, \ldots, r_d }$.
\end{lemma}

Lemma~\ref{lemma.ZP:Witness} was also stated in the conference version (\cite{BrunschR12}, Lemma~25) but Claim~1 of the proof was not correct. Here, we rely on the concept of weak Pareto-optimality (see Definition~\ref{def:Pareto optimality}) and its properties (Lemma~\ref{Recursion Lemma}) since we cannot guarantee~$x$ to be Pareto-optimal in every iteration. However, the Pareto-optimality holds at the beginning of every call to the \tWitnessZ function.

\begin{proof}
Let us consider an arbitrary call \WitnessZ{V, x, K, r, I} for $K \neq ()$. First we show the following claim by induction on~$t$.

\begin{claim}
\label{claim.ZP:weakly Pareto optimal}
In every iteration~$t$ that is reached, $x$ is weakly Pareto-optimal with respect to~$\R^{(r)}_{t+1}$ and $\bSET{ V^{k_1}, \ldots, V^{k_{t+1}} }$. 
\end{claim}

\begin{proof}[Proof of Claim~\ref{claim.ZP:weakly Pareto optimal}]
To begin with, consider $t = d_r$. As~$x$ is Pareto-optimal with respect to~$\S$ and $\bSET{ V^1, \ldots, V^{d+1} }$, $x$ is also Pareto-optimal with respect to $\bigcap_{k \in [d] \setminus K} \S_{P_k}(x)$ and
\[
	\bSET{ V^k \WHERE k \in K \cup (d+1) }
	= \bSET{ V^{k_1}, \ldots, V^{k_{d_r+1}} }
\]
due to Lemma~\ref{Restart Lemma}. Consequently, $x$ is also Pareto-optimal with respect to $\R^{(r)}_{d_r+1} = \S_I(x) \cap \bigcap_{k \in [d] \setminus K} \S_{P_k}(x)$ and $\bSET{ V^{k_1}, \ldots, V^{k_{d_r+1}} }$ due to Proposition~\ref{Subset Proposition}. Note that this property is even stronger than weak Pareto-optimality. We will need this strong version in the induction step when $t=d_r-1$.

Now consider a iteration $t \leq d_r - 1$ that is reached and assume that the induction hypothesis is true for $t+1$. We consider iteration~$t+1$ where $\R^{(r)}_{t+1}$ is defined, and distinguish between two cases. If $\C^{(r)}_{t+1} = \emptyset$, then $x$ is weakly Pareto-optimal with respect to $\R^{(r)}_{t+1} = \R^{(r)}_{t+2} \cap \S_I(x)$ and $\bSET{ V^{k_1}, \ldots, V^{k_{t+1}} }$ due to the induction hypothesis, Lemma~\ref{Recursion Lemma}~(I), and Proposition~\ref{Subset Proposition}.

Let us consider the more interesting case $\C^{(r)}_{t+1} \neq \emptyset$. Since iteration~$t$ is reached, there is no call of the \tWitnessZ function in iteration~$t+1$, i.e., $K_\EQ = ()$ in iteration $t+1$. The induction hypothesis and Lemma~\ref{Recursion Lemma}~(II) yield $V^{k_{t+2}} x \leq V^{k_{t+2}} x^{(r,t+1)}$.

We will show that even $V^{k_{t+2}} x < V^{k_{t+2}} x^{(r,t+1)}$. For this, we assume to the contrary that $V^{k_{t+2}} x = V^{k_{t+2}} x^{(r,t+1)}$ and distinguish between the cases $t = d_r-1$ and $t < d_r-1$.
\begin{enumerate}

	\item If $t = d_r-1$, then we obtain $V^{k_{d_r+1}} x = V^{k_{d_r+1}} x^{(r,d_r)}$ and $V^{k_1 \ldots k_{d_r}} x^{(r,d_r)} < V^{k_1 \ldots k_{d_r}} x$ since $x^{(r,d_r)} \in \C^{(r)}_{d_r}$. Hence, $x^{(r,d_r)} \in \R^{(r)}_{d_r+1}$ dominates~$x$ with respect to $\bSET{ V^{k_1}, \ldots, V^{k_{d_r+1}} }$ which contradicts the fact that~$x$ is Pareto-optimal with respect to~$\R^{(r)}_{d_r+1}$ and $\bSET{ V^{k_1}, \ldots, V^{k_{d_r+1}} }$.
	
	\item If $t < d_r-1$, then $V^{k_{t+2}} x = V^{k_{t+2}} x^{(r,t+1)}$ implies $x|_{P_{k_{t+2}}} = x^{(r,t)}|_{P_{k_{t+2}}}$ as we assume that the $\OKZ$-event occurs. Consequently, $k_{t+2} \in K_\EQ$ in iteration~$t+1$, which contradicts the previous observation that $K_\EQ = ()$ in that iteration.
	
\end{enumerate}
This concludes the proof of Claim~\ref{claim.ZP:weakly Pareto optimal}.
\end{proof}

With Claim~\ref{claim.ZP:weakly Pareto optimal} we are now able to show that a call \WitnessZ{V, x, K, r, I} terminates without a further call to the \tWitnessZ function if and only if $K = ()$. Note that one direction is trivial.

\begin{claim}
\label{claim.ZP:restart}
Consider an arbitrary call \WitnessZ{V, x, K, r, I}. If $K \neq ()$, then this call results in another call to the \tWitnessZ function (and does not terminate in Line~\ref{l.ZP:error}).
\end{claim}

\begin{proof}[Proof of Claim~\ref{claim.ZP:restart}]
Let us assume that there is no further call to the \tWitnessZ function until iteration $t = 0$, i.e., we reach iteration $t = 0$. In accordance with Claim~\ref{claim.ZP:weakly Pareto optimal}, $x$ is weakly Pareto-optimal with respect to~$\R^{(r)}_1$ and $\SET{V^{k_1}}$. Now let us consider iteration $t = 0$. We obtain $\C^{(r)}_0 = \R^{(r)}_1$ since there are no restrictions in this iteration. Consequently, $\C^{(r)}_0 \neq \emptyset$ because $x \in \R^{(r)}_1$. The solution~$x^{(r,0)}$ minimizes~$V^{k_1}$ among all solutions from~$\C^{(r)}_0$. On the other hand, $x^{(r,0)}$ cannot dominate~$x$ strongly. Hence, $V^{k_1} x^{(r,0)} = V^{k_1} x$, i.e., $x^{(r,0)}|_{P_{k_1}} = x|_{P_{k_1}}$ as we assumed that the $\OKZ$-event occurs. Therefore, $k_1 \in K_\EQ$, i.e., $K_\EQ \neq ()$, and thus, the \tWitnessZ function is called in Line~\ref{l.ZP:restart}.
\end{proof}

According to Claim~\ref{claim.ZP:restart} there will be recursive calls until a call of the form \WitnessZ{V, x, (), r, I}. This call immediately returns the set $\R^{(r)}_{d_r+1}$ in Line~\ref{l.ZP:success}. Since $[d] \setminus () = [d]$, we obtain
\[
	\R^{(r)}_{d_r+1}
	= \S_I(x) \cap \bigcap_{k \in [d]} \S_{P_k}(x)
	= \S_{[n]}(x)
	= \SET{x} \DOT
\]
Now consider call number $r-1$ and the iteration~$t_{r-1}$ in this call in which \WitnessZ{V, x, (), r, I} has been called. In this iteration, $K_\EQ \neq ()$ since Line~\ref{l.ZP:restart} is reached. Hence, there is at least one index $k \in K_\EQ$, and for these indices, $r_k$ is set to~$r$ in Line~\ref{l.ZP:last access}. Now, as the next call is of the form \WitnessZ{V, x, (), r, I}, this implies $K_\NEQ = ()$, i.e., all values~$r_k$ for $k \in [d]$ have been set by now and, thus,
\[
	r^\star
	= \max \SET{ r_1, \ldots, r_d }
	= r-1 \COMMA
\]
i.e., the number of the call we currently consider.

Consider the solution $x^{(r^\star,t_{r^\star})}$ defined in Line~\ref{l.ZP:argmin} and let~$K$ be the tuple of call~$r^\star$. Since $K_\NEQ = ()$ this implies $K_\EQ = K$. Hence, $x^{(r^\star,t_{r^\star})}|_{P_k} = x|_{P_k}$ for all $k \in K$ by definition of~$K_\EQ$ in Line~\ref{l.ZP:equal indices}. On the other hand,
\[
  x^{(r^\star,t_{r^\star})}
  \in \C^{(r^\star)}_{t_{r^\star}}
  \subseteq \R^{(r^\star)}_{t_{r^\star}+1}
  \subseteq \R^{(r^\star)}_{d_{r^\star}+1}
  \subseteq \bigcap_{k \in [d] \setminus K} \S_{P_k}(x) \DOT
\]
The first inclusion is due to the definition of $\C^{(r^\star)}_{t_{r^\star}}$ in Line~\ref{l.ZP:winner set}. The second inclusion is due to the observation that always $\R^{(r)}_{t+1} \subseteq \R^{(r)}_{t+2} \subseteq \ldots \subseteq \R^{(r)}_{d_r+1}$ due to the construction of the sets~$\R^{(r)}_t$ in Line~\ref{l.ZP:loser set} and Line~\ref{l.ZP:trivial loser set}. The construction of $\R^{(r^\star)}_{d_{r^\star}+1}$ in Line~\ref{l.ZP:initial set} yields the third inclusion. Hence, $x^{(r^\star,t_{r^\star})}|_{P_k} = x|_{P_k}$ for all $k \in [d] \setminus K$, and according to the previous observations, even for all $k \in K$. Consequently, $x^{(r^\star,t_{r^\star})} = x$. Summarizing the previous results, we obtain that the call \WitnessZ{V, x, [d], 1, ()} ends up in the call \WitnessZ{V, x, (), r, I} for some index tuple~$I$ which immediately returns the set $\R^{(r)}_{d_r+1} = \SET{x} = \SET{x^{(r^\star,t_{r^\star})}}$. \qed
\end{proof}

Like for the simple \tWitness function, we show that it is enough to know some information about the run of the
\tWitnessZ function to reconstruct the solution~$x$. As before, we call this data the \emph{certificate} of~$x$.

\begin{definition}
\label{definition.ZP:Certificate}
Let $r_1, \ldots, r_d$ and $t_1, \ldots, t_{r^\star}$ for $r^\star = \max \SET{ r_1, \ldots, r_d }$ be the indices and~$x^{(r, t)}$ be the vectors constructed during the execution of \WitnessZ{V, x, [d], 1, ()}. Furthermore, consider the tuple~$I$ at the moment when the last call to the \tWitnessZ function terminates. The pair $(I^\star, A)$ for $I^\star = I \cup \big( i^\star_1, \ldots, i^\star_d \big)$, $i^\star_k = \min (P_k \setminus I)$,
and $A = \left. \big[ x^{(1, d_1)}, \ldots, x^{(1, t_1)}, \ldots, x^{(r^\star, d_{r^\star})}, \ldots, x^{(r^\star, t_{r^\star})} \big] \right|_{I^\star}$, is called the \emph{$V$-certificate} of~$x$. We label the columns of~$A$ by~$a^{(r, t)}$. Moreover, we call a pair $(I', A')$ a \emph{certificate} if there is some realization~$V$ such that the $\OKZ$-events occurs and if there exists a Pareto-optimal solution $x \in \S$ such that $(I', A')$ is the $V$-certificate of~$x$. By~$\CS$ we denote the set of all certificates.
\end{definition}

We assume that the indices~$r_k$ and $t_r$ (and hence also the indices $d_r$) are implicitly encoded in a given certificate.
Later we will take these indices into consideration again when we count the number of possible certificates.

\begin{lemma}\newcommand{\AST}{`$*$'}
\label{lemma.ZP:certificate form}
Let~$V$ be a realization for which the $\OKZ$-event occurs and let $(I^\star, A)$ be a $V$-certificate of some Pareto-optimal solution~$x$.
Let~$A$ be of the form
\[
	A = \left[ a^{(1, d_1)}, \ldots, a^{(1, t_1)}, \ldots, a^{(r^\star, d_{r^\star})}, \ldots, a^{(r^\star, t_{r^\star})} \right] \DOT
\]
For a fixed index $k \in [d]$ let
\[
	M = \left. \left[ a^{(1, d_1)}, \ldots, a^{(1, t_1)}, \ldots, a^{(r_k, d_{r_k})}, \ldots, a^{(r_k, t_{r_k})} \right] \right|_J \COMMA
\]
where $J = I^\star \cap P_k \FED (j_1, \ldots, j_m)$. Then~$M$ is of the form
\[
  M = \begin{bmatrix}
    \overline{x_{j_1}} & x_{j_1} & \ldots                 & x_{j_1}     \cr
    *                  & \ddots  & \ddots                 & \vdots      \cr
    \vdots             & \ddots  & \overline{x_{j_{m-1}}} & x_{j_{m-1}} \cr
    *                  & \ldots  & *                      & x_{j_m}
  \end{bmatrix}
  \in \SET{ 0, \ldots, \K }^{|J| \times |J|} \COMMA
\]
where each \AST\ can be an arbitrary value from $\SET{ 0, \ldots, \K }$ (different \AST-entries can represent different values) and where $\overline{z}$ for a value $z \in \SET{ 0, \ldots, \K }$ can be an arbitrary value from $\SET{ 0, \ldots, \K } \setminus \SET{ z }$.
\end{lemma}

\begin{proof}
Consider the call \WitnessZ{V, x, [d], 1, ()} and all subsequent calls \WitnessZ{V, x, K, r, I}. By definition of~$r_k$ we have $r \leq r_k \iff k \in K$ (see Line~\ref{l.ZP:last access}, Line~\ref{l.ZP:not equal indices}, and Line~\ref{l.ZP:restart}). In each call where $r \leq r_k$ one vector~$x^{(r, t)}$ is constructed in each iteration~$t$.
Also, in each iteration except for the last iteration~$t_{r_k}$ of the $r_k\th$ call, when $k \in K_\EQ$ for the first and the last time, one index $i \in P_k$ is chosen and added to~$I$. Since~$J$ consists of the chosen indices $i \in P_k$ and the additional index~$i^\star_k$, matrix~$M$ is a square matrix.

We first consider the last column of~$M$. As $x^{(r_k, t_{r_k})}$ is the last vector constructed before~$k$ is removed from~$K$, 
index~$k$ must be an element of~$K_\EQ$ in iteration~$t_{r_k}$ of call~$r_k$, i.e., $x^{(r_k, t_{r_k})} \big|_{P_k} = x|_{P_k}$.
Hence, the last column of~$M$ has the claimed form because $J \subseteq P_k$.

Now consider the remaining columns of~$M$. Due to the construction of the set~$\R^{(r)}_t$ in Line~\ref{l.ZP:initial set}, Line~\ref{l.ZP:loser set}, and Line~\ref{l.ZP:trivial loser set}, all vectors $z \in \R^{(r)}_t$ coincide with~$x$ in the previously chosen indices~$i$.
As in the case $\C^{(r)}_t \neq \emptyset$ vector~$x^{(r, t)}$ is an element of $\C^{(r)}_t \subseteq \R^{(r)}_{t+1}$ and in the case $\C^{(r)}_t = \emptyset$ vector~$x^{(r, t)}$ is constructed appropriately, the upper triangle of~$M$, excluding the principal diagonal, has the claimed form. The form of the principal diagonal follows from the choice of index $i \in P_k$: In Line~\ref{l.ZP:index choice} we chose~$i$ such that $x^{(r, t)}_i \neq x_i$, in Line~\ref{l.ZP:trivial winner} we construct~$x^{(r,t)}$ explicitely such that $x^{(r, t)}_i \neq x_i$.
\end{proof}

Like in the model without zero-preserving perturbations, our goal is to execute 
the \tWitnessZ function without revealing the entire matrix~$V$.
For this we now consider a variant of the \tWitnessZ function given as Algorithm~\ref{algorithm:WitnessZ simulation} which gets as additional parameters the~$V$-certificate of~$x$,
a shift vector $u \in \SET{ 0, \ldots, \K }^n$, the $\e$-box $B = B_V(x - u)$, and a set~$\S'$ of solutions that are
still under consideration. Recall that at the beginning of every call to the original \tWitnessZ function the set of
solutions that have still to be considered is restricted to a subset of $\bigcap_{k \in [d] \setminus K} \S_{P_k}(x)$ (see Line~\ref{l.ZP:initial set}). The huge amount of information that is necessary to restrict the current set of solutions to those that still have to be considered is not given by the $V$-certificate of~$x$. Thus, we keep track of this set of remaining solutions by passing
it as a parameter. We will see how to update this set without too much knowledge about~$x$ (cf.\ Line~\ref{l.ZP.II:restart} of Algorithm~\ref{algorithm:WitnessZ simulation}).

\IfUseAlgorithmIIeStyle
{
\LinesNumbered
\begin{algorithm*}[h!t]
  \caption{\WitnessZ{V, K, r, I^\star, A, \S', B, u}}
  \label{algorithm:WitnessZ simulation}
  let~$d_r$ be the number of components of~$K$ and let~$K$ be of the form $K = (k_1, \ldots, k_{d_r})$ \;
  set $k_{d_r+1} = d+1$ \;
  let~$b$ be the corner of~$B$ \;
  \lIf{$d_r = 0$}{\Return{$\S'$} \;}\label{l.ZP.II:success}
  set $\R^{(r)}_{d_r+1} = \S' \cap \bigcup_{s=t_r}^{d_r} \S_{I^\star} \big( a^{(r, s)} \big)$\label{l.ZP.II:initial loser set} \;
  \For{$t = d_r, d_r-1, \ldots, 0$}
  {
    set $\C^{(r)}_t = \bSET{ z \in \R^{(r)}_{t+1} \WHERE V^{k_1 \ldots k_t} \cdot (z - u) \leq b|_{k_1 \ldots k_t} } \cap \S_{I^\star} \big( a^{(r, t)} \big)$\label{l.ZP.II:winner set} \;
    \uIf{$\C^{(r)}_t \neq \emptyset$}
    {
      set $x^{(r, t)} = \argmin \bSET{ V^{k_{t+1}} z \WHERE z \in \C^{(r)}_t }$\label{l.ZP.II:winner} \;
      \uIf{$t = t_r$}
     	{
        let $K_\EQ \subseteq K$ be the tuple of indices~$k$ for which $r_k = r$\label{l.ZP.II:equal indices} \;
        set $K_\NEQ = K \setminus K_\EQ$ \;
        \Return{\bWitnessZ{V, K_\NEQ, r+1, I^\star, A, \S' \cap \bigcap_{k \in K_\EQ} \S_{P_k} \big( x^{(r, t)} \big), B, u}}\label{l.ZP.II:restart} \;
      }
      \Else
      {
        set $\R^{(r)}_t = \bSET{ z \in \R^{(r)}_{t+1} \WHERE V^{k_{t+1}} z < V^{k_{t+1}} x^{(r, t)} } \cap \bigcup_{s=t_r}^{t-1} \S_{I^\star} \big( a^{(r, s)} \big)$\label{l.ZP.II:loser set} \;
      }
    }
    \Else
    {
      set $x^{(r, t)} = (\bot, \ldots, \bot)$ \;
      set $\R^{(r)}_t = \R^{(r)}_{t+1} \cap \bigcup_{s=t_r}^{t-1} \S_{I^\star} \big( a^{(r, s)} \big)$ \;
    }
  }
  \Return{$\emptyset$} \;
\end{algorithm*}%
}
{
\begin{algorithm*}[h!t]
  \caption{\WitnessZ{V, K, r, I^\star, A, \S', B, u}}
  \label{algorithm:WitnessZ simulation}
  \begin{algorithmic}[1]
    \STATE Let~$d_r$ be the number of components of~$K$ and let~$K$ be of the form $K = (k_1, \ldots, k_{d_r})$.
    \STATE Set $k_{d_r+1} = d+1$.
    \STATE Let~$b$ be the corner of~$B$.
    \STATE \textbf{if} $d_r = 0$ \textbf{then} \textbf{return} $\S'$ \label{l.ZP.II:success}
    \STATE Set $\R^{(r)}_{d_r+1} = \S' \cap \bigcup_{s=t_r}^{d_r} \S_{I^\star} \big( a^{(r, s)} \big)$. \label{l.ZP.II:initial loser set}
    \FOR{$t = d_r, d_r-1, \ldots, 0$}
      \STATE Set $\C^{(r)}_t = \bSET{ z \in \R^{(r)}_{t+1} \WHERE V^{k_1 \ldots k_t} \cdot (z - u) \leq b|_{k_1 \ldots k_t} } \cap \S_{I^\star} \big( a^{(r, t)} \big)$. \label{l.ZP.II:winner set}
      \IF{$\C^{(r)}_t \neq \emptyset$}
        \STATE Set $x^{(r, t)} = \argmin \bSET{ V^{k_{t+1}} z \WHERE z \in \C^{(r)}_t }$. \label{l.ZP.II:winner}
        \IF{$t = t_r$}
          \STATE Let $K_\EQ \subseteq K$ be the tuple of indices~$k$ for which $r_k = r$. \label{l.ZP.II:equal indices}
          \STATE Set $K_\NEQ = K \setminus K_\EQ$.
          \RETURN \bWitnessZ{V, K_\NEQ, r+1, I^\star, A, \S' \cap \bigcap_{k \in K_\EQ} \S_{P_k} \big( x^{(r, t)} \big), B, u} \label{l.ZP.II:restart}
        \ELSE
          \STATE Set $\R^{(r)}_t = \bSET{ z \in \R^{(r)}_{t+1} \WHERE V^{k_{t+1}} z < V^{k_{t+1}} x^{(r, t)} } \cap \bigcup_{s=t_r}^{t-1} \S_{I^\star} \big( a^{(r, s)} \big)$. \label{l.ZP.II:loser set}
        \ENDIF
      \ELSE
        \STATE Set $x^{(r, t)} = (\bot, \ldots, \bot)$.
        \STATE Set $\R^{(r)}_t = \R^{(r)}_{t+1} \cap \bigcup_{s=t_r}^{t-1} \S_{I^\star} \big( a^{(r, s)} \big)$.
      \ENDIF
    \ENDFOR
    \RETURN $\emptyset$
  \end{algorithmic}
\end{algorithm*}%
}

It is important to break ties in Line~\ref{l.ZP.II:winner} the same way as we did in the original \tWitnessZ function, i.e., we take the lexicographically first solution.

\begin{lemma}
\label{lemma.ZP:same behavior}
Let $(I^\star, A)$ be the $V$-certificate of~$x$, let $u \in \SET{ 0, \ldots, \K }^n$ be an arbitrary vector, and let $B = B_V(x - u)$. Then the call \WitnessZ{V, [d], 1, I^\star, A, \S, B, u} returns~$\SET{x}$.
\end{lemma}

Before we give a formal proof of Lemma~\ref{lemma.ZP:same behavior} we try to give some intuition for it. As for the simple variant
of the \tWitness function we restrict the set of solutions to vectors that look like the vectors we want to reconstruct in the
next iterations of the current call, i.e., we intersect the current set with the set $\bigcup_{s=t_r}^{t-1} \S_{I^\star} \big( a^{(r, s)} \big)$
in iteration~$t$. In this way we only deal with subsets of the original sets, but we do not lose the vectors we want to reconstruct.
In order to reconstruct the vectors, we need more information than in the simple variant: we need to know in which iterations the
recursive calls of \tWitnessZ are made, in each call we need to know which objective functions~$V^k$ must not be considered anymore,
and for each of these objective functions we need to know the vector $x|_{P_k}$. The information when the recursive calls are made
and which objective functions must not be considered anymore is given in the certificate:
The variable~$t_r$ contains the iteration number when the recursive call is made. The index~$r_k$ contains the number of the call where index~$k$ has to be removed from~$K$. Hence,
index~$k$ is removed in the $t_{r_k}\th$ iteration of call~$r_k$. If we can reconstruct~$K_\EQ$ and the vector~$x^{(r, t)}$ in
the iteration where we make the recursive call, then we can also reconstruct the bits of~$x$ at indices $i \in P_k$ for all indices $k \in K_\EQ$
because $x|_{P_k} = x^{(r, t)}|_{P_k}$ for these indices~$k$ (cf.\ Line~\ref{l.ZP.II:restart}).

\begin{proof}
We compare the executions of \WitnessZ{V, x, [d], 1, ()} and \WitnessZ{V, [d], 1, I^\star, A, \S, B, u} and show the following claim by induction on~$r$.
\begin{claim}
\label{claim.ZP:same behavior}
If there is a call of the form \WitnessZ{V, x, K, r, I} during the execution of the call \WitnessZ{V, x, [d], 1, ()}, then there is also a call of the form \WitnessZ{V, K, r, I^\star, A, \S', B, u} for $\S' = \bigcap_{k \in [d] \setminus K} \S_{P_k}(x)$ during the execution of the call \WitnessZ{V, [d], 1, I^\star, A, \S, B, u}.
\end{claim}
\begin{proof}[Proof of Claim~\ref{claim.ZP:same behavior}]
For the case $r=1$ it is true if we recall the convention that $\bigcap_{k \in ()} \S_{P_k}(x) = \S$. Now let us consider an arbitrary call $r+1$ and assume that Claim~\ref{claim.ZP:same behavior} holds for~$r$. Hence, we can assume that there are calls of the form \WitnessZ{V, x, K, r, I} and \WitnessZ{V, K, r, I^\star, A, \S', B, u} for $\S' = \bigcap_{k \in [d] \setminus K} \S_{P_k}(x)$. We now show that both calls are executed essentially the same way. Formally, we prove the following claims by induction on~$t$, where~$\R'^{(r)}_t$, $\C'^{(r)}_t$, $x'^{(r, t)}$, $K'_\EQ$, and~$K'_\NEQ$ refer to the sets, vectors, and tuples from the call \WitnessZ{V, K, r, I^\star, A, \S', B, u}.

\begin{claim}
\label{claim.ZP:same behavior I}
$\R'^{(r)}_t \subseteq \R^{(r)}_t$ for all $t \in \SET{ t_r + 1, \ldots, d_r+1 }$.
\end{claim}

\begin{claim}
\label{claim.ZP:same behavior II}
$x'^{(r, t)} = x^{(r, t)}$ for all $t \in \SET{ t_r, \ldots, d_r }$ for which $\C^{(r)}_t \neq \emptyset$.
\end{claim}

\begin{claim}
\label{claim.ZP:same behavior III}
$x^{(r, s)} \in \R'^{(r)}_t$ for all $t \in \SET{ t_r + 1, \ldots, d_r+1 }$ and all $s \in \SET{ t_r, \ldots, t-1}$ for which $\C^{(r)}_s \neq \emptyset$.
\end{claim}

\begin{proof}[Proof of Claim~\ref{claim.ZP:same behavior I}, Claim~\ref{claim.ZP:same behavior II}, and Claim~\ref{claim.ZP:same behavior III}]
We apply a downward induction on~$t$. For the beginning, consider $t = d_r + 1$. We have
\begin{align*}
	\R'^{(r)}_{d_r+1} &= \S' \cap \bigcup_{s=t_r}^{d_r} \S_{I^\star}(a^{(r,s)})\ \text{for}\ \S' = \bigcap_{k \in [d] \setminus K} \S_{P_k}(x) \quad \text{and} \cr
	\R^{(r)}_{d_r+1} &= \S_I(x) \cap \bigcap_{k \in [d] \setminus K} \S_{P_k}(x) \DOT
\end{align*}
Due to the construction of the vectors~$x^{(r,s)}$ and the definition of~$a^{(r,s)}$,
\[
	a^{(r,s)}|_I
	= x^{(r,s)}|_I
	= x|_I
\]
for all $s = t_r, \ldots, d_r$ (see Lemma~\ref{lemma.ZP:certificate form}). The inclusion $I^\star \supseteq I$ yields
\[
	\S_{I^\star}(a^{(r,s)})
	\subseteq \S_I(a^{(r,s)})
	= \S_I(x)
\]
for all $s = t_r, \ldots, d_r$. Consequently, $\R'^{(r)}_{d_r+1} \subseteq \R^{(r)}_{d_r+1}$. For Claim~\ref{claim.ZP:same behavior II} nothing has to be shown in the initial step $t = d_r + 1$ of the induction. Let us consider Claim~\ref{claim.ZP:same behavior III} and let $s \in \SET{ t_r, \ldots, d_r }$ be an arbitrary index for which $\C^{(r)}_s \neq \emptyset$. Then
\[
  x^{(r,s)}
  \in \C^{(r)}_s
  \subseteq \R^{(r)}_{s+1}
  \subseteq \R^{(r)}_{d_r+1}
  \subseteq \bigcap_{k \in [d] \setminus K} \S_{P_k}(x) \DOT
\]
Furthermore, $x^{(r,s)} \in \S_J(a^{(r,s)})$ for every index tuple~$J$ due to the definition of~$a^{(r,s)}$. Consequently, $x^{(r,s)} \in \R'^{(r)}_{d_r+1}$.

For the induction step let $t \leq d_r$. Due to the occurence of the $\OKZ$-event and the fact that $B = B_V(x-u)$, we obtain
\begin{align*}
  \C'^{(r)}_t
  &= \bSET{ z \in \R'^{(r)}_{t+1} \WHERE V^{k_1 \ldots k_t} \cdot (z-u) \leq b|_{k_1 \ldots k_t} } \cap \S_{I^\star}(a^{(r,t)}) \cr
  &= \bSET{ z \in \R'^{(r)}_{t+1} \WHERE V^{k_1 \ldots k_t} z < V^{k_1 \ldots k_t} x } \cap \S_{I^\star}(a^{(r,t)}) \quad \text{and} \cr
  \C^{(r)}_t
  &= \bSET{ z \in \R^{(r)}_{t+1} \WHERE V^{k_1 \ldots k_t} z < V^{k_1 \ldots k_t} x } \DOT
\end{align*}
Since $\R'^{(r)}_{t+1} \subseteq \R^{(r)}_{t+1}$, we obtain $\C'^{(r)}_t \subseteq \C^{(r)}_t$. First, we consider the case $\C^{(r)}_t = \emptyset$ which implies $\C'^{(r)}_t = \emptyset$ and $t \geq t_r + 1$. The inequality follows from the fact that in iteration~$t_r$ the \tWitnessZ function is called (Line~\ref{l.ZP:restart} of Algorithm~\ref{algorithm:WitnessZ}) which implies $\C^{(r)}_{t_r} \neq \emptyset$. In this case,
\begin{align*}
	\R'^{(r)}_t &= \R'^{(r)}_{t+1} \cap \bigcup_{s=t_r}^{t-1} \S_{I^\star}(a^{(r,s)}) \quad \text{and} \cr
	\R^{(r)}_t &= \R^{(r)}_{t+1} \cap \S_I(x) \COMMA
\end{align*}
where~$I$ is the updated index tuple~$I$. Due to the construction of the vectors~$x^{(r,s)}$ (see Lemma~\ref{lemma.ZP:certificate form}) and the definition of the vectors~$a^{(r,s)}$, we know that
\[
	a^{(r,s)}|_I
	= x^{(r,s)}|_I
	= x|_I
\]
for all $s = t_r, \ldots, t-1$. As $I^\star \supseteq I$, this implies
\[
	\bigcup_{s=t_r}^{t-1} \S_{I^\star}(a^{(r,s)})
	\subseteq \bigcup_{s=t_r}^{t-1} \S_I(a^{(r,s)})
	= \S_I(x) \DOT
\]
As $\R'^{(r)}_{t+1} \subseteq \R^{(r)}_{t+1}$ in accordance with the induction hypothesis, Claim~\ref{claim.ZP:same behavior I}, we obtain $\R'^{(r)}_t \subseteq \R^{(r)}_t$. For Claim~\ref{claim.ZP:same behavior II} nothing has to be shown here. Let $s \in \SET{ t_r, \ldots, t-1 }$ be an arbitrary index for which $\C^{(r)}_s \neq \emptyset$. Then $x^{(r,s)} \in \R'^{(r)}_{t+1}$ by Claim~\ref{claim.ZP:same behavior III} of the induction hypothesis, $x^{(r,s)} \in \S_{I^\star}(a^{(r,s)})$, and consequently $x^{(r,s)} \in \R'^{(r)}_t$.

Let us finally consider the case $\C^{(r)}_t \neq \emptyset$. Claim~\ref{claim.ZP:same behavior III} of the induction hypothesis yields $x^{(r,t)} \in \R'^{(r)}_{t+1}$. Since $x^{(r,t)} \in \S_{I^\star} \big( a^{(r,t)} \big)$ and $V^{k_1 \ldots k_t} x^{(r,t)} < V^{k_1 \ldots k_t} x$, also $x^{(r,t)} \in \C'^{(r)}_t$ and, thus, $\C'^{(r)}_t \neq \emptyset$. Hence, $x'^{(r,t)} = x^{(r,t)}$ as $\C'^{(r)}_t \subseteq \C^{(r)}_t$. Claim~\ref{claim.ZP:same behavior I} and Claim~\ref{claim.ZP:same behavior III} have only to be validated if $t \geq t_r$, i.e., we can assume that $K_\EQ = ()$. Then
\[
	\R'^{(r)}_t
	= \bSET{ z \in \R'^{(r)}_{t+1} \WHERE V^{k_{t+1}} z < V^{k_{t+1}} x^{(r,t)} } \cap \bigcup_{s=t_r}^{t-1} \S_{I^\star}(a^{(r,s)})
\]
because $x'^{(r,t)} = x^{(r,t)}$, and
\[
	\R^{(r)}_t
	= \bSET{ z \in \R^{(r)}_{t+1} \WHERE V^{k_{t+1}} z < V^{k_{t+1}} x^{(r,t)} } \cap \S_I(x) \DOT
\]
With the same argument used for the case $\C^{(r)}_t = \emptyset$ we obtain
\[
	\R'^{(r)}_{t+1} \cap \bigcup_{s=t_r}^{t-1} \S_{I^\star}(a^{(r,s)})
	\subseteq \R^{(r)}_{t+1} \cap \S_I(x)
\]
and, hence, $\R'^{(r)}_t \subseteq \R^{(r)}_t$. Consider an arbitrary index $s \in \SET{ t_r, \ldots, t-1 }$ for which $\C^{(r)}_s \neq \emptyset$. Then
\[
	x^{(r,s)}
	\in \C^{(r)}_s
	\subseteq \R^{(r)}_{s+1}
	\subseteq \R^{(r)}_t \DOT
\]
In particular, $V^{k_{t+1}} x^{(r,s)} < V^{k_{t+1}} x^{(r,t)}$ (see Line~\ref{l.ZP:loser set}) and, hence, $V^{k_{t+1}} x^{(r,s)} < V^{k_{t+1}} x'^{(r,t)}$ because $x'^{(r,t)} = x^{(r,t)}$. Furthermore, $x^{(r,s)} \in \R'^{(r)}_{t+1}$ due to the induction hypothesis, Claim~\ref{claim.ZP:same behavior III}, and $x^{(r,s)} \in \S_{I^\star}(a^{(r,s)})$. Consequently, $x^{(r,s)} \in \R'^{(r)}_t$. 
\end{proof}

The induction step of the proof of Claim~\ref{claim.ZP:same behavior} follows from the three claims above: Let us assume that there is a call of the form \WitnessZ{V, x, \hat{K}, r+1, \hat{I}}. By the definition of~$t_r$, this call is executed in iteration~$t_r$ of call~$r$. Consequently, $x^{(r,t_r)} \in \C^{(r)}_{t_r} \neq \emptyset$. Applying Claim~\ref{claim.ZP:same behavior III} for $s = t_r$ and $t = t_r + 1$, we obtain $x^{(r,t_r)} \in \R'^{(r)}_{t_r+1}$. As $x^{(r,t_r)} \in \C^{(r)}_{t_r}$, the inequalities $V^{k_1 \ldots k_{t_r}} x^{(r,t_r)} < V^{k_1 \ldots k_{t_r}} x$ hold, which are equivalent to $V^{k_1 \ldots k_{t_r}} \cdot (x^{(r,t_r)} - u) \leq b|_{k_1 \ldots k_{t_r}}$ due to the occurence of the $\OKZ$-event. Furthermore, $x^{(r,t_r)} \in \S_{I^\star}(a^{(r,t_r)})$ by the definition of $a^{(r,t_r)}$. Hence, $x^{(r,t_r)} \in \C'^{(r)}_{t_r}$ (see Line~\ref{l.ZP.II:winner set}), i.e., $\C'^{(r)}_{t_r} \neq \emptyset$. Moreover, $x'^{(r, t_r)} = x^{(r, t_r)}$ in accordance with Claim~\ref{claim.ZP:same behavior II}. In iteration~$t_r$ of the call \WitnessZ{V, K, r, I^\star, A, \S', B, u} Line~\ref{l.ZP.II:equal indices} is reached. By the definition of the values~$r_k$ we obtain $K'_\EQ = K_\EQ$, and hence, $x'^{(r,t_r)}|_{P_k} = x^{(r,t_r)}|_{P_k} = x|_{P_k}$ for all $k \in K'_\EQ = K_\EQ$ due to the definition of~$K_\EQ$. In Line~\ref{l.ZP.II:restart}, there is a call of the form \WitnessZ{V, K'_\NEQ, r+1, I^\star, A, \S' \cap \bigcap_{k \in K'_\EQ} \S_{P_k}(x^{(r,t_r)}), B, u}. The correctness of Claim~\ref{claim.ZP:same behavior} follows because
\[
  K'_\NEQ
  = K \setminus K'_\EQ
  = K \setminus K_\EQ
  = K_\NEQ
  = \hat{K} \COMMA
\]
where~$\hat{K}$ is the parameter from the call \WitnessZ{V, x, \hat{K}, r+1, \hat{I}}, and because
\begin{align*}
  \S' \cap \bigcap_{k \in K'_\EQ} \S_{P_k}(x^{(r,t_r)}), B, u)
  &= \bigcap_{k \in [d] \setminus K} \S_{P_k}(x) \cap \bigcap_{k \in K'_\EQ} \S_{P_k}(x^{(r,t_r)}) \cr
  &= \bigcap_{k \in [d] \setminus K} \S_{P_k}(x) \cap \bigcap_{k \in K'_\EQ} \S_{P_k}(x) \cr
  &= \bigcap_{k \in [d] \setminus K'_\NEQ} \S_{P_k}(x) \DOT \qed
\end{align*}
\end{proof}

Let us finish the proof of Lemma~\ref{lemma.ZP:same behavior}. According to Lemma~\ref{lemma.ZP:Witness}, the call \WitnessZ{V, x, [d], 1, ()} returns the set $\SET{x} \neq \emptyset$. Hence, there is a call of the form \WitnessZ{V, x, K, r, I} for $K = ()$ (see Line~\ref{l.ZP:success}). Due to Claim~\ref{claim.ZP:same behavior}, there must be also a call of the form \WitnessZ{V, (), r, I^\star, A, \S', B, u} for $\S' = \bigcap_{k \in [d] \setminus ()} \S_{P_k}(x) = \SET{x}$.  This set is immediately returned in Line~\ref{l.ZP.II:success}. \qed
\end{proof}

By the choice of the vector~$u$ we can control which information about~$V$ has to be
known in order to be able to execute the call \WitnessZ{V, [d], 1, I^\star, A, \S, B, u}.
While Lemma~\ref{lemma.ZP:same behavior} is correct for every choice of $u \in \SET{ 0, \ldots, \K }^n$,
we have to choose~$u$ carefully in order for the following probabilistic analysis to work.
Later we will see that $u^\star = u^\star(I^\star, A)$, given by
\begin{equation}
\label{eq.ZP:shift vector}
u^\star_i = \begin{cases}
  |x_i-1| & \text{if}\ i \in (i^\star_1, \ldots, i^\star_d) \COMMA \cr
  x_i     & \text{if}\ i \in I^\star \setminus (i^\star_1, \ldots, i^\star_d) \COMMA \cr
  0       & \OTHERWISE \COMMA
\end{cases}
\end{equation}
is well-suited for our purpose. Recall that $i^\star_k \in P_k$ are the indices that have been added to~$I$ in the definition of the $V$-certificate to obtain~$I^\star$. Furthermore, $x_i$ is given by the last column of~$A$ for every index $i \in I^\star$ (cf.\ Lemma~\ref{lemma.ZP:Witness}). Hence, vector~$u^\star$ can be defined with the information that is contained in the $V$-certificate of~$x$; we do not have to know the vector~$x$ itself.

In the next step, we bound the number of Pareto-optimal solutions.
For this, consider the following function $\chi_{I^\star, A, B}(V)$, parameterized by an
arbitrary certificate $(I^\star, A) \in \CS$ and an arbitrary $\e$-box $B \in \B_\e$, that is
defined as follows: $\chi_{I^\star, A, B}(V) = 1$ if
\WitnessZ{V, [d], 1, I^\star, A, \S, B, u^\star(I^\star, A)} returns a set~$\SET{x'}$ such
that $B_V \big( x' - u^\star(I^\star, A) \big) = B$, and $\chi_{I^\star, A, B}(V) = 0$ otherwise.

\begin{corollary}
\label{corollary.ZP:realization bound}
Assume that the $\OKZ$-event occurs. Then the number $\PO(V)$ of Pareto-optimal solutions is at most
\[
	\sum \limits_{(I^\star, A) \in \CS} \sum \limits_{B \in \B_\e} \chi_{I^\star, A, B}(V) \DOT
\]
\end{corollary}

\begin{proof}
Let~$x$ be a Pareto-optimal solution, let $(I^\star, A)$ be the $V$-certificate of~$x$, and let $B = B_V \big( x - u^\star(I^\star, A) \big) \in \B_\e$. Due to Lemma~\ref{lemma.ZP:same behavior}, \WitnessZ{V, [d], 1, I^\star, A, \S, B, u^\star(I^\star, A)} returns~$\SET{x}$. Hence, $\chi_{I^\star, A, B}(V) = 1$. It remains to show that this function $x \mapsto (I^\star, A, B')$ defined within the previous lines is injective. Let~$x_1$ and~$x_2$ be distinct Pareto-optimal solutions and let $(I^\star_1, A_1)$ and $(I^\star_2, A_2)$ be the $V$-certificates of~$x_1$ and~$x_2$, respectively. If $(I^\star_1, A_1) \neq (I^\star_2, A_2)$, then~$x_1$ and~$x_2$ are mapped to distinct triplets. Otherwise, $u^\star(I^\star_1, A_1) = u^\star(I^\star_2, A_2)$ and, hence, $B_V \big( x_1 - u^\star(I^\star_1, A_1) \big) \neq B_V \big( x_2 - u^\star(I^\star_2, A_2) \big)$ because of the $\OKZ$-event. Consequently, also in this case~$x_1$ and~$x_2$ are mapped to distinct triplets.
\end{proof}

Corollary~\ref{corollary.ZP:realization bound} immediately implies a bound on
the expected number of Pareto-optimal solutions.

\begin{corollary}
\label{corollary.ZP:expectation bound}
The expected number of Pareto-optimal solutions is bounded by
\[
  \Ex[V]{\PO(V)} \leq \sum_{(I^\star, A) \in \CS} \sum_{B \in \B_\e} \Pr[V]{E_{I^\star, A, B}} + (\K+1)^n \cdot \Pr[V]{\overline{\OKZ(V)}}
\]
where $E_{I^\star, A, B}$ denotes the event that the call \WitnessZ{V, [d], 1, I^\star, A, \S, B, u^\star(I^\star, A)} returns a set~$\SET{x'}$ such that $B_V \big( x' - u^\star(I^\star, A) \big) = B$.
\end{corollary}

\begin{proof}
By applying Corollary~\ref{corollary.ZP:realization bound}, we obtain
\begin{align*}
&\Ex[V]{\PO(V)} \cr
  &= \Ex[V]{\PO(V) \,\big|\, \OKZ(V)} \cdot \Pr[V]{\OKZ(V)} + \Ex[V]{\PO(V) \,\big|\, \overline{\OKZ(V)}} \cdot \Pr[V]{\overline{\OKZ(V)}} \cr
  &\leq \Ex[V]{\left. \sum_{(I^\star, A) \in \CS} \sum_{B \in \B_\e} \chi_{I^\star, A, B}(V) \right| \OKZ(V)} \cdot \Pr[V]{\OKZ(V)} + |\S| \cdot \Pr[V]{\overline{\OKZ(V)}} \cr
  &\leq \Ex[V]{\sum_{(I^\star, A) \in \CS} \sum_{B \in \B_\e} \chi_{I^\star, A, B}(V)} + (\K+1)^n \cdot \Pr[V]{\overline{\OKZ(V)}} \cr
  &= \sum_{(I^\star, A) \in \CS} \sum_{B \in \B_\e} \Pr[V]{E_{I^\star, A, B}} + (\K+1)^n \cdot \Pr[V]{\overline{\OKZ(V)}} \DOT \qedhere
\end{align*}
\end{proof}

We will see that the first term of the sum in
Corollary~\ref{corollary.ZP:expectation bound} can be bounded independently
of~$\e$ and that the second term tends to~$0$ for $\e \to 0$.
First of all, we analyze the size of the certificate space.

\begin{lemma}
\label{lemma.ZP:certificate space}
The size of the certificate space is bounded by
\[
  |\CS| = (\K+1)^{(d^2+d)(d^3+d^2+d)} \cdot O \big( n^{d^3+d^2} \big) \DOT
\]
\end{lemma}

\begin{proof}
Consider the execution of the call \WitnessZ{V, x, [d], 1, ()} and let $r^\star = \max \SET{ r_1, \ldots, r_d }$
be the maximum of the values $r_1, \ldots, r_d$. Including this call with number~$1$,
there can be at most~$d$ calls to the \tWitnessZ function except for the call with
number~$r^\star+1$ that
terminates due to $d_{r^\star+1} = 0$. This is because in each of the other calls at least one
index $k \in [d]$ is removed from the tuple~$K$. Hence, $r_1, \ldots, r_d \in [d]$.
Consequently, there are at most~$d^d$ possibilities for these numbers.
In the $r\th$ call, the iteration number~$t_r$ is an element of
$[d_r]_0 \subseteq [d]_0$, and hence, there are at most $(d+1)^{r^\star} \leq (d+1)^d$ possibilities to choose
iteration numbers $t_1,\ldots,t_{r^\star}$. In each iteration, at most~$d$ indices~$i$ are added to the
tuple~$I$. As there are at most~$d$ calls and at most $d+1$ iterations each call,
tuple~$I$ contains at most $d^2 \cdot (d+1)$ indices in total. Hence,
there are at most
\[
	\sum_{k=1}^{d^2(d+1)} n^k \leq d^2 \cdot (d+1)\cdot n^{d^2\cdot(d+1)}
\]
choices for~$I$. Once~$I$ is fixed, also the
indices in $I^\star\setminus I$ are fixed because the indices added to~$I$ in Definition~\ref{definition.ZP:Certificate}
are determined by~$I$. The tuple~$I^\star$ contains $|I| + d \leq d^3+d^2+d$ indices. In each call~$r$, at most $d+1$
vectors~$x^{(r, t)}$ are generated. Hence, matrix~$A$ has at
most $d \cdot (d+1)$ columns and at most $d^3+d^2+d$ rows. This yields
the claimed bound
\begin{align*}
|\CS|
  &\leq d^d \cdot (d+1)^d \cdot d^2 \cdot (d+1) \cdot n^{d^2(d+1)} \cdot (\K+1)^{d(d+1) \cdot (d^3+d^2+d)} \cr
  &\leq 2^{d+1} \cdot d^{2d+3} \cdot n^{d^3+d^2} \cdot (\K+1)^{(d^2+d)(d^3+d^2+d)} \cr
  &= (\K+1)^{(d^2+d)(d^3+d^2+d)} \cdot O \big( n^{d^3+d^2} \big) \DOT \qedhere
\end{align*}
\end{proof}

In the next step we analyze how much information of~$V$ is required in order 
to perform the call \WitnessZ{V, [d], 1, I^\star, A, \S, B, u} for a fixed certificate $(I^\star, A)$. We will see that~$V$ does not need to be revealed completely and that some randomness remains even after the necessary information to perform the call has been revealed. This is the key observation for analyzing
the probability $\Pr[V]{E_{I^\star, A, B}}$. For this, let~$V$ be an arbitrary realization,
i.e., we do not condition on the $\OKZ$-event anymore. Let $I^\star_k = I^\star \cap P_k$ for $k \in [d]$. We
apply the principle of deferred decisions and assume that for every $k \in [d]$ the coefficients of~$V^k$ belonging to indices $i \notin I^\star_k$ are fixed arbitrarily. We denote this part of~$V^k$ by~$V^k_{\overline{I^\star_k}}$ and concentrate on the remaining part of~$V^k$ which we denote by~$V^k_{I^\star_k}$.

As in the model with non-zero-preserving perturbations, the call \WitnessZ{V, [d], 1, I^\star, A, \S, B, u} can be executed without the full knowledge of $V^1_{I^\star_1}, \ldots, V^d_{I^\star_d}$. We write the linear combinations of~$V^k_{I^\star_k}$ of the calls $r = 1, \ldots, r_k-1$ and of call~$r_k$ that suffice to be known into the following matrices~$P_k$ and~$Q_k$, respectively:
\begin{align*}
  P_k
  &= \left. \left[ p^{(1, d_1)}_k, \ldots, p^{(1, t_1)}_k, \ldots, p^{(r_k-1, d_{r_k-1})}_k, \ldots, p^{(r_k-1, t_{r_k-1})}_k \right] \right|_{I^\star_k}
  \quad \text{and} \cr
  Q_k
  &= \left. \left[ p^{(r_k, d_{r_k})}_k, \ldots, p^{(r_k, j_k)}_k, p^{(r_k, j_k-2)}_k - p^{(r_k, j_k-1)}_k, \ldots, p^{(r_k, t_{r_k})}_k - p^{(r_k, j_k-1)}_k \right] \right|_{I^\star_k}
\end{align*}
for $p^{(r, t)}_k = a^{(r, t)} \big|_{I^\star_k} - u|_{I^\star_k}$, where $a^{(r, t)}$ are the columns of matrix~$A$. The index $j_k \in [d_{r_k}]$ denotes the index for which $k_{j_k} = k$ in the $r_k\th$ call of the \tWitnessZ function, i.e., $V^k$ is the $j_k\th$ objective function in~$K$ in the call $r = r_k$ and it is not considered anymore in iterations $t < j_k$.

Note that the matrices $P_k = P_k(I^\star, A, u)$ and $Q_k = Q_k(I^\star, A, u)$ depend, among others, on the choice of~$u$. Furthermore, the indices~$j_k$ are determined by the certificate $(I^\star, A)$. To be precise, the indices $r_1, \ldots, r_d$, which are implicitely given by the certificate $(I^\star, A)$, contain the information which objectives are still under consideration in a certain call~$r$ of the \tWitnessZ function: These are all objectives~$V^k$ for which $r_k \geq r$.

Observe that matrix~$Q_k$ has
\[
	(d_{r_k}-j_k+1) + ((j_k-2)-t_{r_k}+1) = d_{r_k}-t_{r_k}
\]
columns and matrix~$P_k$ has
\[
  \sum_{r=1}^{r_{k-1}} (d_r-t_r+1)
  = \sum_{r=1}^{r_k} (d_r-t_r+1) - (d_{r_k}-t_{r_k}+1)
  = |I^\star_k| - (d_{r_k}-t_{r_k}+1)
\]
columns. The last equation is due to the fact that in each call $r < r_k$ in each iteration one index~$i_k$ is chosen. In call $r = r_k$ one index~$i_k$ is chosen in each iteration $t > t_{r_k}$. The equation follows since~$I^\star_k$ contains one index more than the number of indices~$i_k$ that are chosen during the execution of the \tWitnessZ function (see Definition~\ref{definition.ZP:Certificate}). Moreover, observe that all entries of~$P_k$ and~$Q_k$ are from $\SET{ -\K, \ldots, \K }$.

\begin{lemma}
\label{lemma.ZP:output determination}
Let $u \in \SET{ 0, \ldots, \K }^n$ be an arbitrary shift vector, let $(I^\star, A) \in \CS$ be an arbitrary certificate, and let~$U$ and~$W$ be two realizations for which $U^k_{\overline{I^\star_k}} = W^k_{\overline{I^\star_k}}$ and $U^k_{I^\star_k} \cdot q = W^k_{I^\star_k} \cdot q$ for all indices $k \in [d]$ and all columns~$q$ of one of the matrices $P_k(I^\star, A, u)$ and $Q_k(I^\star, A, u)$. Then for every $\e$-box $B \in \B_\e$ the calls \WitnessZ{U, [d], 1, I^\star, A, \S, B, u} and \WitnessZ{W, [d], 1, I^\star, A, \S, B, u} return the same result.
\end{lemma}

\begin{proof}
We fix an index $k \in [d]$ and analyze which information of $V^k_{I^\star_k}$ is required for the execution of the call \WitnessZ{V, [d], 1, I^\star, A, \S, B, u}. By the construction of index~$r_k$ we know that only in the calls $r = 1, \ldots, r_k$ information about the function~$V^k$ must be available as in all subsequent calls this function is not considered anymore.

Since in call~$r$ we consider only vectors that coincide with one of the vectors $x^{(r, t_r)}, \ldots, x^{(r, d_r)}$ (see Line~\ref{l.ZP.II:initial loser set}) in the indices $i \in I^\star$, it suffices to know all linear combinations
\[
	V^k_{I^\star_k} \cdot \left( a^{(r, t)} \big|_{I^\star_k} - u|_{I^\star_k} \right) \DOT
	= V^k_{I^\star_k} \cdot p^{(r, t)}_k \DOT
\]
For all call numbers $r = 1, \ldots, r_k-1$ and all iterations $t = t_r, \ldots, d_r$ in these calls the vector $p^{(r, t)}_k$ is a column of matrix~$P_k$.

It remains to analyze which information of $V^k_{I^\star_k}$ is required in the $r_k\th$ call of the \tWitnessZ function. First, we observe that $t_{r_k} \leq j_k - 1$. This is due to the fact that $x^{(r_k, t_{r_k})} \big|_{P_k} = x|_{P_k}$, i.e., $V^k x^{(r_k, t_{r_k})} = V^k x$, since this is the iteration when~$k$ is removed from tuple~$K$. On the other hand, $x^{(r_k, t_{r_k})} \in \C^{(r_k)}_{t_{r_k}}$, which means that $V^{k_s} x^{(r_k, t_{r_k})} < V^{k_s} x$ for all indices $s = 1, \ldots, t_{r_k}$, where $k_1, \ldots, k_{d_{r_k}}$ denote the indices tuple~$K$ consists of in call $r = r_k$. Hence, $t_{r_k} < j_k$ as $k = k_{j_k}$.

There are only three lines where information about~$V$ is required: Line~\ref{l.ZP.II:winner set}, Line~\ref{l.ZP.II:winner}, and Line~\ref{l.ZP.II:loser set}. For Line~\ref{l.ZP.II:winner set} the values
\[
	V^k_{I^\star_k} \cdot \left( a^{(r_k, t)} \big|_{I^\star_k} - u|_{I^\star_k} \right)
	= V^k_{I^\star_k} \cdot p^{(r_k, t)}_k
\]
from iteration $t = d_{r_k}$ down to iteration~$j_k$ are needed. The vectors $p^{(r_k,j_k)}, \ldots, p^{(r_k, d_{r_k})}$ are columns of matrix~$Q_k$. In Line~\ref{l.ZP.II:winner} no additional information about~$V^k_{I^\star}$ is required since all considered vectors agree on indices $i \in I^\star_k$ with each other. For Line~\ref{l.ZP.II:loser set} only in iteration $t = j_k-1$ the values
\[
	V^k_{I^\star_k} \cdot \left( a^{(r_k, s)} \big|_{I^\star_k} - a^{(r_k, {j_k-1})} \big|_{I^\star_k} \right)
\]
for $s = t_{r_k}, \ldots, j_k-2$ are required. Observe that the vectors
\begin{align*}
	p^{(r_k, s)}_k - p^{(r_k, j_k-1)}_k
	&= \left( a^{(r_k, s)} \big|_{I^\star_k} - u|_{I^\star_k} \right) - \left( a^{(r_k, j_k-1)} \big|_{I^\star_k} - u|_{I^\star_k} \right) \cr
	&= a^{(r_k, s)} \big|_{I^\star_k} - a^{(r_k, {j_k-1})} \big|_{I^\star_k}
\end{align*}
for $s = t_{r_k}, \ldots, j_k-2$ are columns of matrix~$Q_k$.

As~$U$ and~$W$ agree on all necessary information, both calls return the same result.
\end{proof}

In the remainder of this section we assume that~$V^k_{\overline{I^\star_k}}$ and the $\e$-box~$B$ are fixed. 
In accordance with Lemma~\ref{lemma.ZP:output determination}, the output of the call
\WitnessZ{V, [d], 1, I^\star, A, \S, B, u} is determined if the linear combinations of~$V^k_{I^\star_k}$ given by the columns of the
matrices~$P_k$ and~$Q_k$ are fixed arbitrarily, i.e., it does not depend on the remaining randomness in the coefficients. We are interested in the event $E_{I^\star, A, B}$, i.e., in the
event that the output is a set~$\SET{x'}$ such that $V^{1 \ldots d} \cdot \big( x' - u^\star(I^\star, A) \big) \in B$.
Since $(I^\star, A)$ is a $V$-certificate of~$x$ for some~$V$ and~$x$, the output is always of the form~$\SET{x'}$ due to
Lemma~\ref{lemma.ZP:same behavior}. Hence, event $E_{I^\star, A, B}$ occurs if and only if for all indices~$k$ the
relation $V^k_{I^\star_k} \cdot \big( x'-u^\star(I^\star, A) \big)\big|_{I^\star_k} \in C_k$ holds for some interval~$C_k$ of length~$\e$ that
depends on the linear combinations of~$V_{I^\star_{\ell}}$ given by~$P_\ell$ and~$Q_\ell$ for all indices $\ell \in [d]$.

\begin{lemma}
\label{lemma.ZP:linear independence}
For every fixed index $k \in [d]$ the columns of matrix $P_k \big( I^\star, A, u^\star(I^\star, A) \big)$, of matrix $Q_k \big( I^\star, A, u^\star(I^\star, A) \big)$, and the vector $p^{(r_k, t_{r_k})}_k$ are linearly independent.
\end{lemma}

\begin{proof}
Consider the square matrix~$\hat{Q}_k$ consisting of the vectors~$p^{(r, t)}_k$, for $r \in \SET{ 1, \ldots, r_k }$ and $t \in \SET{ t_r, \ldots, d_r }$. Matrix~$\hat{Q}_k$ can be obtained from the matrix~$M$ of Lemma~\ref{lemma.ZP:certificate form} by subtracting the vector $u^\star|_{I^\star_k}$ from all of its columns. Due to Lemma~\ref{lemma.ZP:certificate form} and due to the construction of $u^\star = u^\star(I^\star, A)$ (see Equation~\ref{eq.ZP:shift vector}) matrix~$\hat{Q}_k$ is a lower triangular matrix and the elements of the principal diagonal are from the set $\SET{ -\K, \ldots, \K } \setminus \SET{0}$. This is because $x_i - u^\star_i = x_i - |x_i-1| \in \SET{ -1, 1 }$ for $i = i^\star_k$ and $\overline{x_i} - u^\star_i = \overline{x_i} - x_i \neq 0$ for all $i \in I^\star_k \setminus (i^\star_k)$, where~$\overline{z}$ for $z \in \SET{ 0, \ldots, \K }$ represents an arbitrary value from $\SET{ 0, \ldots, \K }$ not equal to~$z$. Hence, the vectors $p^{(r, t)}_k$ are linearly independent.

The columns of~$Q_k$ and $p^{(r_k, t_{r_k})}_k$ are linear combinations of the vectors $p^{(r_k, t_{r_k})}_k, \ldots, p^{(r_k, d_{r_k})}_k$, whereas the columns of matrix~$P_k$ are the remaining columns of matrix~$\hat{Q}_k$. As the vectors~$p^{(r, t)}_k$ are linearly independent, it suffices to show that the columns of matrix~$Q_k$ and vector~$p^{(r_k, t_{r_k})}_k$ are linearly independent. For this, we consider an arbitrary linear combination of the columns of matrix~$Q_k$ and vector~$p^{(r_k, t_{r_k})}_k$ and show that it is~$0$ if and only if all coefficients are~$0$. For sake of simplicity, we drop the index~$k$ in the remainder of this proof and write~$r$, $j$, and~$p^{(r,t)}$ instead of~$r_k$, $j_k$, and~$p^{(r_k,t)}_k$, respectively.
\[
  \sum_{t=j}^{d_r} \lambda_t \cdot p^{(r, t)} + \sum_{t=t_r}^{j-2} \lambda_t \cdot \left( p^{(r, t)} - p^{(r, j-1)} \right) + \mu \cdot p^{(r, t_r)} = 0 \DOT
\]
If $t_r = j - 1$, then this equation is equivalent to
\[ 
  \sum_{t=t_r+1}^{d_r} \lambda_t \cdot p^{(r, t)} + \mu \cdot p^{(r, t_r)} = 0 \DOT
\]
Therefore, all coefficients are~$0$ due to the linear independence of the vectors~$p^{(r, t)}$. If $t_r < j-1$, which is the only case remaining due to previous observations, then the equation is equivalent to
\[
  \sum_{t=j}^{d_r} \lambda_t \cdot p^{(r, t)} + \sum_{t=t_r+1}^{j-2} \lambda_t \cdot p^{(r, t)} - \left( \sum_{t=t_r}^{j-2} \lambda_t \right) \cdot p^{(r, j-1)} + \left( \lambda_{t_r} + \mu \right) \cdot p^{(r, t_r)} = 0 \DOT
\]
The linear independence of the vectors~$p^{(r, t)}$ implies $\lambda_t = 0$ for $t \in \bSET{ t_r+1, \ldots, j-2 } \cup \bSET{ j, \ldots, d_r }$, $\sum_{t=t_r}^{j-2} \lambda_t = 0$, and $\lambda_{t_r} + \mu = 0$. Consequently, also $\lambda_{t_r} = 0$ and, thus, $\mu = 0$. This means that all coefficients are~$0$. In both cases, the linear independence of the columns of~$Q_k$ and the vector~$p^{(r, t_r)}$ follows.
\end{proof}

\begin{corollary}
\label{corollary.ZP:success unlikely}
Let $\gamma = d^3 + d^2 + d$. For an arbitrary certificate $(I^\star, A) \in \CS$ the probability of the event $E_{I^\star, A, B}$ is bounded by
\[
  \Pr[V]{E_{I^\star, A, B}} \leq (2\gamma\K)^{\gamma-d} \phi^\gamma \e^d
\]
and by
\[
  \Pr[V]{E_{I^\star, A, B}} \leq 2^d (\gamma\K)^{\gamma-d} \phi^d \e^d
\]
if all densities are quasiconcave.
\end{corollary}

\begin{proof}
For a fixed index $k \in [d]$ we write the columns of matrix~$P_k$, of matrix~$Q_k$, and vector~$p^{(r_k, t_{r_k})}_k$ into one
matrix $Q_k' \in \SET{ -\K, \ldots, \K }^{|I^\star_k| \times |I^\star_k|}$ (the number of columns is
\[
	\big( |I^\star_k| - (d_{r_k}-t_{r_k}+1) \big) + (d_{r_k}-t_{r_k}) + 1
	= |I^\star_k|
\]
due to previous observations) and consider the matrix
\[
  Q' = \begin{bmatrix}
    Q'_1   & \NULL  & \ldots & \NULL  \cr
    \NULL  & \ddots & \ddots & \vdots \cr
    \vdots & \ddots & \ddots & \NULL  \cr
    \NULL  & \ldots & \NULL  & Q'_d
  \end{bmatrix}
  \in \SET{ -\K, \ldots, \K }^{|I^\star| \times |I^\star|} \DOT
\]
This matrix has full rank due to Lemma~\ref{lemma.ZP:linear independence}. Now we permute the columns of~$Q'$ to obtain a matrix~$Q$ whose last~$d$ columns belong to the last column of one of the matrices~$Q_k$. This means that the last~$d$ columns are
\[
	\left( p^{(r_1,t_{r_1})}_1, \NULL^{|I^\star_2|}, \ldots, \NULL^{|I^\star_d|} \right), \ldots, \left( \NULL^{|I^\star_1|}, \ldots, \NULL^{|I^\star_{d-1}|}, p^{(r_d,t_{r_d})}_d \right) \DOT
\]
For all $k \in [d]$ and every index $i_k \in I^\star_k$ let $X_i = V^k_i$ be the $i\th$ coefficient of~$V^k$. Event $E_{I^\star, A, B}$ holds if and only if the~$d$ linear combinations of the variables~$X_i$ given by the last~$d$ columns of~$Q$ fall into a $d$-dimensional hypercube~$C$ depending on the linear combinations of the variables~$X_i$ given by the remaining columns. The claim follows by applying Theorem~\ref{theorem.Prob:enough randomness} for matrix $A = Q'^\T$ and due to the fact that $|I^\star| \leq \gamma$ (see proof of Lemma~\ref{lemma.ZP:certificate space}).
\end{proof}

\begin{proof}[Proof of Theorem~\ref{thm:MainZeroPreserving}]
We begin the proof by showing that the $\OKZ$-event is likely to happen.
For all indices $t \in [d]$ and all solutions $x, y \in \S$ for which $x|_{P_t} \neq y|_{P_t}$ the probability
that $\big| V^t x - V^t y \big| \leq \e$ is bounded by $2\phi\e$. To see this, choose one index $i \in P_t$ for which $x_i \neq y_i$ and apply the principle of deferred decisions by fixing all coefficients~$V^t_j$ for $j \neq i$ arbitrarily. Then the value~$V^t_i$ must fall into an interval of length $2\e/|x_i-y_i| \leq 2\e$. The probability for this is bounded by $2\phi\e$. A union bound over all indices $t \in [d]$ and over all pairs $(x, y) \in \S \times \S$ for which $x|_{P_t} \neq y|_{P_t}$ yields
\[
	\Pr[V]{\overline{\OKZ(V)}} \leq 2(\K+1)^{2n}d\phi\e \DOT
\]

Let $\gamma = d^3 + d^2 + d$. We set
\[
	s = \begin{cases}
		(2\gamma\K)^{\gamma-d} \phi^\gamma = \K^{\gamma-d} \cdot O(\phi^\gamma) & \text{for general density functions} \COMMA \cr
		2^d (\gamma \K)^{\gamma-d} \phi^d = \K^{\gamma-d} \cdot O(\phi^d)       &  \text{for quasiconcave density functions} \DOT
	\end{cases}
\]
Then we obtain
\begin{align*}
  \Ex[V]{\PO(V)}
  &\leq \sum_{(I^\star, A) \in \CS} \sum_{B \in \B_\e} \Pr[V]{E_{I^\star, A, B}} + (\K+1)^n \cdot \Pr[V]{\overline{\OKZ(V)}} \cr
  &\leq \sum_{(I^\star, A) \in \CS} \sum_{B \in \B_\e} s \cdot \e^d + (\K+1)^n \cdot 2(\K+1)^{2n}d\phi\e \cr
  &= |\CS| \cdot |\B_\e| \cdot s \cdot \e^d + (\K+1)^n \cdot 2(\K+1)^{2n}d\phi\e \cr
  &= (\K+1)^{(d^2+d)(d^3+d^2+d)} \cdot O \big( n^{d^3+d^2} \big) \cdot \left( \frac{2n\K}{\e} \right)^d \cdot s \cdot \e^d + 2(\K+1)^{3n}d\phi\e \cr
  &= \K^{(d^2+d)(d^3+d^2+d)+d} \cdot O \big( n^{d^3+d^2+d} \big) \cdot s + 2(\K+1)^{3n}d\phi\e \DOT
\end{align*}
The first inequality is due to Corollary~\ref{corollary.ZP:expectation bound}.
The second inequality is due to Corollary~\ref{corollary.ZP:success unlikely}. The third inequality stems from Lemma~\ref{lemma.ZP:certificate space}. Since this bound holds for any $\e > 0$ for which $1/\e$ is integral, it also holds for the limit $\e \to 0$. Hence, we obtain
\[
  \Ex[V]{\PO(V)}
  = \K^{(d^2+d)(d^3+d^2+d)+d} \cdot O \big( n^{d^3+d^2+d} \big) \cdot s \DOT
\]
Substituting~$s$ and~$\gamma$ by their definitions yields
\begin{align*}
  \Ex[V]{\PO(V)}
  &= \K^{(d^2+d)(d^3+d^2+d)+d} \cdot O \big( n^{d^3+d^2+d} \big) \cdot \K^{d^3+d^2+d-d} \cdot O \big( \phi^{d^3+d^2+d} \big) \cr
  &= \K^{(d^2+d+1)(d^3+d^2+d)} \cdot O \big( (n\phi)^{d^3+d^2+d} \big) \cr
  &\leq \K^{(d+1)^5} \cdot O \big( (n\phi)^{d^3+d^2+d} \big)
\end{align*}
for general densities and
\begin{align*}
  \Ex[V]{\PO(V)}
  &= \K^{(d^2+d)(d^3+d^2+d)+d} \cdot O \big( n^{d^3+d^2+d} \big) \cdot \K^{\gamma-d} \cdot O \big( \phi^d \big) \cr
  &= \K^{(d^2+d+1)(d^3+d^2+d)} \cdot O \big( n^{d^3+d^2+d} \phi^d \big) \cr
  &\leq \K^{(d+1)^5} \cdot O \big( n^{d^3+d^2+d} \phi^d \big)
\end{align*}
for quasiconcave densities.
\end{proof}

\pagebreak

\section{Some Probability Theory}

In this chapter we lay the probabilistic foundation of this article. We consider linearly independent linear combinations of independent random variables and show that they behave to some extent like independent random variables.

Let $X_1, \ldots, X_n$ be independent random variables with densities $f_i \colon [-1, 1] \to [0, \phi]$ for all $i \in [n]$ and let $A \in \SET{ -\K, \ldots, \K }^{m \times n}$ be an integer matrix. Furthermore, let $(Y_1, \ldots, Y_{m-k}, Z_1, \ldots, Z_k)^\T = A \cdot (X_1, \ldots, X_n)^\T$ be an $m$-dimensional random vector whose entries are linear combinations of the random variables $X_1, \ldots, X_n$, and let~$C$ be an arbitrary function that maps every tuple $(y_1, \ldots, y_{m-k}) \in \RR^{m-k}$ to a $k$-dimensional hypercube $C(y_1, \ldots, y_{m-k}) \subseteq \RR^k$ with side length~$\e$. We want to bound the probability that the random vector $(Z_1, \ldots, Z_k)^\T$ falls into the random hypercube $C(Y_1, \ldots, Y_{m-k})$ from above.

Before we state the main theorem of this section, let us discuss two special cases. One simple case is when the matrix~$A$ is of the form $A = [\ID{m}, \ZERO{m}{n-m}]$. In this case, the random variables $Y_i = X_i$, $i = 1, \ldots, m-k$, and $Z_j = X_{m-k+j}$, $j = 1, \ldots, k$, are independent. In order to bound the probability $\Pr{(Z_1, \ldots, Z_k) \in C(Y_1, \ldots, Y_{m-k})}$, we can apply the principle of deferred decisions and assume that the outcome of the variables $Y_1, \ldots, Y_{m-k}$ has been revealed by an adversary, say $Y_i = y_i$ for $i = 1, \ldots, m-k$. Hence, the hypercube $C(Y_1, \ldots, Y_{m-k}) = C(y_1, \ldots, y_{m-k})$ is fixed and not random anymore. However, we still have not revealed the outcome of the random variables $Z_1, \ldots, Z_k$. As the random variables $Y_1, \ldots, Y_{m-k}, Z_1, \ldots, Z_k$ are independent, we obtain
\begin{align*}
	&\Pr{(Z_1, \ldots, Z_k) \in C(Y_1, \ldots, Y_{m-k}) \,|\, (Y_1, \ldots, Y_{m-k}) = (y_1, \ldots, y_{m-k})} \cr
	&= \Pr{(Z_1, \ldots, Z_k) \in C(y_1, \ldots, y_{m-k})} \cr
	&= \Pr{(X_{m-k+1}, \ldots, X_m) \in C(y_1, \ldots, y_{m-k})} \cr
	&\leq (\phi \e)^k \DOT
\end{align*}
Observe that we used the simple structure of~$A$ twice: First, we obtained independence which is why the first equation holds. Second, the event $(Z_1, \ldots, Z_k) \in C(y_1, \ldots, y_{m-k})$ can be directly translated into the event $(X_{m-k+1}, \ldots, X_m) \in C(y_1, \ldots, y_{m-k})$ that only depends on the~$k$ random variables $X_{m-k+1}, \ldots, X_m$ and the hypercube $C(y_1, \ldots, y_{m-k})$. In general, all random variables $X_1, \ldots, X_n$ and a more complex set $\hat{C}(y_1, \ldots, y_{m-k})$ that depends on the last~$k$ rows of~$A$ have to be considered.

Now let us consider a second special case in which each of the last~$k$ rows $a_{m-k+1}^\T, \ldots, a_m^\T$ of~$A$ is a linear combination of the first $m-k$ rows $a_1^\T, \ldots, a_{m-k}^\T$. In particular, for every index $i = 1, \ldots, k$ we can write $a_{m-k+i}$ as
\[
	a_{m-k+i} = \sum_{j=1}^{m-k} \lambda^{(i)}_j \cdot a_j
\]
for appropriate coefficients $\lambda^{(i)}_j$. As the function~$C$ we consider the function that maps a tuple $(y_1, \ldots, y_{m-k})$ to the hypercube $[b_1, b_1+\e] \times \ldots \times [b_k, b_k+\e]$, where~$b_i$ is defined as
\[
	b_i = \sum_{j=1}^{m-k} \lambda^{(i)}_j \cdot y_j \DOT
\]
With this choice and the notation $X = (X_1, \ldots, X_n)^\T$ we obtain
\begin{align*}
	Z_i
	&= a_{m-k+i}^\T X
	= \sum_{j=1}^{m-k} \lambda^{(i)}_j \cdot a_j^\T X
	= \sum_{j=1}^{m-k} \lambda^{(i)}_j \cdot Y_j
\end{align*}
for all $i = 1, \ldots, k$. Hence, $(Z_1, \ldots, Z_k)$ falls into the hypercube $C(Y_1, \ldots, Y_{m-k})$ for every realization of the random variables $X_1, \ldots, X_n$. Consequently, $\Pr{(Z_1, \ldots, Z_k) \in C(Y_1, \ldots, Y_{m-k})} = 1$. The reason why we can define a function~$C$ with such a property is that the outcome of the random variables $Z_1, \ldots, Z_k$ is determined when the outcome of the random variables $Y_1, \ldots, Y_{m-k}$ has been revealed since the last~$k$ rows of~$A$ can be expressed as linear combinations of the first $m-k$ rows of~$A$. Hence, in this special case we cannot obtain any non-trivial bound.

The following theorem claims a non-trivial bound for the probability that the random vector $(Z_1, \ldots, Z_k)$ falls into the hypercube $C(Y_1, \ldots, Y_{m-k})$ for the case when the rows of matrix~$A$ are linearly independent. The first inequality of Theorem~\ref{theorem.Prob:enough randomness} has been shown for binary matrices by R\"oglin and Teng~\cite{RoeglinT09} (see Lemma~3.3). Their proof can be easily generalized to arbitrary integer matrices. For the sake of completeness we state it here.

\begin{theorem}
\label{theorem.Prob:enough randomness}
Let $m \leq n$ be integers and let $X_1, \ldots, X_n$ be independent random variables, each with a probability density function $f_i \colon [-1, 1] \to [0,\phi]$, let $A \in \SET{ -\K, \ldots, \K }^{m \times n}$ be a matrix of rank~$m$, let $k \in [m-1]$ be an integer, let
\[
	(Y_1, \ldots, Y_{m-k}, Z_1, \ldots, Z_k)^\T
	= A \cdot (X_1, \ldots, X_n)^\T
\]
be the linear combinations of $X_1, \ldots, X_n$ given by~$A$, and let~$C$ be a function mapping a tuple $(y_1, \ldots, y_{m-k}) \in \RR^{m-k}$ to a hypercube $C(y_1, \ldots, y_{m-k}) \subseteq \RR^k$ with side length~$\e$. Then
\[
  \Pr{(Z_1, \ldots, Z_k) \in C(Y_1, \ldots, Y_{m-k})}
  \leq (2m\K)^{m-k} \phi^m \e^k \DOT
\]
If all densities~$f_i$ are quasiconcave, then even the stronger bound
\[
  \Pr{(Z_1, \ldots, Z_k) \in C(Y_1, \ldots, Y_{m-k})}
  \leq 2^k (m\K)^{m-k} \phi^k \e^k
\]
holds.
\end{theorem}

Theorem~\ref{theorem.Prob:enough randomness} states that, for quasiconcave densities, linearly independent linear combinations of independent random variables almost behave like independent random variables when considering the event $(Z_1, \ldots, Z_k) \in C(Y_1, \ldots, Y_{m-k})$ with respect to~$\phi$ and~$\e$: The bound in Theorem~\ref{theorem.Prob:enough randomness} deviates from the bound derived for the special case $Y_i = X_i$ for $i = 1, \ldots, m-k$ and $Z_j = X_{m-k+j}$ for $j = 1, \ldots, k$ (see beginning of this section) only by a factor of $2^k (m\K)^{m-k}$.

\begin{proof}
First of all we show that we can assume w.l.o.g.\ that $n = m$. Otherwise, we can choose~$m$ indices $i_1 < \ldots < i_m \in [n]$ for which the matrix $A' = [a_{i_1}, \ldots, a_{i_m}]$ is a full-rank square submatrix of~$A$. For the sake of simplicity let us assume that $i_k = k$ for $k = 1, \ldots, m$. We apply the principle of deferred decisions and assume that $X_{m+1}, \ldots, X_n$ are fixed arbitrarily to some values $x_{m+1}, \ldots, x_n$.

Let $A'' = [a_{m+1}, \ldots, a_n]$, $A''_1 = A''|_{1, \ldots, m-k}$, $A''_2 = A''|_{m-k+1, \ldots, m}$, and $x = (x_{m+1}, \ldots, x_n)$. Let us further introduce the random vector
\[
	(Y'_1, \ldots, Y'_{m-k}, Z'_1, \ldots, Z'_k)
	= A' \cdot (X_1, \ldots, X_m)
\]
and the function
\[
	C'(Y'_1, \ldots, Y'_{m-k})
	= C((Y'_1, \ldots, Y'_{m-k}) + A''_1 \cdot x) - A''_2 \cdot x \DOT
\]
Observing that
\begin{align*}
	(Y'_1, \ldots, Y'_{m-k}) + A''_1 \cdot x
	&= (Y_1, \ldots, Y_{m-k}) \quad \text{and} \cr
	(Z'_1, \ldots, Z'_k) + A''_2 \cdot x
	&= (Z_1, \ldots, Z_k) \COMMA
\end{align*}
we obtain
\begin{align*}
  (Z_1, \ldots, Z_k) \in C(Y_1, \ldots, Y_{m-k})
  &\iff (Z'_1, \ldots, Z'_k) \in C(Y_1, \ldots, Y_{m-k}) - A''_2 \cdot x \cr
  &\iff (Z'_1, \ldots, Z'_k) \in C'(Y'_1, \ldots, Y'_{m-k}) \DOT
\end{align*}
The probability of the last event can be bounded by applying Theorem~\ref{theorem.Prob:enough randomness} for the $m \times m$-matrix~$A'$.

In the remainder of this proof we assume that $n = m$. As matrix~$A$ is a full-rank square matrix, its inverse~$A^{-1}$ exists and we can write
\begin{align*}
  \Pr{(Z_1, \ldots, Z_k) \in C(Y_1, \ldots, Y_{n-k})}
  &= \int_{y \in \RR^{n-k}} \int_{z \in C(y)} f_{Y,Z}(y, z) \;\! \d z \;\! \d y \cr
  &= \int_{y \in \RR^{n-k}} \int_{z \in C(y)} |\det(A^{-1})| \cdot f_X \big( A^{-1} \cdot (y, z) \big) \d z \;\! \d y \cr
  &\leq \int_{y \in \RR^{n-k}} \int_{z \in C(y)} f_X \big( A^{-1} \cdot (y, z) \big) \d z \;\! \d y \cr
  &\leq \e^k \cdot \int_{y \in \RR^{n-k}} \max_{z \in \RR^k} f_X \big( A^{-1} \cdot (y, z) \big) \d y \COMMA
\end{align*}
where $f_{Y, Z}$ denotes the common density of the variables $Y_1, \ldots, Y_{n-k}, Z_1, \ldots, Z_k$ and $f_X = \prod_{i=1}^n f_i$ denotes the common density of the variables $X_1, \ldots, X_n$. The second equality is due to a change of variables, the first inequality stems from the fact that $|\det(A^{-1})| = |1/\det A| \leq 1$ since~$A$ is an integer matrix.

In general, we can bound the integral in the formula above by
\begin{align*}
  \int_{y \in \RR^{n-k}} \max_{z \in \RR^k} f_X \big( A^{-1} \cdot (y, z) \big) \d y
  &\leq \int_{y \in [-n\K, n\K]^{n-k}} \max_{z \in \RR^k} f_X \big( A^{-1} \cdot (y, z) \big) \d y \cr
  &\leq \int_{y \in [-n\K, n\K]^{n-k}} \phi^n \d y \cr
  &= (2n\K)^{n-k} \phi^n \cr
  &= (2m\K)^{m-k} \phi^m \COMMA
\end{align*}
where the first inequality is due to the fact that all variables~$Y_i$ can only take values in the interval $[-n\K, n\K]$ as all entries of matrix~$A$ are from $\SET{ -\K, \ldots, \K }$ and as all variables~$X_j$ can only take values in the interval $[-1,1]$.

To prove the statement about quasiconcave functions we first consider arbitrary rectangular functions, i.e., functions that are
constant on a given interval, and~$0$ otherwise. This will be the main part of our analysis. Afterwards, we analyze sums of rectangular
functions and, finally, we show that quasiconcave functions can be approximated by such sums.

\begin{lemma}
\label{lemma.Prob:rectangular functions}
For $i \in [n]$ let $\phi_i \geq 0$, let $I_i \subseteq \RR$ be an interval of length~$\ell_i$, and let $f_i \colon \RR \to \RR$ be the function
\[
  f_i(x) = \begin{cases}
    \phi_i & \tIF x \in I_i,\cr
    0      & \OTHERWISE.
  \end{cases}
\]
Moreover, let $f \colon \RR^n \to \RR$ be the function $f(x_1, \ldots, x_n) = \prod_{i=1}^n f_i(x_i)$ and let $A \in \SET{ -\K, \ldots, \K }^{n \times n}$ be an invertible matrix. Then
\[
  \int_{y \in \RR^{n-k}} \max_{z \in \RR^k} f \big( A^{-1} \cdot (y, z) \big) \d y
  \leq 2^k \cdot (n-k)! \cdot \K^{n-k} \cdot \chi \cdot \sum_I \prod_{i \notin I} \ell_i
\]
where $\chi = \prod_{i=1}^n \phi_i$ and where the sum runs over all tuples $I = (i_1, \ldots, i_k)$ for which $1 \leq i_1 < \ldots < i_k \leq n$.
\end{lemma}

\begin{proof}
Function~$f$ takes the value~$\chi$ on the $n$-dimensional box $Q = \prod_{i=1}^n I_i$ and is~$0$ otherwise. Hence,
\[
  \int_{y \in \RR^{n-k}} \max_{z \in \RR^k} f \big( A^{-1} \cdot (y, z) \big) \d y
  = \chi \cdot \VOL(Q')
\]
for
\begin{align*}
  Q'
  &= \SET{ y \in \RR^{n-k} \WHERE \exists z \in \RR^k \ \text{such that} \ A^{-1} \cdot (y, z) \in Q } \cr
  &= \SET{ y \in \RR^{n-k} \WHERE \exists z \in \RR^k \exists x \in Q \ \text{such that} \ (y, z) = A \cdot x } \cr
  &= (P \cdot A) (Q) \COMMA
\end{align*}
where $P \DEF \big[ \ID{n-k}, \ZERO{(n-k)}{k} \big]$ is the projection matrix that removes the last~$k$ entries from a vector of
length~$n$. In the remainder of this proof we bound the volume of $M(Q)$ where $M \DEF P \cdot A \in \SET{ -\K, \ldots, \K }^{(n-k) \times n}$. Let $a_i \FED c_i^{0}$ and $b_i \FED c_i^{1}$ be the left and the right bound of interval~$I_i$, respectively. For an index tuple $I = (i_1, \ldots, i_k)$, $1 \leq i_1 < \ldots < i_k \leq n$, and a bit tuple $J = (j_1, \ldots, j_k) \in \SET{ 0, 1 }^k$, let
\[
  F_I^J = \prod_{i=1}^n \begin{cases}
    \bSET{ c_{i_t}^{j_t} } & \tIF i = i_t \in I \COMMA \cr
    I_i                    & \tIF i \notin I \COMMA
  \end{cases}
\]
be one of the $2^k \cdot \binom{n}{k}$ $(n-k)$-dimensional faces of~$Q$. We show that $M(Q) \subseteq \bigcup_I \bigcup_J M \big( F_I^J \big)$. Let $y \in M(Q)$, i.e., there is a vector $x \in Q$ such that $y = M \cdot x$. Now, consider the polytope
\[
  R = \SET{ (x', s') \in \RR^n \times \RR^n \WHERE M \cdot x' = y', \ x' + s' = b', \ \text{and} \ x', s' \geq 0 } \COMMA
\]
where $y' = y - M \cdot a$ and $b' = b - a$ for $a = (a_1, \ldots, a_n)$ and $b = (b_1, \ldots, b_n)$. This polytope is bounded and non-empty because $(x - a, b - x) \in R$. Consequently, there exists a basic feasible solution $(x^\star, s^\star)$. 
As there are~$2n$ variables and $2n-k$ constraints,
this solution has at least~$k$ zero-entries, i.e., there are indices $1 \leq i_1 < \ldots < i_k \leq n$ such that either $x^\star_{i_t} = 0$ (in that case set $j_t = 0$) or $x^\star_{i_t} = b'_{i_t}$ (in that case set $j_t = 1$) for all $t \in [k]$. Now, consider the vector $\hat{x} = x^\star + a \in [0, b'] + a = Q$. We obtain $M \cdot \hat{x} = y$ and $\hat{x}_{i_t} = c_{i_t}^{j_t}$ for all $t \in [k]$. Hence, $x \in F_I^J$ for $I = (i_1, \ldots, i_k)$ and $J = (j_1, \ldots, j_k)$, and thus $y \in M \big( F_I^J \big)$.

Due to this observation we can bound the volume of $M(Q)$ by $\sum_I \sum_J \VOL \big( M \big( F_I^J \big) \big)$. 
It remains to show how to bound the volume $\VOL \big( M \big( F_I^J \big) \big)$. For the sake of simplicity we only consider $I = (n-k+1, \ldots, n)$ in the following analysis. Let $\phi^J \colon \RR^{n-k} \to F_I^J$ be the function $\phi^J(x) = T \cdot x + v^J$, where $T = \big[ \ID{n-k}, \ZERO{(n-k)}{k} \big]^\T$ and $v^J = \big( 0, \ldots, 0, c_{n-k+1}^{j_1}, \ldots, c_n^{j_k} \big)$. Using function~$\phi^J$ is the canonical way to describe the affine subspace defined by face~$F_I^J$: it adds the fixed coordinates of~$F_I^J$ to a given vector of length $n-k$. Hence, function~$\phi^J$, restricted to the domain $F' = \prod_{i=1}^{n-k} I_i$, is bijective. With $\psi = M \circ \phi^J$ we obtain
\begin{align*}
  \VOL \big( M \big( F_I^J \big) \big)
  &= \int_{\psi(F')} 1 \;\! \d x
  = \int_{F'} |\det \D \psi(x)| \;\! \d x
  = \int_{F'} |\det(M \cdot T)| \;\! \d x \cr
  &= |\det(M \cdot T)| \cdot \VOL(F')
  = |\det(M \cdot T)| \cdot \prod_{i=1}^{n-k} \ell_i \DOT
\end{align*}
In general, the second equality only holds if~$\psi$ is injective. If~$\psi$ is not injective, then $M(F_I^J) = \psi(F')$ is not full-dimensional, i.e., $\VOL(M(F_I^J)) = 0$, and $\det(M \cdot T) = 0$ since~$\psi$ is affine linear. Hence, the second equality also holds in the case when~$\psi$ is not injective.

Matrix $M \cdot T = P \cdot A \cdot T$ is an $(n-k) \times (n-k)$-submatrix of~$A$. Thus, $|\det(M \cdot T)| \leq (n-k)! \cdot \K^{n-k}$, and we obtain the bound
\begin{align*}
  \int_{y \in \RR^{n-k}} \max_{z \in \RR^k} f \big( A^{-1} \cdot (y, z) \big) \d y
  &= \chi \cdot \VOL(M(Q))
  \leq \chi \cdot \sum_I \sum_J \VOL \big( M\big( F_I^J \big) \big) \cr
  &\leq \chi \cdot \sum_I \sum_J (n-k)! \cdot \K^{n-k} \cdot \prod_{i \notin I} \ell_i \cr
  &= 2^k \cdot (n-k)! \cdot \K^{n-k} \cdot \chi \cdot \sum_I \prod_{i \notin I} \ell_i \DOT \qedhere
\end{align*}
\end{proof}

In the next step we generalize the statement of Lemma~\ref{lemma.Prob:rectangular functions} to sums of rectangular functions.

\begin{corollary}
\label{corollary.Prob:staircases}
Let $N_1, \ldots, N_n$ be positive integers, let $\phi_{i,k} \geq 0$ be a non-negative real, let $I_{i,k} \subseteq \RR$ be an interval of length~$\ell_{i,k}$, and let $f_{i,k} \colon \RR \to \RR$ be the function
\[
  f_{i,k}(x) = \begin{cases}
    \phi_{i,k} & \tIF x \in I_{i,k} \COMMA \cr
    0          & \OTHERWISE \COMMA
  \end{cases}
\]
$i \in [n]$, $k \in [N_i]$. Furthermore, let $f_i \colon \RR \to \RR$ be the function $f_i = \sum_{k=1}^{N_i} f_{i,k}$, let $f \colon \RR^n \to \RR$ be the function $f(x_1, \ldots, x_n) = \prod_{i=1}^n f_i(x_i)$, and let $A \in \SET{ -\K, \ldots, \K }^{n \times n}$ be an invertible matrix. Then
\[
  \int_{y \in \RR^{n-k}} \max_{z \in \RR^k} f \big( A^{-1} \cdot (y, z) \big) \d y
  \leq 2^k \cdot (n-k)! \cdot \K^{n-k} \cdot \sum_I \left( \left( \prod_{i \notin I} \sigma_i \right) \cdot \left( \prod_{i \in I} \chi_i \right) \right)
\]
where $\sigma_i = \sum_{k=1}^{N_i} \phi_{i,k} \cdot \ell_{i,k}$ and $\chi_i = \sum_{k=1}^{N_i} \phi_{i,k}$ and where the first sum runs over all tuples $I = (i_1, \ldots, i_k)$ for which $1 \leq i_1 < \ldots < i_k \leq n$.
\end{corollary}

\begin{proof}
For indices $k_i \in [N_i]$ let $f_{k_1, \ldots, k_n}(x_1, \ldots, x_n) = \prod_{i=1}^n f_{i,k_i}(x_i)$. This function is of the form assumed in Lemma~\ref{lemma.Prob:rectangular functions} and takes only values~$0$ and $\chi_{k_1, \ldots, k_n} = \prod_{i=1}^n \phi_{i,k_i}$. We can write function~$f$ as
\begin{align*}
  f(x_1, \ldots, x_n)
  &= \prod_{i=1}^n f_i(x_i)
  = \prod_{i=1}^n \sum_{k_i=1}^{N_i} f_{i,k_i}(x_i)
  = \sum_{k_1=1}^{N_1} \ldots \sum_{k_n=1}^{N_n} \prod_{i=1}^n f_{i,k_i}(x_i) \cr
  &= \sum_{k_1=1}^{N_1} \ldots \sum_{k_n=1}^{N_n} f_{k_1, \ldots, k_n}(x_1, \ldots, x_n) \DOT
\end{align*}
For the sake of simplicity we write $\sum_{k_i}$ instead of $\sum_{k_i=1}^{N_i}$ and $\sum_{k_i \colon i \in (i_1, \ldots, i_\ell)}$ instead of $\sum_{k_{i_1}} \ldots \sum_{k_{i_\ell}}$. We can bound the integral as follows:
\begin{align*}
  \int_{y \in \RR^{n-k}} &\max_{z \in \RR^k} f \big( A^{-1} \cdot (y, z) \big) \d y
  = \int_{y \in \RR^{n-k}} \max_{z \in \RR^k} \sum_{k_i \colon i \in [n]} f_{k_1, \ldots, k_n} \big( A^{-1} \cdot (y, z) \big) \d y \cr
  &\leq \int_{y \in \RR^{n-k}} \sum_{k_i \colon i \in [n]} \max_{z \in \RR^k} f_{k_1, \ldots, k_n} \big( A^{-1} \cdot (y, z) \big) \d y \cr
  &= \sum_{k_i \colon i \in [n]} \int_{y \in \RR^{n-k}} \max_{z \in \RR^k} f_{k_1, \ldots, k_n} \big( A^{-1} \cdot (y, z) \big) \d y \cr
  &\leq \sum_{k_i \colon i \in [n]} \left( 2^k \cdot (n-k)! \cdot \K^{n-k} \cdot \chi_{k_1, \ldots, k_n} \cdot \sum_I \prod_{i \notin I} \ell_{i,k_i} \right) \cr
  &= 2^k \cdot (n-k)! \cdot \K^{n-k} \cdot \sum_{k_i \colon i \in [n]} \left( \prod_{i \in [n]} \phi_{i,k_i} \cdot \sum_I \prod_{i \notin I} \ell_{i,k_i} \right) \cr
  &= 2^k \cdot (n-k)! \cdot \K^{n-k} \cdot \sum_I \sum_{k_i \colon i \in [n]} \left( \prod_{i \in [n]} \phi_{i,k_i} \cdot \prod_{i \notin I} \ell_{i,k_i} \right) \COMMA
\end{align*}
where the second inequality is due to Lemma~\ref{lemma.Prob:rectangular functions}. Now,
\begin{align*}
  \sum_{k_i \colon i \in [n]} \left( \prod_{i \in [n]} \phi_{i,k_i} \cdot \prod_{i \notin I} \ell_{i,k_i} \right)
  &= \sum_{k_i \colon i \in [n]} \left( \prod_{i \in I} \phi_{i,k_i} \cdot \prod_{i \notin I} (\phi_{i,k_i} \cdot \ell_{i,k_i}) \right) \cr
  &= \left( \sum_{k_i \colon i \in I} \prod_{i \in I} \phi_{i,k_i} \right) \cdot \left( \sum_{k_i \colon i \notin I} \prod_{i \notin I} (\phi_{i,k_i} \cdot \ell_{i,k_i}) \right) \cr
  &= \left( \prod_{i \in I} \sum_{k_i} \phi_{i,k_i} \right) \cdot \left( \prod_{i \notin I} \sum_{k_i} (\phi_{i,k_i} \cdot \ell_{i,k_i}) \right) \cr
  &= \left( \prod_{i \in I} \chi_i \right) \cdot \left( \prod_{i \notin I} \sigma_i \right) \COMMA
\end{align*}
which completes the proof of Corollary~\ref{corollary.Prob:staircases}.
\end{proof}

To finish the proof of Theorem~\ref{theorem.Prob:enough randomness} we round the probability densities~$f_i$ as follows: For an arbitrarily small positive real~$\delta$ let $g_i \DEF \ceil{ f_i/\delta } \cdot \delta$, i.e., we round~$f_i$ up to the next integral multiple of~$\delta$. As the densities~$f_i$ are quasiconcave, there is a decomposition of~$g_i$ such that $g_i = \sum_{k=1}^{N_i} f_{i,k}$ where
\[
  f_{i,k} = \left\{ \begin{array}{c@{\quad:\quad}l}
    \phi_{i,k} & x \in I_{i,k} \COMMA \cr
    0 & \OTHERWISE \COMMA
  \end{array} \right.
  \qquad \text{and} \qquad
  \chi_i \DEF \sum_{k=1}^{N_i} \phi_{i,k} = \max_{x \in [-1,1]} g_i(x) \COMMA
\]
where~$I_{i,k}$ are intervals of length~$\ell_{i,k}$ and~$\phi_{i,k}$ are positive reals. The second property is the interesting one and stems from the quasiconcaveness of~$f_i$. Informally speaking the two-dimensional shape bounded by
the horizontal axis and the graph of~$g_i$ is a stack of rectangles aligned with axes (see Figure~\ref{fig:quasiconcave}). Therefore, the sum~$\chi_i$ of the rectangles' heights which appears in the formula of Corollary~\ref{corollary.Prob:staircases} is approximately~$\phi$. Without the quasiconcaveness~$\chi_i$ might be unbounded.

\begin{figure*}
  \begin{center}
    \includegraphics[width=0.4\textwidth]{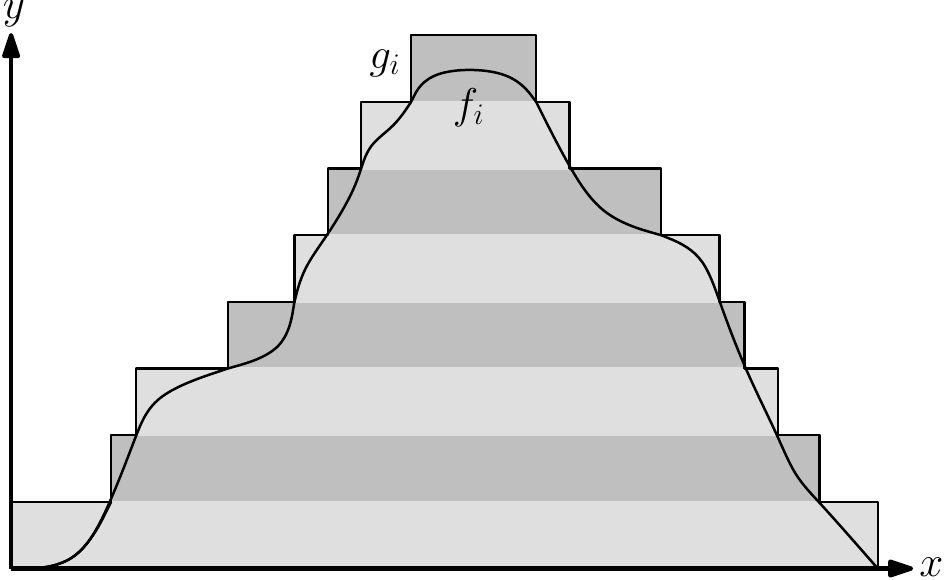}
  \end{center}
  \caption{Area of a quasi-concave function covered by a ``stack'' of rectangles with approximately the same area}
  \label{fig:quasiconcave}
\end{figure*}

Applying Corollary~\ref{corollary.Prob:staircases}, we obtain
\newcommand{\mult}{\!\cdot\!}
\begin{align*}
  \int_{y \in \RR^{n-k}} \max_{z \in \RR^k} f_X \big( A^{-1} \mult (y, z) \big) \d y
  &= \int_{y \in \RR^{n-k}} \max_{z \in \RR^k} \prod_{i=1}^n f_i \big( \big( A^{-1} \mult (y, z) \big)_i \big) \d y \cr
  &\leq \int_{y \in \RR^{n-k}} \max_{z \in \RR^k} \prod_{i=1}^n g_i \big( \big( A^{-1} \mult (y, z) \big)_i \big) \d y \cr
  &\leq 2^k \mult (n-k)! \mult \K^{n-k} \mult \sum_{I} \left( \prod_{i \notin I} \sum_{k_i=1}^{N_i} (\phi_{i,k_i} \mult \ell_{i,k_i}) \right) \mult \left( \prod_{i \in I} \chi_i \right) \cr
  &= 2^k \mult (n-k)! \mult \K^{n-k} \mult \sum_{I} \left( \prod_{i \notin I} \int_{[-1,1]} g_i \;\! \d x \right) \mult \left( \prod_{i \in I} \chi_i \right) \DOT
\end{align*}
Since $0 \leq \int_{[-1,1]} g_i \;\! \d x \leq \int_{[-1,1]} (f_i + \delta) \;\! \d x = 1 + 2\delta$ and $0 \leq \chi_i \leq \sup_{x \in [-1, 1]} f_i(x) + \delta \leq \phi + \delta$, this implies
\begin{align*}
  \int_{y \in \RR^{n-k}} \max_{z \in \RR^k} f_X \big( A^{-1} \cdot (y, z) \big) \d y
  &\leq 2^k \cdot (n-k)! \cdot \K^{n-k} \cdot \sum_{I} \left( \prod_{i \notin I} (1 + 2\delta) \right) \cdot \left( \prod_{i \in I} (\phi + \delta) \right) \cr
  &= 2^k \cdot (n-k)! \cdot \K^{n-k} \cdot \sum_{I} (1 + 2\delta)^{n-k} \cdot (\phi + \delta)^k \cr
  &= 2^k \cdot (n-k)! \cdot \K^{n-k} \cdot \binom{n}{k} \cdot (1 + 2\delta)^{n-k} \cdot (\phi + \delta)^k \cr
  &\leq 2^k \cdot n^{n-k} \cdot \K^{n-k} \cdot (1 + 2\delta)^{n-k} \cdot (\phi + \delta)^k \DOT
\end{align*}
As this bound is true for arbitrarily small reals $\delta > 0$, we obtain the desired bound of
\[
  2^k (n \K)^{n-k} \phi^k
  = 2^k (m \K)^{m-k} \phi^k \DOT \qedhere
\]
\end{proof}

\section{Conclusions and Open Problems}

With the techniques developed in this article we settled two questions posed by
Moitra and O'Donnell~\cite{MoitraO11}: For quasiconcave densities we showed that
the exponent of~$\phi$ in the bound for the smoothed number of Pareto-optimal
solutions is exactly~$d$. Moreover, we significantly improved on the previously
best known bound for higher moments of the smoothed number of Pareto-optima
by R\"{o}glin and Teng~\cite{RoeglinT09}.

Maybe even more interesting are our results for the model of zero-preserving
perturbations suggested by Spielman and Teng~\cite{SpielmanT04} and Beier and
V\"ocking~\cite{BeierV06}. For this model we proved the first non-trivial bound
on the smoothed number of Pareto-optimal solutions. We showed that this result
can be used to analyze multiobjective optimization problems with polynomial
and even more general objective functions. Furthermore, our result implies that
the smoothed running time of the algorithm proposed by Berger et al.~\cite{BergerRZ11}
to compute a path trade in a routing network is polynomially bounded for every constant number
of autonomous systems. We believe that there are many more such applications of
our result in the area of multiobjective optimization.

There are several interesting open questions. First of all it would be interesting
to find asymptotically tight bounds for the smoothed number of Pareto-optimal solutions.
There is still a gap between our upper bound of $O(n^{2d} \phi^d)$ 
for quasiconcave $\phi$-smooth instances and the best lower bound of $\Omega(n^{d-1.5}\phi^d)$~\cite{BrunschGRR14}.
Only for the case $d = 1$ we can show that the upper bound is tight~\cite{BrunschGRR14}. 

Especially for zero-preserving perturbations there is still a lot of work to do.
We conjecture that our techniques can be extended to also bound higher moments 
of the smoothed number of Pareto-optima for $\phi$-smooth instances with
zero-preserving perturbations. However, we feel that even our bound for the first moment is too
pessimistic as we do not have a lower bound showing that  
setting coefficients to~$0$ can lead to larger Pareto sets.
It would be very interesting to either prove a lower bound that shows that zero-preserving
perturbations can lead to larger Pareto-sets than non-zero-preserving perturbations or
to prove a better upper bound for zero-preserving perturbations.

\pagebreak

\bibliographystyle{plain}
\bibliography{literature}

\begin{thebibliography}{10}

\bibitem{BeierPhD}
Ren{\'e} Beier.
\newblock {\em Probabilistic Analysis of Discrete Optimization Problems}.
\newblock PhD thesis, Universit{\"a}t des Saarlandes, 2004.

\bibitem{BeierRV07}
Ren{\'e} Beier, Heiko R{\"o}glin, and Berthold V{\"o}cking.
\newblock The smoothed number of {P}areto optimal solutions in bicriteria
  integer optimization.
\newblock In {\em Proceedings of the 12th International Conference on Integer
  Programming and Combinatorial Optimization (IPCO)}, pages 53--67, 2007.

\bibitem{BeierV04}
Ren{\'e} Beier and Berthold V{\"o}cking.
\newblock Random knapsack in expected polynomial time.
\newblock {\em Journal of Computer and System Sciences}, 69(3):306--329, 2004.

\bibitem{BeierV06}
Ren{\'e} Beier and Berthold V{\"o}cking.
\newblock Typical properties of winners and losers in discrete optimization.
\newblock {\em SIAM Journal on Computing}, 35(4):855--881, 2006.

\bibitem{BergerRZ11}
Andr{\'e} Berger, Heiko R{\"o}glin, and Ruben van~der Zwaan.
\newblock Path trading: Fast algorithms, smoothed analysis, and hardness
  results.
\newblock In {\em Proceedings of the 10th International Symposium on
  Experimental Algorithms (SEA)}, pages 43--53, 2011.

\bibitem{BrunschGRR14}
Tobias Brunsch, Navin Goyal, Luis Rademacher, and Heiko R{\"o}glin.
\newblock Lower bounds for the average and smoothed number of pareto-optima.
\newblock {\em Theory of Computing}, 2014.
\newblock to appear.

\bibitem{BrunschR12}
Tobias Brunsch and Heiko R{\"o}glin.
\newblock Improved smoothed analysis of multiobjective optimization.
\newblock In {\em Proceedings of the 44th Annual ACM Symposium on Theory of
  Computing (STOC)}, pages 407--426, 2012.

\bibitem{CorleyM85}
H.~William Corley and I.~Douglas Moon.
\newblock Shortest paths in networks with vector weights.
\newblock {\em Journal of Optimization Theory and Application}, 46(1):79--86,
  1985.

\bibitem{Ehrgott99}
Matthias Ehrgott.
\newblock Integer solutions of multicriteria network flow problems.
\newblock {\em Investigacao Operacional}, 19:229--243, 1999.

\bibitem{Ehrgott05}
Matthias Ehrgott.
\newblock {\em Multicriteria Optimization}.
\newblock Springer, 2005.

\bibitem{EhrgottG02}
Matthias Ehrgott and Xavier Gandibleux.
\newblock Multiobjective combinatorial optimization.
\newblock In Matthias Ehrgott and Xavier Gandibleux, editors, {\em Multiple
  Criteria Optimization -- State of the Art Annotated Bibliographic Surveys},
  pages 369--444. Kluwer Academic Publishers, 2002.

\bibitem{Hansen80}
Pierre Hansen.
\newblock Bicriterion path problems.
\newblock In {\em Multiple Criteria Decision Making: Theory and Applications},
  volume 177 of {\em Lecture Notes in Economics and Mathematical Systems},
  pages 109--127, 1980.

\bibitem{KlamrothW00}
Kathrin Klamroth and Margaret~M. Wiecek.
\newblock Dynamic programming approaches to the multiple criteria knapsack
  problem.
\newblock {\em Naval Research Logistics}, 47(1):57--76, 2000.

\bibitem{MoitraO11}
Ankur Moitra and Ryan O'Donnell.
\newblock Pareto optimal solutions for smoothed analysts.
\newblock {\em SIAM Journal on Computing}, 41(5):1266--1284, 2012.

\bibitem{Mueller-HannemannW01}
Matthias M{\"u}ller-Hannemann and Karsten Weihe.
\newblock Pareto shortest paths is often feasible in practice.
\newblock In {\em Proceedings of the 5th International Workshop on Algorithm
  Engineering (WAE)}, pages 185--198, 2001.

\bibitem{MustafaG98}
Adli Mustafa and Mark Goh.
\newblock Finding integer efficient solutions for bicriteria and tricriteria
  network flow problems using dinas.
\newblock {\em Computers {\&} Operations Research}, 25(2):139--157, 1998.

\bibitem{NemhauserU69}
George~L. Nemhauser and Zev Ullmann.
\newblock Discrete dynamic programming and capital allocation.
\newblock {\em Management Science}, 15(9):494--505, 1969.

\bibitem{RoeglinT09}
Heiko R{\"o}glin and Shang-Hua Teng.
\newblock Smoothed analysis of multiobjective optimization.
\newblock In {\em Proceedings of the 50th Annual IEEE Symposium on Foundations
  of Computer Science (FOCS)}, pages 681--690, 2009.

\bibitem{SankarST06}
Arvind Sankar, Daniel~A. Spielman, and Shang-Hua Teng.
\newblock Smoothed analysis of the condition numbers and growth factors of
  matrices.
\newblock {\em SIAM Journal on Matrix Analysis Applications}, 28(2):446--476,
  2006.

\bibitem{SkriverA00}
Anders J.~V. Skriver and Kim~Allan Andersen.
\newblock A label correcting approach for solving bicriterion shortest-path
  problems.
\newblock {\em Computers {\&} Operations Research}, 27(6):507--524, 2000.

\bibitem{SpielmanT04}
Daniel~A. Spielman and Shang-Hua Teng.
\newblock Smoothed analysis of algorithms: Why the simplex algorithm usually
  takes polynomial time.
\newblock {\em Journal of the ACM}, 51(3):385--463, 2004.

\bibitem{SpielmanT09}
Daniel~A. Spielman and Shang-Hua Teng.
\newblock Smoothed analysis: an attempt to explain the behavior of algorithms
  in practice.
\newblock {\em Communications of the ACM}, 52(10):76--84, 2009.

\end{thebibliography}

\end{document}